\documentclass[acmsmall,screen,authorversion,nonacm]{acmart}
\usepackage[T1]{fontenc}
\usepackage{listings}
\usepackage{algorithm}
\usepackage{algorithmic}
\usepackage{amsmath}
\usepackage{mathtools}
\usepackage{stmaryrd}
\usepackage{wrapfig}
\usepackage{etoolbox}
\usepackage{cleveref}
\usepackage{xcolor}
\usepackage{times}
\usepackage{longtable}
\usepackage{pdflscape}
\usepackage{orcidlink}
\usepackage{bbding}
\usepackage[font=small,justification=raggedright]{caption}
\usepackage{tikz}
\usetikzlibrary{calc}
\usetikzlibrary{backgrounds}
\usetikzlibrary{arrows.meta}
\usetikzlibrary{bending}

\usepackage{pgfplots}
\pgfplotsset{compat=1.18}

\author{Florian Sextl}
\orcid{0009-0003-5839-0726}
\email{florian.sextl@tuwien.ac.at}
\affiliation{
    \institution{TU Wien}
    \department[0]{Institute of Logic and Computation}
    \department[1]{Research Unit for Formal Methods in Systems Engineering}
    \city{Vienna}
    \country{Austria}
}
\authornote{Main author, other authors in alphabetical order of their last names}
\author{Adam Rogalewicz}
\orcid{0000-0002-7911-0549}
\email{rogalew@fit.vut.cz}
\affiliation{
    \institution{Brno University of Technology}
    \department{Faculty of Information Technology}
    \city{Brno}
    \country{Czechia}
}
\author{Tom\'{a}\v{s} Vojnar}
\orcid{0000-0002-2746-8792}
\email{vojnar@fi.muni.cz}
\affiliation{
    \institution{Masaryk University}
    \department{Faculty of Informatics}
    \city{Brno}
    \country{Czechia}
}
\affiliation{
    \institution{Brno University of Technology}
    \department{Faculty of Information Technology}
    \city{Brno}
    \country{Czechia}
}
\author{Florian Zuleger}
\orcid{0000-0003-1468-8398}
\email{florian.zuleger@tuwien.ac.at}
\affiliation{
    \institution{TU Wien}
    \department[0]{Institute of Logic and Computation}
    \department[1]{Research Unit for Formal Methods in Systems Engineering}
    \city{Vienna}
    \country{Austria}
}

\title{Compositional Shape Analysis with Shared Abduction and Biabductive Loop Acceleration (Extended Version)}



\lstdefinestyle{mystyle}{
    commentstyle=\color{ACMPurple},
    keywordstyle=\color{ACMGreen},
    numberstyle=\tiny\color{black},
    stringstyle=\color{ACMPurple},
    basicstyle=\ttfamily\scriptsize,
    breakatwhitespace=false,
    breaklines=true,
    captionpos=b,
    keepspaces=true,
    numbers=none,
    numbersep=5pt,
    showspaces=false,
    showstringspaces=false,
    showtabs=false,
    tabsize=2,
    abovecaptionskip=-1pt,
    numberbychapter=false
}
\newcommand{\Null}{\mathit{NULL}}

\newcommand{\exit}[1]{\textsc{Exit}(#1)}
\newcommand{\mathdef}{\mathrel{\mathop:}=}
\newcommand{\effect}{\tau}

\newcommand{\expast}{\theta}

\newcommand{\soundw}{\triangleright}
\newcommand{\newname}{\text{Brush}}
\newcommand{\conf}{\mathit{Config}}

\newcommand{\AbsState}{Analysis State}

\newcommand{\absstate}{analysis state}
\newcommand{\noshapeeffect}{\mathit{rem}}
\newcommand{\broomv}{a361d01badf45c420b57158f2e6d738cb45d1dd9}
\newcommand{\inferv}{f93cb281edb33510d0a300f1e4c334c6f14d6d26}
\newcommand{\statesep}{\ \|\ }
\newcommand{\formsep}{:}
\newcommand{\oldsh}{\mathcal{P}}
\newcommand{\newsh}{\mathcal{Q}}
\newcommand{\highlight}[2]{\colorbox{#1}{$\displaystyle #2$}}
\newcommand{\ls}{\mathit{ls}}
\newcommand{\cond}{c}
\newcommand{\changed}{\chi}
\newcommand{\transfm}{\mu}
\newcommand{\textttt}[1]{\texttt{\scriptsize#1}}

\makeatletter
\patchcmd\WF@putfigmaybe{\lower\intextsep}{}{}{\fail}%
\AddToHook{env/wrapfigure/begin}{\setlength{\intextsep}{0pt}}
\makeatother

\newcommand{\ana}{A}
\newcommand{\abs}{\alpha}

\makeatletter
\renewcommand{\fnum@figure}{Figure \thefigure}
\makeatother
\floatname{algorithm}{Procedure}
\Crefname{algorithm}{procedure}{procedures}
\Crefname{algorithm}{Procedure}{Procedures}
\Crefname{remark}{remark}{remarks}
\Crefname{remark}{Remark}{Remarks}
\Crefname{figure}{figure}{figures}
\Crefname{figure}{Figure}{Figures}

\begin{abstract}
    Biabduction-based shape analysis is a compositional verification and analysis
    technique that can prove memory safety in the presence of complex, linked data
    structures.
    Despite its usefulness, several open problems persist for this kind of analysis;
    two of which we address in this paper.
    On the one hand, the original analysis is path-sensitive but cannot
    combine safety requirements for related branches.
    This causes the analysis to require additional soundness checks and decreases 
    the analysis' precision.
    We extend the underlying symbolic execution and propose a framework for 
    \emph{shared abduction} where a common pre-condition is maintained for related 
    computation branches.

    On the other hand, prior implementations lift loop acceleration methods from forward
    analysis to biabduction analysis by applying them separately on the pre- and 
    post-condition, which can lead to imprecise or even unsound acceleration results
    that do not form a loop invariant.
    In contrast, we propose \emph{biabductive loop acceleration}, which explicitly
    constructs and checks candidate loop invariants.
    For this, we also introduce a novel heuristic called \emph{shape extrapolation}.
    This heuristic takes advantage of locality in the handling of list-like data
    structures (which are the most common data structures found in low-level code)
    and jointly accelerates pre- and post-conditions by extrapolating the related
    shapes.

    In addition to making the analysis more precise, our techniques also make
    biabductive analysis more efficient since they are sound in just one analysis phase.
    In contrast, prior techniques always require two phases (as the first phase can
    produce contracts that are unsound and must hence be verified).
    We experimentally confirm that our techniques improve on prior techniques; both
    in terms of precision and runtime of the analysis.
\end{abstract}

\begin{document}
\citestyle{acmnumeric}
\maketitle

\begin{acks}
    This work was supported by the Czech Science
    Foundation project 23-06506S and the FIT BUT internal project FIT-S-23-8151.
    The work of the Austrian team leading to this result has received funding from the European Union's Horizon 2020 research and innovation programme under grant agreement No 101034440.
    The collaboration of the teams was also partially supported under the project VASSAL:
    ``Verification and Analysis for Safety and Security of Applications in
    Life'' funded by the European Union under Horizon Europe WIDERA Coordination and Support Action/Grant Agreement No. 101160022.
    \includegraphics[width=.03\textwidth]{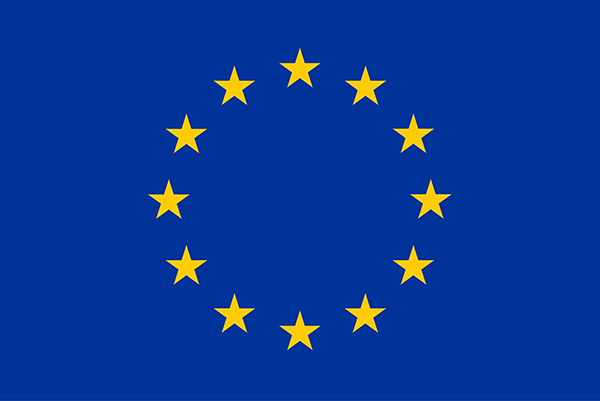}
\end{acks}


\section{Introduction}\label{sec:intro}

Over the last two decades, \emph{shape analysis} has proven to be one of the
most useful techniques for ensuring memory safety in programs.
This kind of analysis focuses on verifying memory-safe handling of linked data
structures by representing them as abstract memory \emph{shapes}.
Thereby, memory safety can often be verified with fully automatic reasoning for
a wide range of data structures.
Examples of successful shape analyzers include the tools Predator
\citep{predator}, which has won a number of medals at the well-known SV-COMP
competition (see \citep{SVCOMP20,SVCOMP24}), and Infer \citep{Infer}, which has
been used for several years to check large code bases at Meta (formerly
Facebook).

\paragraph{Biabduction-Based Shape Analysis}

Among the reasoning principles underlying shape analysis, \emph{biabduction} has
the unique ability to enable compositional analysis for open programs
(i.e., program fragments) by synthesizing invariants and function contracts
consisting of separation logic formulas \citep{reynolds}.
The ability to synthesize contracts, i.e., pairs of pre- and post-conditions,
allows for a highly modular inter-procedural analysis.
In addition, the compositional analysis with biabduction enjoys what 
\citet{biabd_conference,biabd} have called ``graceful imprecision'', i.e., the
analysis will find useful results for some parts of a program, even if it
introduced imprecisions for other parts.
In contrast, closed program analyses are likely to build up such imprecisions
and fail, even if they could handle further parts of the program otherwise.
Due to this and due to not requiring programmers to develop verification 
harnesses, biabduction-based shape analysis is often 
considered to be more useful for large-scale verification compared to techniques
for closed programs, which are advantageous for smaller, self-contained programs.
However, this advantage comes at the cost of more complex computation principles
as well as generally less expressive abstract shapes.
Moreover, existing biabduction-based shape analyses can compute unsound results
and, thus, require a second analysis phase to filter out these results.

\paragraph{Branching and Abduction.}

The highly path-sensitive analysis proposed by \citet{biabd_conference,biabd} does not
work well with branching if the branch to be taken cannot be determined purely
from the pre-condition of the analyzed function (e.g., because it depends on user
input or because the used logical fragment cannot express the dependency
sufficiently), see \citep[section 4.3]{biabd}.
We call these cases \emph{non-determinable branching}.
The problem with these cases arises since maintaining different pre-conditions
for each program
path is, in general, insufficient because the only sound precondition might
consist of a combination of these.
Thus, Calcagno et al. suggested a heuristic procedure for merging
pre-conditions, but their approach may fail to compute any valid pre-conditions
(see \Cref{subsec:sharemot}) at all.
In contrast, we present a novel technique of \emph{shared abduction}, which
allows for sound pre-condition computation across program branches.
The technique extends biabductive symbolic execution by tracking which program
locations share which pre-condition requirements.
Due to this, shared abduction circumvents the need for a verification phase for
programs with arbitrary branching, and, at the same time, can infer non-trivial
contracts in more cases than the traditional analysis (since it does not give up
on the cases where some sharing of information between branches is necessary).

\paragraph{Accelerating Biabductive Analysis}

Symbolic analyses generally require loop acceleration techniques to allow the
analyses to reach a fixed point.
In shape analysis, this acceleration replaces concrete pointers with more
abstract, typically inductive shape predicates such as list segments.
The analysis by \citet{biabd_conference} lifts such abstraction to the setting
of biabductive analysis by applying abstraction separately to the pre- and the
post-condition.
However, such a direct lifting is not guaranteed to result in a sound loop
invariant (see \citep[section 4.3]{biabd} or \Cref{subsec:forwardcomp}).
In contrast, we introduce a novel \emph{biabductive loop acceleration} scheme
that constructs candidate loop invariants after analyzing the loop body once.
This step allows us to verify the soundness of the candidate invariant
explicitly through another symbolic execution of the loop body.
To construct the candidate invariant, we introduce a novel \emph{shape
extrapolation} heuristic, which exploits the locality of typical data structure
traversals to find fitting abstract shape predicates.

\paragraph{Two Phases for Soundness?}

Even though the two-phase approach is easy to implement as the two phases
primarily differ in which biabduction rules are used, and even though most
related works rely on the second phase (see \citep{Broom,bottomup}), it is
quite natural to wonder whether the overhead of the two phases can be reduced.
Overcoming this overhead for the broader family of biabduction-based shape 
analyses is exactly the goal of this paper.
Thereby, our shared abduction technique avoids the unsoundness problem for
non-determinable branching and circumvents the need for the second phase (at the
same time, producing more non-trivial contracts than the previous approaches).
In addition, our loop acceleration approach only requires us to check the
extrapolated invariant for soundness, but this check is much more local and less
costly than the second analysis phase.\footnote{The traditional
approaches to biabduction, such as \citep{Broom,bottomup}, will analyze each loop
at least twice (to get to a fixpoint) in each of the phases, i.e., each loop is
analyzed \emph{at least four times}, but often even more (see \Cref{tab:iter}).
On the other hand, our approach may soundly find a loop invariant within one
analysis phase, which analyses the loop in general only twice.}
Our approach to constructing the invariant is heuristic and, hence, not always
applicable, but the underlying invariant check is still guaranteed to assert 
soundness.
In addition, our experiments demonstrate that our approach can significantly
improve the efficiency, i.e., reduce the needed number of loop iterations and
lower the runtime compared to the two-phase biabduction architecture in
practice.
We conjecture that this is because programmers do commonly write loops in a way
compatible with our approach.

\paragraph{Main Contributions}

The main contributions of this work in the context of analyzing sequential,
non-recursive programs are the following:

\begin{itemize}

  \item A novel sound analysis for loop-free code based on \textit{shared
  abduction} (\Cref{sec:worlds}).

  \item A novel sound way to construct and check loop invariants as part of
  biabductive shape analysis via \emph{biabductive loop acceleration}.
  It uses a novel heuristic to exploit locality via \textit{shape extrapolation}
  on list-manipulating programs (\Cref{sec:extrapol}).

  \item Formal proofs of shared abduction and biabductive loop acceleration
  being sound (\Cref{thm:worlds,thm:loops} with proofs in \Cref{app:proofs}).

  \item An experimental evaluation based on a proof-of-concept implementation 
  applied to a number of small-scale but rather challenging programs, including
  real-life library functions, that show the superiority of our approach
  with regard to runtime and completeness compared to established analyzers
  (\Cref{sec:impl}).

\end{itemize}

\paragraph{General Limitations}

Our acceleration approach is currently limited to programs manipulating various
kinds of lists (singly or doubly linked, possibly circular, nested, and
intrusive).
While this restriction coincides with prior work~\citep{biabd,Broom}, we hope
that exploiting locality for loop acceleration will apply to further data
structures (such as trees), but we must leave this for future work.
Furthermore, we only focus on non-recursive programs, following most previous
biabductive shape analysis approaches.
\citet{10.1016/j.scico.2017.05.007} introduced an extension to handle
recursion via a fixed point computation of the function contract,
but this is orthogonal to our work.
Moreover, our prototype tool is focused on low-level C code,
which rarely contains recursion anyway.



\section{Motivation}\label{sec:mot}

This section introduces the main ideas and motivations behind our work on an
intuitive level while also discussing why the classical biabductive shape analysis
cannot handle the presented example code well.
\Cref{sec:worlds,sec:extrapol} contain the technical details.

\subsection{Cross-Branch Abduction Sharing}\label{subsec:sharemot}

\Cref{lst:mot1} shows a program fragment that works on a data node based on user 
input read from the command line via \texttt{scanf}.
Based on this, \texttt{user\_choice} takes the user input and calls
\texttt{memcpy} with the corresponding arguments.
The exact invocation depends on the user input,\footnote{We chose user input 
as an easy to understand example of non-determinable input. Other cases of such 
input includes IO operations such as incoming network traffic or reading from a  
file.} which is modeled by a
non-deterministic choice in the analysis, and, in the case of
$\texttt{in\_mode}= 1$, also depends on whether $\texttt{curr}=\texttt{hd}$.
We call these kinds of branching \emph{non-determinable}, since the branch taken
at runtime can't be determined from the function parameters alone. 

\begin{wrapfigure}[15]{l}{0.4\textwidth}
  \begin{lstlisting}[language=C,label=lst:mot1,
    caption={Non-determinable branching}]
int user_choice(node *hd, node *lst, 
    node *curr, node *out) {
  int in_mode = 0;
  int read = scanf("%d", &in_mode);

  if (read <= 0) {
    return -1;
  } else if (in_mode == 0) {
    memcpy(out, hd,...);
  } else if (in_mode == 1) {
    if (curr != hd) {
      memcpy(out, curr,...);
    } else {
      memcpy(out, lst,...);
    }
  }
  return in_mode;
}\end{lstlisting}%
\end{wrapfigure}

\paragraph{Problem}

Non-determinable branching is difficult to handle for a
path-sensitive biabduction-based shape analysis as proposed in~\citep{biabd,biabd_conference}.
This is because such an analysis will generate one precondition per program 
branch in \Cref{fig:mot1}, expressed with standard separation logic 
connectives.\footnote{We use ``$\formsep$'' to separate the formulas' pure and 
  spatial parts (if any).
  In contrast to the program variables \emph{hd}, etc., the
  $\ell_i$ variables are purely logical and implicitly universally quantified.
  We write program variables in formulas in \emph{italic} and otherwise in
  \texttt{typewriter} font.}
We note that the preconditions can contain branching conditions that depend on the function's arguments, e.g., the preconditions in \Cref{fig:mot1} contain the predicates $\mathit{curr}=\mathit{hd}$ and $\mathit{curr}\neq\mathit{hd}$.
However, non-determinable branching, such as for $\mathit{in\_mode}=0$, cannot be modeled in terms of the function's arguments, and hence such conditions can never be part of a precondition.
Then the problem arises that the different branches require different memory locations to be allocated (note the  different pointers arguments to \texttt{memcpy}),
e.g. $\mathit{hd}\mapsto\ell_1*\mathit{out}\mapsto\ell_4$ for the branch with
$\texttt{in\_mode}=0$.
However, due to the non-deterministic input, none of the required allocations for one branch guarantee a memory-safe execution for all user inputs.

\begin{figure}[t]
  \captionsetup{skip=-1pt}
  \begin{gather*}
    \textsf{true}\: \text{ for } \texttt{read}\leq0 \lor \texttt{in\_mode}\geq2 \quad\mid\quad \mathit{hd}\mapsto\ell_1*\mathit{out}\mapsto\ell_4\: \text{ for } \texttt{in\_mode}=0\\
    \mathit{curr}\neq\mathit{hd}\formsep\mathit{curr}\mapsto\ell_3*\mathit{out}\mapsto\ell_4 \text{ and }
    \mathit{curr}=\mathit{hd}\formsep\mathit{lst}\mapsto\ell_2*\mathit{out}\mapsto\ell_4 \ \text{ for } \texttt{in\_mode}=1
  \end{gather*}
  \caption{Insufficient candidate pre-conditions for \texttt{user\_choice}}
  \label{fig:mot1}
  \Description{Insufficient candidate pre-conditions for all branches of the
    program. The pre-conditions contain at most two points-to predicates, although
    less than three cannot be sound in general due to the non-determinism of the
    function.}
\end{figure}

\paragraph{Previous Solutions}

%
This problem has already been noticed in the original work~\citep{biabd,biabd_conference} and was partially addressed by a heuristic that would combine pre-conditions such that they could cover move branches.
This heuristic has been implemented as an optional strategy in the Abductor tool and was subsequently made the default in the Infer tool.
However, this heuristic does not guarantee that the found pre-conditions are
sound, and it can easily miss safe pre-conditions even for simple loop-free code.
Indeed, for the example above, the heuristic finds the combined pre-condition $\mathit{curr}\neq\mathit{hd}\formsep\mathit{hd}\mapsto\ell_1*\mathit{lst}\mapsto\ell_2*\mathit{curr}\mapsto\ell_3*\mathit{out}\mapsto\ell_4$,
but it does not produce a sound pre-condition for the case $\texttt{curr}=\texttt{hd}$.
We also remark that the heuristic is quite fragile as renaming the variable
\texttt{hd} to \texttt{first} enables the heuristic to find a safe pre-condition for the case $\texttt{curr}=\texttt{hd}$ in Infer while leading to a crash for
Abductor.\footnote{We observed this behavior with the commit \inferv\\ found at
  \url{https://github.com/facebook/infer} and the publicly available Abductor
  release at \url{http://www0.cs.ucl.ac.uk/staff/p.ohearn/abductor.html}.}

We finally note that recent work \citep{outcomesl}, developed concurrently
with our approach, also addresses the problem of unsound pre-conditions for
branching programs.
They introduce a specialized operator called \emph{tri-abduction}, which generalizes the setting of \emph{bi-abduction} to simultaneously compute a combined pre-condition for two branches.
While this operator offers more precision than the (classical) biabduction operator we rely on in this paper, it is unclear how to build a realistic symbolic execution based on the tri-abduction operator.
To this date, such an analysis has only been sketched but not implemented.
We comment more on the relationship to our approach in \Cref{sec:relatedwork}.

\paragraph{Shared Abduction}
%
The fundamental problem discussed above is that the different program paths
cannot be analyzed in isolation; instead, we must combine their preconditions.
That is, biabduction-based analyzers need to track which program configurations
can be reached from the same initial configuration and synchronize the abduced
requirements.
Moreover, the analysis needs to be precise in tracking which configurations are
allowed to exchange such information -- otherwise, inconsistencies can be
introduced by exchanging information among independent program points.
Our solution is, therefore, to track exactly which program configurations can be
reached from a common pre-condition and explicitly share newly found
requirements with all of these configurations (and only such configurations).
\Cref{sec:worlds} introduces how this technique, which we call \textit{shared
abduction}, allows sound handling of all kinds of branching.

Our technique has the advantage of being lightweight and easily implementable on
top of an existing biabductive analysis.  In the example above, our analysis
first abduces the precondition
$\mathit{hd}\mapsto\ell_1*\mathit{out}\mapsto\ell_4$ for the case
$\texttt{in\_mode}=0$.
The analysis then proceeds with the branch for $\texttt{in\_mode}=1$, making
a case distinction on $\mathit{curr}=\mathit{hd}$.
Shared abduction retains the required allocation
$\mathit{hd}\mapsto\ell_1*\mathit{out}\mapsto\ell_4$ for both cases as this
requirement is already part of the shared precondition.
Then, by analyzing the nested branches, the requirements
$\mathit{curr}\neq\mathit{hd} \formsep \mathit{hd}\mapsto\ell_1 *
\mathit{curr}\mapsto\ell_3 * \mathit{out}\mapsto\ell_4$ and
$\mathit{curr}=\mathit{hd} \formsep \mathit{hd}\mapsto\ell_1 *
\mathit{lst}\mapsto\ell_2 * \mathit{out}\mapsto\ell_4$ are computed.
We note that this guarantees the soundness
of the found pre-condition and its completeness regarding the branching, thus
outperforming the previous heuristic.

\subsection{Shape Extrapolation for Biabductive Acceleration}\label{subsec:extramot}

Our second contribution aims at the analysis of loops, which generally
requires accelerating the symbolic execution to allow the analysis to reach a
fixed point.

\paragraph{Problems}

\begin{wrapfigure}[10]{l}{0.45\textwidth}
  \begin{lstlisting}[language=C,label=lst:mot2,caption={Nested list traversal}]
void weighted_sum(o_node *o, long *sum){
  while (o != NULL) {
    i_node *i = o->inner;
    while (i != NULL) {
      *sum = (*sum) + 
        (o->wgt * i->elem);
      i = i->next;
    }
    o = o->next_o;
  } 
}\end{lstlisting}
\end{wrapfigure}

The prior technique for loop acceleration, proposed by \citep{biabd_conference,biabd} and adopted
in \citep{Broom}, separately abstracts the pre- and post-condition
with no other information than the formulas themselves taken into account.
Intuitively, after analyzing some loop iterations and applying the abstraction 
operator, the obtained formulas will stabilize, and a fixed point is reached.
Thereby, the abstraction follows the intuitive principle of collecting linked
memory blocks with a similar layout into a single abstract
shape predicate.
In the context of simple singly-linked lists, this means that the abstraction
procedure scans the formula for points-to predicates $x \mapsto \ell$ and
$\ell\mapsto z$, linked by a location $\ell$ (i.e., the target of the first 
predicate contains the address of the second), or a linked list segment 
$\ls(x,\ell)$ and a points-to $\ell\mapsto z$, respectively.
Abstraction then replaces these predicates with the single predicate $\ls(x,z)$.
However, abstraction cannot be applied when there is a program variable $y$ that
references $\ell$, e.g., as $y = \ell$.
This is not supported since it would lose the information that
variable $y$ is allocated (note that $\ell$ does not occur in $\ls(x,z)$
anymore).
More generally, abstraction cannot be applied if there is a program
variable $y$ whose value depends on $\ell$, such as $y = v\land \ell\mapsto v$. 
While the abstraction principle is intuitive, there are also multiple 
drawbacks, which we discuss next.

$(1)$ For the example in \Cref{lst:mot2}, Abductor, Infer, and Broom do not reach a fixed point for the inner loop and thus cannot synthesize any contract.
This is because the value pointed to by sum after $n$ loop iterations is $\ell_\mathit{sum}+\ell_1\cdot \ell_w+\dots +\ell_n\cdot \ell_w$, where $\ell_w$
is the value pointed to by $o.\texttt{wgt}$ and the $\ell_i$ are the \texttt{elem} values of the list nodes traversed so far.
Thus, the dependence on the values $\ell_i$ blocks abstraction (as described above).
We note that the design of a stronger abstraction operator is not straight-forward because we also need to track values in memory precisely as they could be relevant for memory accesses based on pointer arithmetic in other parts of the program.

\begin{wrapfigure}[6]{l}{0.4\textwidth}
  \begin{lstlisting}[language=C,label=lst:mot3,caption={Offset list traversal}]
void traverse_skip_two(node *list) {
  node *tmp = list->next->next;
  while (tmp != NULL) {
    tmp = tmp->next;
  } 
}\end{lstlisting}
\end{wrapfigure}

$(2)$ A formula-based abstraction operator can easily lose too much information.
For example, Abductor, Infer, and Broom fail for the simple example in \Cref{lst:mot3}.
The reason is as follows:
The abstraction operation (as described above) contracts pointers chains of length at least two into a list segment, resulting in the formula
$\mathit{list}\neq\Null\formsep \mathit{ls}(\mathit{list},\Null)$.
This predicate describes a non-empty list segment with at least one node.
However, this formula does not suffice as a pre-condition that guarantees memory safety because \texttt{traverse\_skip\_two} requires a list of length at least two as input.
We provide a more detailed comparison with our work in \Cref{subsec:forwardcomp}.

$(3)$ Furthermore, acceleration based on abstraction (as described above and implemented in Abductor, Infer, and Broom) can be highly inefficient.
In general, every loop will require at least two (often three) analysis 
iterations, as abstraction can often only be applied after the second loop
iteration and a fixed point can only be checked for after another iteration.
In the presence of inner loops, such as for the example in \Cref{lst:mot2}, this quickly multiplies, e.g., amounting to nine analysis iterations for the inner
loop in \Cref{lst:mot2} just for the first analysis phase.

\paragraph{Locality}

The problems described above are mostly related to the direct application of abstraction
for acceleration and its missing ability to take into account more information
about the loop, e.g. the observation that inductive data structures
are often traversed one step at a time.
For example, in \Cref{lst:mot2}, the nested linked list is traversed in such a way
to compute the weighted sum of the elements in the inner lists.
It is apparent that each of the two loops operates on a local, shifting view of
the respective traversed list plus some context.
For the inner loop, this means that the loop only operates on a unique
\texttt{i\_node} at a time while also accessing the same variables
\texttt{sum} and \texttt{o->wgt} in each iteration.
Similarly, the outer loop only operates on one \texttt{o\_node} at a time.
These shifting views on the traversed shapes are akin to ``local actions'' (see 
\citep{abstrsl}).

\paragraph{Biabductive Loop Acceleration}

This observation allows us to extrapolate what the analysis abduces from a
single iteration to arbitrarily many iterations and directly compute a candidate
loop invariant if applicable.
We call this heuristic \textit{shape extrapolation}.
It is part of our \emph{biabductive loop acceleration}, which consists of three
main steps: First, we use the analysis result of a single iteration to obtain
locality information about the shape and the context; second, we use this
information to extrapolate the shape to an abstract one; third, we check that
the heuristically constructed state is a sound invariant.

\begin{figure}[t]
  \pgfmathsetmacro{\sep}{1.7}
  \begin{tikzpicture}[
    cell/.style={rectangle,draw=black,minimum width=1cm,minimum height=1cm},
    list/.style={rectangle},
    MyLongArrow/.style args={#1 -- #2}{
        insert path={let \p1=($(#1)-(#2)$) in}, 
        single arrow, draw=black, minimum width=15mm, minimum height={veclen(\x1,\y1)}, inner sep=0mm, single arrow head extend=1mm, double arrow head extend=1mm
    },
    scale=0.97,
    every node/.style={transform shape}
    ]
  
  \node[list] (1) at (0,0) {$s_1:(I\neq\Null\land i=I \formsep$};
  \node[list] (2) [right of = 1,xshift=3.2cm] {$\highlight{ACMRed}{I.\texttt{next}\mapsto\ell_1*I.\texttt{elem}\mapsto\ell_2}$};
  \node[list] (3) [right of= 2,xshift=2.6cm] {$*o.\texttt{wgt}\mapsto\ell_3*$};
  \node[list] (4) [right of=3,xshift=1.5cm] {$\highlight{ACMBlue}{\mathit{sum}\mapsto \ell_4}\statesep$};
  \node[list] (5) at (.3,-.7) {$I\neq\Null\land i=\ell_1\formsep $};
  \node[list] (6) [right of =5,xshift=2.9cm] {$\highlight{ACMOrange}{I.\texttt{next}\mapsto\ell_1*I.\texttt{elem}\mapsto\ell_2}$};
  \node[list] (7) [right of=6,xshift=2.6cm] {$*o.\texttt{wgt}\mapsto\ell_3*$};
  \node[list] (8) [right of=7,xshift=1.6cm] {$\highlight{ACMLightBlue}{\mathit{sum}\mapsto\ell_4+\ell_3\cdot\ell_2})$};

  \node[list] (11) at (0,-\sep) {$s_\mathit{inv}:(I\neq\Null\land i=I \formsep$};
  \node[list] (12) [right of = 11,xshift=2.2cm] {$\highlight{ACMRed}{\ls(I,\ell_1)}*$};\node[list] (122) [right of = 12,xshift=.75cm] {$\highlight{ACMRed}{\ls(\ell_1,\Null)}$};
  \node[list] (13) [right of= 122,xshift=1.85cm] {$*o.\texttt{wgt}\mapsto\ell_3*$};
  \node[list] (14) [right of=13,xshift=1.5cm] {$\highlight{ACMBlue}{\mathit{sum}\mapsto \ell_4}\statesep$};
  \node[list] (15) at (.4,-\sep-.7) {$I\neq\Null\land i=\ell_1\formsep $};
  \node[list] (16) [right of =15,xshift=1.8cm] {$\highlight{ACMOrange}{\ls(I,\ell_1)}*$};
  \node[list] (9) [right of =16,xshift=.75cm] {$\highlight{ACMRed}{\ls(\ell_1,\Null)}$};
  \node[list] (17) [right of=9,xshift=1.85cm] {$*o.\texttt{wgt}\mapsto\ell_3*$};
  \node[list] (18) [right of=17,xshift=1.4cm] {$\highlight{ACMLightBlue}{\mathit{sum}\mapsto\top})$};

  \node (19) [right of=18, xshift=1.1cm, yshift=.9cm] {$ $};

  \begin{scope}[on background layer]
  \draw[-{[flex']>}, in=140, out=230,draw=ACMRed,line width=2pt,dashed,looseness=.9] (2.west) to (12.west);
  \draw[-{[flex']>}, in=140, out=230,draw=ACMOrange,line width=2pt,dashed,looseness=.9] (6.west) to (16.west);
  \draw[-{[flex']>}, in=50, out=300,draw=ACMRed,line width=2pt,dashed] (2.east) to (9.east);
  \draw[-{[flex']>}, in=50, out=300,draw=ACMRed,line width=2pt,dashed,looseness=1.3] (2.east) to (122.east);

  \draw[-{[flex']>}, in=150, out=220,draw=ACMBlue,line width=2pt,dashed] (4.west) to (14.west);

  \draw[-, in=90, out=0,draw=ACMBlue,line width=2pt,looseness=1.5,dashed] (4.east) to (19);
  \draw[-, in=90, out=0,draw=ACMLightBlue,line width=2pt,looseness=1.3,dashed] (8.east) to (19);
  \draw[-{[flex']>}, in=0, out=-90,draw=ACMLightBlue,line width=2pt,dashed] (19.north) to (18.east);
  \draw[-{[flex']>}, in=0, out=-90,draw=ACMBlue,line width=2pt,dotted] (19.north) to (18.east);
  \end{scope}

\end{tikzpicture}
  \caption{State $s_1$ after the first loop iteration analysis and the constructed
  candidate invariant $s_\mathit{inv}$, with color-coded arrows showing the
  information flow between different subformulas}
  \Description{The figure shows that the shape information after the analysis of
  one loop iteration is used to construct the shapes of the candidate invariant.
  It also shows how other changes, like the value of \texttt{sum}, is
  abstracted as part of the invariant.}
  \label{fig:mot_extrapol}
\end{figure}
In the case of \texttt{weighted\_sum}, after the first iteration of the inner 
loop, the analysis finds the state
$s_1$ depicted in \Cref{fig:mot_extrapol}, consisting of a pre- and post-condition separated by $\statesep$.
Our analysis then partitions the pre- as well as the post-condition into a shape 
and a context part, where the shape part is $I.\texttt{next}\mapsto\ell_1*I.\texttt
{elem}\mapsto\ell_2$, and the context is $o.\texttt{wgt}\mapsto\ell_3*\mathit{sum}\mapsto \ell_4$ 
for the pre-condition and $o.\texttt{wgt}\mapsto\ell_3*\mathit{sum}\mapsto \ell_4 + \ell_3\cdot \ell_2$ for the post-condition.
The following heuristic obtains this partitioning:
We consider the changed variables (here \texttt{i},\texttt{sum}) whose value moved to some pointer value (here \texttt{i} whose value moved to \texttt{i->next}).
The predicates associated with these variables are then put into the shape parts and the others into the context.

Based on this partitioning, our procedure directly constructs a (candidate) loop invariant $s_\mathit{inv}$ (see \Cref{fig:mot_extrapol}).
Thereby, our procedure accelerates the shape part of the pre- as well as the post-condition  (here we obtain in both cases the list segment $\ls(I,\ell_1)$, with internal next pointer \texttt{next} and data field \texttt{elem}).
Intuitively, these predicates correspond to the loop iterations up to the current point. 
In addition, we add predicates $\ls(\ell_1,\Null)$ to the shape part of $s_\mathit{inv}$, for both the pre-and post-condition,
which are taken as the accelerated predicate $\ls(I,\ell_1)$ of the pre-condition of $s_1$, where $I$ has been replaced with $l_1$, which is the current value of $i$, and $l_1$ has been replaced with $\Null$, which has been obtained from the loop condition.
Intuitively, these predicates correspond to the future loop iterations up to the loop's termination.
We refer to the red and orange colors in Fig.~\ref{fig:mot_extrapol} to illustrate the information flow.
The context part of $s_\mathit{inv}$ keeps the context of $s_1$, except that our procedure abstracts the value of \texttt{sum} in the post-condition (with the unknown value $\top$) because it cannot be tracked precisely.
Next, our procedure checks that $s_\mathit{inv}$ is indeed a loop invariant, which requires one symbolic execution of the loop body and an entailment check.

Finally, based on the loop invariant our analysis constructs a contract that abstracts the inner
loop and that can be used for the analysis of the outer loop:
\begin{align*}
  (i=I      & \formsep \ls(I,\Null)*o.\texttt{wgt}\mapsto\ell_3*\mathit{sum}\mapsto
  \ell_4\statesep                                                                         \\
  \ i=\Null & \formsep \ls(I,\Null)*o.\texttt{wgt}\mapsto\ell_3*\mathit{sum}\mapsto\top
  ).
\end{align*}
Based on this contract, our procedure then also accelerates the outer loop in a
similar fashion.
Lastly, we note that our procedure requires exactly two
iterations per loop (in sum four): 
one to analyze the effects of the loop and a
second to check whether the constructed state is a loop invariant (as opposed to
the nine iterations in sum mentioned earlier for the traditional acceleration).
Similarly, our approach fails fast if the constructed candidate invariant is 
unsound instead of requiring a second analysis phase with many more analysis
steps.



\section{Preliminaries}\label{sec:prelim}

We present our new techniques for a simple but standard setting
that is described in the following since it does not require any specific logic
fragment or biabduction method.

\subsection{Programming Language and Memory Model}

Let \textit{Var} be a countably infinite set of variables and \textit{Val}
be a countably infinite set of values containing the value $\top$.
Furthermore, let \textit{Fld} be a finite set of field names and
$\mathit{Loc}\subseteq \mathit{Val}$ be the set of memory locations such that
$\Null \in \mathit{Loc}$.
If a value $v$ has a field $f$, we write $v.f$ to denote the value part
corresponding to that field.
Moreover, dereferencing without explicit fields is implicitly encoded as dereferencing an auxiliary field \texttt{data}.
Further, we assume that $\mathbb{N}\subseteq \mathit{Val}$.

\begin{figure}[t]
    \begin{gather*}
        \mathit{expr} \mathdef \Null \mid k\in\mathbb{N} \mid x\in\mathit{Var}
        \mid\ ?\mid \textsf{unop } \mathit{expr} \mid \mathit{expr}
        \textsf{ binop } \mathit{expr} \quad
        \oplus\ \mathdef\ =\mid\neq\mid\leq\mid\geq\mid <\mid >\\
        \mathit{stmt} \mathdef\ x\!=\!\mathit{expr} \mid x_1\!=\!*x_2.f \mid
        *x_1.f\!=\!x_2\mid \textsc{return}\ x
        \mid \textsc{Assume}(x_1\oplus x_2)\mid \\
        \textsc{Assert}(x_1\oplus x_2) \mid x_r=f(x_1,\dots,x_n)\mid x=\textsc{Alloc}(f_1,\dots, f_n)\mid \textsc{Free}(x)
    \end{gather*}
    \caption{The syntax of the programming language $L$.}
    \Description{The syntax of our programming language L. It is close to a subset
        of C99 with load and stores, as well as assume and assert commands.}
    \label{fig:lang}
\end{figure}

\begin{definition}[Programming Language L]\label{def:prog-lang}
    We define a C-like \textit{programming language L} in \Cref{fig:lang}.
    The language comprises standard expressions, statements for reading ($x_1=*x_2.f$) and
    writing ($*x_1.f=x_2$) through pointers (with the C-like syntactic sugar
    of $x\texttt{->}f$ for $*x.f$), an additional non-determinism operator
    $?$, as well as implicit control flow statements \textsc{Assume} and
    \textsc{Assert}.
    Moreover, we include the C-like commands for memory (de-)allocation \textsc{Alloc} and \textsc{Free}.
\end{definition}

We represent \textit{functions} in $L$ by a function name $f$, a list of
argument variables $a_1,\dots,a_n$, $n \geq 0$, and a
function body $\mathit{body}_f$ that consists of a \textit{control flow graph}.
A \textit{control flow graph} (CFG) is a tuple
$(V,E,\mathit{entry}_f,\mathit{exit}_f)$ such that $V$ is a set of program
locations with dedicated locations $\mathit{entry}_f,\mathit{exit}_f\in V$,
and $E\subseteq V\times\mathit{stmt}\times V$ is a set of edges between
program locations labeled with statements from $\mathit{stmt}$.
A \textit{trace} $t$ of a CFG $(V,E,\mathit{entry}_f,\mathit{exit}_f)$ is an
alternating sequence $[v_0,st_1,v_1,\dots,st_n,v_n]$, $n \geq 0$,  of
vertices $v_i\in V$ and statements $st_{i+1} \in \mathit{stmt}$ such
that $(v_i,st_{i+1},v_{i+1})\in E$ for all consecutive
$v_i,st_{i+1},v_{i+1}$ in $t$, $0 \leq i < n$.
If only a part of a trace is relevant, we write
$[t,v_0,st_1,\dots,v_n]$ for the trace continuing from $v_0$ and reaching $v_n$ such that $t$ ends in $v_0$.

We note that CFGs as stated in \Cref{def:prog-lang} can be used to model
arbitrary branching and looping constructs (such as \texttt{if-then-else} and
\texttt{while}), and hence the basic statements of programming language $L$ do
not need to cover these features.
We will further make the following assumptions:
\begin{enumerate}
    \item Each function $f$ is either loop-free or consists of a single loop
    such that the loop header is $\mathit{entry}_f$;
    i.e., we require that the CFG of $f$ is either acyclic or all
    back-edges of $f$ (the edges returning to a loop header) return to
    $\mathit{entry}_f$.
    This assumption is w.l.o.g, as loops that are 
    embedded in a bigger context can be represented by calls to a function
    whose body is precisely the loop.\footnote{More complex cyclic control 
    flow, e.g. describing the common \texttt{break}/\texttt{continue}/\texttt{goto}
    statements can also be emulated by introducing auxiliary out parameters
    which are then used to guide the control flow.
    }

    \item The programs do not contain (mutually) recursive functions.

    \item Each vertex $v\in V$ has at most two outgoing edges in $E$.
\end{enumerate}

\begin{definition}[Program Configuration]
    A \textit{program configuration} $\mathit{cnf}\in\conf$ is either a pair
    $(S,H)$ consisting of a \textit{stack} $S$ and a \textit{heap} $H$ or the 
    dedicated \textit{err} configuration.
    The stack $S : \mathit{Var} \rightharpoonup \mathit{Val}$ is a partial map
    from variables to values.
    The heap $H : (\mathit{Loc}\times \mathit{Fld}) \rightharpoonup_\mathit{fin}
        \mathit{Val}$ partially maps finitely many pairs of memory locations and
    field names into values.
\end{definition}

The semantics of the programming language $L$ is standard (its formalization can be
found in \Cref{fig:langsem} in the appendix).
We use the notation $\left(\mathit{cnf}_1,\mathit{st}\right) \rightsquigarrow
\mathit{cnf_2}$ to denote that a program reaches a configuration
$\mathit{cnf}_2$ from a configuration $\mathit{cnf}_1$ by executing a statement
\textit{st}.
The semantics of traces is defined as the transitive closure
$\rightsquigarrow^*$ with regard to the statements in the trace
($\mathit{cnf}_i\in \conf$):\begin{align*}
    \left(\mathit{cnf}_1 , [v_0]\right)           & \ \rightsquigarrow^*
    \mathit{cnf}_1,  \\
    \left(\mathit{cnf}_1, [t,v_1,st_2,v_2]\right) & \ \rightsquigarrow^* \mathit{cnf}_3
    \text{ if } (v_1,st_2,v_2)\in E \land \left(\mathit{cnf}_1, [t,v_1]\right)
    \rightsquigarrow^*\mathit{cnf}_2 \land \left(\mathit{cnf}_2, st_2\right)
    \rightsquigarrow \mathit{cnf}_3.
\end{align*}

\subsection{Separation Logic}

Next, we introduce a simple separation logic (\emph{SL}) fragment that is suited
for biabduction-based shape analysis.
Even though most shape analyses in recent literature are based on
more sophisticated fragments, this simple fragment suffices to define
our central contributions, which can be easily lifted to more powerful
separation logic fragments as well (indeed, we use a more expressive fragment
in our later presented experiments).
The formulas of \emph{SL} are based on symbolic heaps \citep{decidable}.

\begin{definition}[Separation Logic SL]\label{def:seplog}
    The \textit{separation logic fragment SL} contains the standard connectives
    of separation logic and an inductive predicate \textit{ls} denoting a singly-linked list segment.
    \Cref{fig:seplogsyntax} shows the full syntax of \emph{SL}.
    Symbolic heaps $\varphi$ clearly distinguish between spatial
    parts $\Sigma$ and pure parts $\Pi$ and are combined to disjunctive symbolic
    heaps $\Delta$.
\end{definition}

Based on this, we define \textit{contracts} for functions in the programming
language to be pairs of formulas from \textit{SL} where we call
the parts of the pair a \emph{pre-condition} and a \emph{post-condition}, respectively,
and denote them as $c.\mathit{pre}$ and $c.\mathit{post}$ for a contract $c$.
Basic contracts for all statements in the language defined in \emph{SL} can be found in \Cref{app:ctxt}.

Formulas from \textit{SL} are also evaluated against program configurations
with a judgement $\vDash$, such that $\mathit{cnf}\vDash P$
denotes that \textit{cnf} is a model of formula $P$.
\Cref{fig:seplogsemantics} in the appendix describes this in full detail.
We note that we chose the standard semantics for $*$
and $\mapsto$.
Furthermore, we introduce the \textit{entailment} judgement of \textit{SL},
written $P\vdash Q$, as: $P\vdash Q \text{ iff }
    \forall \mathit{cnf}\in \conf.\ \mathit{cnf}\vDash P \Longrightarrow
    \mathit{cnf}\vDash Q. $ Below, we may use $\mathit{SL}$ to refer
directly to the language of \textit{SL} formulas.

\begin{figure}[t]
    \begin{gather*}
        \varepsilon\ \mathdef\ \Null \mid k\in \mathbb{N} \mid x\in\mathit{Var}
        \mid \text{\textsf{unop} } \varepsilon \mid \varepsilon_1
        \text{ \textsf{binop} } \varepsilon_2\mid \top \quad \oplus\ \mathdef\
        =\mid\neq\mid\leq\mid\geq\mid <\mid >
        \\
        \Sigma\ \mathdef\ x.f\mapsto \varepsilon \mid \Sigma_1*\Sigma_2\mid
        \mathit{ls}(x,\varepsilon) \mid \text{\textsf{emp}}\quad \Pi\mathdef
        \Pi_1\land\Pi_2\mid \textsf{true}\mid \varepsilon_1 \oplus \varepsilon_2\\
        \varphi \mathdef \Pi\formsep \Sigma \quad \Delta \mathdef \varphi \lor \Delta
        \mid \varphi
    \end{gather*}
    \caption{The syntax of the separation logic fragment \emph{SL}.}
    \Description{The syntax of the chosen separation logic SL. It includes standard
        logic operators, as well as a points-to predicate and an
        inductive list segment predicate.}
    \label{fig:seplogsyntax}
\end{figure}

\paragraph{Variables in SL}
We call all variables occurring in a program $Pr$ the \textit{program
    variables} $\mathit{PVar}\subseteq \mathit{Var}$ of $Pr$ and assume that the
program variables are unique for each function in $Pr$.
We call the variables in $\mathit{LVar}\mathdef\mathit{vars}(P)\setminus
    \mathit{PVar}$ logical variables of a formula $P$.
We define the dedicated logical variable $\mathit{return}_f$ to denote the return
value of a function $f$, if any.

\begin{definition}[Normal Form]\label{def:normal}
    Similar to \citet{otherabstraction}, we define formulas in \textit{SL} to be
    in normal form if they satisfy the following requirements:
    $(1)$ all variables in \textit{PVar} are defined uniquely by an equality with
    a logical variable that denotes their current value;
    $(2)$ other than in these equalities, \textit{PVar}s do not occur in any
    other term.
    This normal form guarantees that even if a program variable has a complex
    value described by a compound term, its value is always represented by a
    single logical variable.
    For the sake of readability, we omit the explicit \textit{PVar} equalities
    in most examples and only show how the formulas would look after simplification.
\end{definition}

\begin{example}
    The formula $x=X\land i<13\formsep \mathit{ls}(X,y)*y\mapsto i$ with $\{x,y,i\}
        \subseteq \mathit{PVar}$ is equivalent to the normal form $x=X\land y=\ell_y
        \land i=\ell_i\land \ell_i<13\formsep \mathit{ls}(X,\ell_y)*\ell_y\mapsto\ell_i$ where
    $\ell_y,\ell_i \in \mathit{LVar}$ are fresh.
\end{example}

\paragraph{Further Notation}

We denote with $P[x/y]$ the formula $P$ with the variable $y\in
    \mathit{LVar}$ substituted with $x$ or with the equality for
$y\in \mathit{PVar}$ exchanged in the normal form of $P$ with $y=x$, respectively.
We often denote a formula $\Pi\formsep\Sigma$ by only $\Pi$ or $\Sigma$ if $\Sigma =
    \textsf{emp}$ or $\Pi = \textsf{true}$, respectively.
Furthermore, we denote the composition of formulas $\varphi_1=\Pi_1\formsep\Sigma_1$
and $\varphi_2=\Pi_2\formsep\Sigma_2$ as $\varphi_1*\varphi_2\mathdef
    \Pi_1\land\Pi_2\formsep\Sigma_1*\Sigma_2$.

\begin{definition}[Abstraction]
    An abstraction function $\abs:\mathit{SL}\rightarrow\mathit{SL}$ takes
    a formula in \textit{SL} and returns a potentially different formula such
    that it abstracts a given formula $P$ such that $P \vdash \abs(P)$.
\end{definition}

\begin{example}
    An abstraction procedure $\abs$ as described by
    \citet{abstraction_origin} abstracts consecutive pointer chains into list
    segments, i.e., $\abs(a.\texttt{next}\mapsto b * b.\texttt{next}\mapsto c) =
        \ls(a,c)$.\footnote{As this step loses information about $b$, it is only applied
        in contexts in which $b$ is not relevant otherwise. See, e.g., \citep{david}.}
\end{example}

\subsection{Biabduction-based Shape Analysis}\label{subsec:analysis}

\begin{definition}[Biabduction]
    \textit{Biabduction} is the process of solving a query $P*\boxed{M}\vdash
        Q*\boxed{F}$ for given SL formulas $P$ and $Q$ by computing an
    \textit{antiframe} (or missing part) $M$ and a \textit{frame} $F$ such that
    the entailment is valid.
\end{definition}

We are only interested in solutions for $M$ that do not contradict $P$, as
otherwise, the entailment would be trivially valid.
A \textit{biabduction procedure} is then an algorithm that, given two
formulas, either computes a fitting frame and anti-frame or fails.
The steps to compute a frame and anti-frame are called \emph{frame inference}
and \emph{abduction}, respectively.
For the sake of saving space, we do not develop a full biabduction procedure here but
refer the reader to \citep{biabd,biabd_conference,Broom} for detailed
descriptions.

\begin{definition}[\AbsState s]
    An \textit{\absstate}\:$s$ is an intermediate contract $(P\statesep Q)$ where
    $P\in \varphi$ and $Q\in \Delta$.
    To distinguish these from finished contracts, we call $P$ the candidate
    pre-condition ($s.\mathit{pre}$) and $Q$ the current post-condition ($s.\mathit{curr}$).
    In an \absstate, each function argument $a_i\in \mathit{PVar}$ is
    associated with an \textit{anchor} variable $A_i\in \mathit{AnchVar}
        \subseteq \mathit{LVar}$ (in upper case) denoting its value at $\mathit{entry}_f$.
    We omit equalities of the form $x=X$ from $\Pi_P$ if they are not relevant.
\end{definition}

\paragraph{Biabductive Symbolic Execution Step}
Let there be an \absstate\:$(P\statesep Q)$ at a program location $l$ for a
statement $st$ with contract $(L \statesep R)$ and a location $l'$ such
that $(l,st,l')\in E$.
Then $st$ can be symbolically executed by solving the biabduction query
$Q*\boxed{M}\vdash L*\boxed{F}$ resulting in the new \absstate\:$(P*M\statesep R*F)$.
As in \citep{biabd,biabd_conference,Broom}, we require that
(1)~$\mathit{var}(M) \subseteq \mathit{LVar}$ and that (2)~$P*M$ is
satisfiable.
If such an $M$ does not exist, we say that the biabduction \emph{fails}.

\begin{definition}[Biabduction-based Shape Analysis $\ana_B$]\label{def:BiAbdSA}
    A basic \emph{biabduction-based shape analysis} $\ana_{B,\abs}$ uses a biabduction procedure $B$
    and an abstraction procedure $\abs$ to analyze programs in our 
    programming language.
    Thereby, it analyzes the functions \emph{bottom-up} along the
    call tree, starting from its leaves.
    In each step, the analysis takes an \absstate\:and symbolically executes
    the next statement from it by updating the state accordingly.
    In the case of multiple contracts, the analysis has to determine the
    applicable ones and continue from each of these.

    Furthermore, the analysis runs for a function $f$ until it reaches a fixed point,
    i.e., until no new \absstate s are computed.
    A common way to check for this condition is to check whether new \absstate s
    entail already computed ones.
    To enforce termination, $\ana_{B,\abs}$  also applies $\abs$ to abstract
    the \absstate s at loop heads.
    Finally, the pairs of candidate pre-conditions and current post-conditions
    forming the \absstate s that reached $\mathit{exit}_f$ become its
    contracts.
\end{definition}

We now fix an arbitrary, but correct biabduction-based shape analysis 
$\ana_{B,\abs}$ , which we extend in the following sections.

\begin{definition}[Soundness of \AbsState s]\label{def:soundst}
    An \absstate\:$s=(P\statesep Q)$ is called \textit{sound} for a trace $t$,
    written as the Hoare triple $\{P\}\ t\ \{Q\}$, iff
    $$\forall \mathit{cnf},\mathit{cnf}'\in \conf.\ \mathit{cnf}\vDash
        P\land \left(\mathit{cnf},t\right)\rightsquigarrow^* \mathit{cnf}'
        \Longrightarrow \mathit{cnf}'\neq \mathit{err}\land \mathit{cnf}'\vDash
        Q. $$
    Similarly, a function contract $c=(P, Q)$ is
    \textit{sound} for $\mathit{body}_f$, written
    $\{P\}\ \mathit{body}_f\ \{Q\}$, iff 
    $\{P\}\ t\ \{Q\}$ holds for all
    traces $t=[\mathit{entry}_f,\dots,\mathit{exit}_f]$ through $\mathit{body}_f$.
\end{definition}

\paragraph{Initial \AbsState s}
The \textit{initial \absstate} $s_0$ for function $f$ has $s_0.\mathit{pre} = \texttt{true}$, $s_0.\mathit{curr}=\bigwedge
    \{x=X\mid x\in \mathit{PVar}\land X\in \mathit{AnchVar}\}$,
which denotes that each program variable has a fixed but initially unrestricted 
value (anchor) at the start of $f$.

\paragraph{Handling of \textsc{Assume}}
Following the seminal work \citep{biabd,biabd_conference} and the more recent
\citep{Broom}, we define biabductive shape analysis to split its states at branching 
points according to the branching condition.
As the literature contains sufficient explanations of this mechanism (called
assume-as-assume and assume-as-assert), we only give a brief intuition here.
If the branching condition can be expressed in terms of the function arguments, 
i.e., if the branch taken can be statically determined purely from the function 
arguments, the analysis includes the two cases into the pre-conditions of the 
resulting states.
This treatment is equivalent to handling the branching condition's \textsc{Assume} 
statements as if they were \textsc{Assert} statements instead.
Otherwise, the \absstate s for the branches have the same pre-condition, and the 
branching condition cases are only added to the corresponding post-conditions.


\section{Sound Branching Analysis with Shared Abduction}\label{sec:worlds}

\begin{wrapfigure}[10]{l}[-.025\textwidth]{0.32\textwidth}
    \vspace*{-1mm}
    \begin{lstlisting}[language=C,label=lst:branching,caption={Nested branching},numbers=left]
  int nested(node *x, node *y,node *z){
    if (?) {
      if (y != NULL) {
        return y->data;
      } else {
        return z->data;
      }
    } else {
      return x->data;
    } 
  }\end{lstlisting}
\end{wrapfigure}

As the example in \Cref{lst:mot1} is rather convoluted, we introduce the simpler
\Cref{lst:branching} to show how exactly the technique works.
There, the function \texttt{nested} loads from one of
the three pointer arguments, depending on a non-deterministic
condition $?$ on Line~2 and a deterministic one on Line 3.
Regardless of the values of the function arguments, an execution can
either take the \texttt{then} or the \texttt{else} branch of the outer
\texttt{if-then-else}.
Therefore, the original analysis simply splits the \absstate\:without abducing
any pre-condition.
In contrast, the branches of the inner \texttt{if-then-else} can be
distinguished by whether the argument \texttt{y} is initially a null pointer,
leading the analysis to abduce different pre-conditions for each branch.
Altogether, the classical biabduction-based shape analysis will find three unsound contracts for the function, one for each possible code path, similar to the following:
\begin{align*}
    (x.\texttt{data}\mapsto \ell_1 \statesep &
      \mathit{return}_\mathit{nested}=\ell_1\formsep x.\texttt{data}\mapsto \ell_1) \\
    (y\neq\Null\formsep y.\texttt{data}\mapsto \ell_2\statesep &
      y\neq\Null\land\mathit{return}_\mathit{nested}=\ell_2\formsep
      y.\texttt{data}\mapsto \ell_2) \\
    (y=\Null\formsep z.\texttt{data}\mapsto \ell_3\statesep &
      y=\Null\land\mathit{return}_\mathit{nested}=\ell_3\formsep
      z.\texttt{data}\mapsto \ell_3)
\end{align*}

As introduced in \Cref{subsec:sharemot}, our new technique overcomes this
unsoundness issue and shares requirements abduced with related \absstate s.
To guarantee that the requirements are only shared with actually related
\absstate s, we introduce so-called \emph{extended
analysis states} or \emph{worlds} for short.

\begin{definition}[Worlds]
    \emph{Worlds} comprise a shared pre-condition $P$ and multiple
    current post-conditions $Q_i^{l_i}$ at possibly different program
    locations $l_i$: $(P\statesep Q_0^{l_0}\lor \dots \lor Q_n^{l_n})$
\end{definition}

We stress the seemingly small but crucial difference between the current
post-conditions used in our notion of worlds and the previously defined
abstract states: the latter are, in general, also allowed to use disjunctions
but are missing the labeling by program locations (allowing the disjuncts to be associated with different program paths).
Moreover, worlds do not require the logic itself to contain disjunctions but merely
simulates them with its structure.

\begin{definition}[Soundness of Worlds]\label{def:soundw}
    A world $w=(P\statesep Q_0^{l_0}\lor \dots \lor Q_n^{l_n})$ is \textit{sound} for a
    trace $t=[v_0,\dots,v_n]$, written $\{P\}\ t\ \{Q_0^{l_0}\lor \dots \lor Q_n^{l_n}\}$, iff
    \[\forall \mathit{conf},\mathit{conf}'\in \conf.\ \mathit{conf}\vDash P\land
    \left(\mathit{conf},t\right)\rightsquigarrow^* \mathit{conf}' \Longrightarrow
    \mathit{conf}'\neq \mathit{err}\land \exists i.\ l_i=v_n\land
    \mathit{conf}'\vDash Q_i. \]
\end{definition}

\begin{definition}[Shared Abduction]\label{def:SharedAbduction}
    If the analysis finds a non-empty anti-frame for any of the
    world's current post-conditions $Q_i$, it is added to the shared pre-condition
    $P$ and to all other current post-conditions.
    We call this \textit{shared abduction}.
    This step is motivated by the frame rule of separation logic and works as
    follows:
    If $M$ and $F$ are the solution to the biabduction query $Q_i*\boxed{M}
    \vdash L*\boxed{F}$ where $(L\statesep R)$ is the contract of the statement
    $\mathit{st}$ that is the label of the edge $(l_i,\mathit{st},l_{i'})$, then the
    world $(P\statesep Q_0^{l_0}\lor\dots \lor Q_i^{l_i}\lor\dots\lor
    Q_n^{l_n})$ gets updated to:
    \[\left(P*\color{ACMRed}M\color{black}\statesep
    \left(Q_0*\color{ACMRed}M\color{black}\right)^{l_0}\lor \dots \lor
    \left(Q_i*\color{ACMRed}M\color{black}\right)^{l_i}\lor \dots \lor
    \left(Q_n*\color{ACMRed}M\color{black}\right)^{l_n}\lor
    \color{ACMRed}\left(F*R\right)^{l_{i'}}\color{black}\right).\]
\end{definition}

\paragraph{Analysis with Worlds}

Whereas \absstate s can be split at branching statements by simply
duplicating them and adding the respective assumptions,
world splits need to be treated differently.
The two branches must share their abduced pre-conditions if the branch taken 
cannot be determined from the initial program state.
Therefore, such a branching point with condition $c$ at a location
$l_i$ with two successor locations $l_j$ and $l_k$ for a current post-condition
$Q_i$ leads to transforming the world from $(P\statesep Q_0^{l_0}\lor\dots \lor Q_i^{l_i}\lor\dots\lor
Q_n^{l_n})$ to
$(P\statesep Q_0^{l_0}\lor\dots \lor Q_i^{l_i}\lor\dots\lor Q_n^{l_n}\lor
\color{ACMRed}(Q_i \land \cond)^{l_j}\color{black}\lor\color{ACMRed}(Q_i\land \neg \cond)^{l_k}\color{black})$,
where two new post-conditions are added to the world.

In contrast, if the branch can be determined from the initial program state,
the whole world must be split into two to ensure shared abduction works correctly.
This means that the world at the branching point is exchanged with two new worlds:
\begin{gather*}
    (P\land \color{ACMRed}\cond\color{black}\statesep (Q_0\land
    \color{ACMRed}\cond\color{black})^{l_0}\lor\dots \lor (Q_i\land
    \color{ACMRed}\cond\color{black})^{l_i} \lor\dots\lor (Q_n\land
    \color{ACMRed}\cond\color{black})^{l_n}\lor \color{ACMRed}(Q_i\land
    \cond)^{l_j}\color{black}),\\
    (P\land \color{ACMRed}\neg \cond\color{black}\statesep(Q_0\land \color{ACMRed}\neg
    \cond\color{black})^{l_0}\lor\dots \lor (Q_i\land\color{ACMRed} \neg
    \cond\color{black})^{l_i}\lor\dots\lor (Q_n\land \color{ACMRed}\neg
    \cond\color{black})^{l_n}\lor \color{ACMRed}(Q_i\land \neg
    \cond)^{l_k}\color{black}).
\end{gather*}

\begin{theorem}[Loop-free Soundness with Worlds]\label[theorem]{thm:worlds}
    Let $\ana_{B,\abs}$  return only sound contracts for functions without branching.
    Further, let $\ana_{B,\abs}'$ be the biabduction-based shape analysis obtained by
    extending $\ana_{B,\abs}$  to use worlds as its \absstate s and to apply shared
    abduction.
    Then, the contracts computed by $\ana_{B,\abs}'$ for loop-free functions
    are sound.
\end{theorem}
\begin{proof}
    See \Cref{subsec:appproofshared}.
\end{proof}

\begin{example}
    With these ideas, the function in \Cref{lst:branching} can be analyzed as
    follows.
    At the start of the function, the world is equivalent to an initial
    \absstate{}:
    $$\left(\textsf{true}\statesep (x=X\land y=Y\land z=Z)_0^1\right)$$

    We denote program locations with their respective
    lines in the listing and only show the current post-conditions with the
    highest line number for each branch.
    Furthermore, we add subscripts to identify the different current
    post-conditions and worlds uniquely.
    At the outer \texttt{if-then-else}, the current post-condition is split into
    two as the branching condition cannot be related to the function arguments
    due to non-determinism.
    We further ignore the condition in the formula as it has no further
    relevance either way.
    \[\left(\textsf{true}\statesep (x=X\land y=Y\land z=Z)_0^3\lor (x=X\land y=Y\land
    z=Z)_1^{9}\right)\]
    If the analysis chooses w.l.o.g.~to first proceed with post-condition $1$, 
    it will abduce that $X.\texttt{data}$ needs to be
    allocated and share this information with the rest of the world:
    \begin{align*}
      (\color{ACMRed}X.\texttt{data}\mapsto \ell_1\color{black}\statesep & (x=X\land
      y=Y\land z=Z\formsep\color{ACMRed}X.\texttt{data}\mapsto \ell_1\color{black})_0^3
      \\\lor\ &(x=X\land y=Y\land z=Z\land
      \color{ACMRed}\mathit{return}_\mathit{nested}=\ell_1\formsep X.\texttt{data}\mapsto
      \ell_1\color{black})_1^{11})
    \end{align*}
    Thus, the current post-condition in the \texttt{then} branch now also
    requires as a pre-condition that $X.\texttt{data}$ is allocated and will not
    be unsound due to missing this information.
    The analysis can then choose to proceed with the current post-condition $0$
    and find that it can relate the branching condition with the function
    arguments.
    Therefore, the world needs to be split, as the two cases of condition are
    expressed as part of the world's pre-condition.
    To be more precise, the world is split based on whether $Y$ is $\Null$ (with omitted anchor equalities):
    \begin{align*}
        (\color{ACMRed}Y\neq\Null\color{black}\formsep &X.\texttt{data}\mapsto \ell_1\statesep\ 
        (\dots\land
        \color{ACMRed}Y\neq\Null\color{black}\formsep X.\texttt{data}\mapsto
        \ell_1)_0^{4} \\ &\lor\ (\dots\land \mathit{return}_\mathit{nested}=\ell_1\land
        \color{ACMRed}Y\neq\Null\color{black}\formsep X.\texttt{data}\mapsto
        \ell_1)_1^{11})_0,\\
        (\color{ACMRed}Y=\Null\color{black}\formsep &X.\texttt{data}\mapsto \ell_1\statesep 
        (\dots\land \color{ACMRed}Y=\Null\color{black}\formsep
        X.\texttt{data}\mapsto \ell_1)_0^6 \\ 
        &\lor\ (\dots\land
        \mathit{return}_\mathit{nested}=\ell_1\land \color{ACMRed}Y=\Null\color{black}\formsep X.\texttt{data}\mapsto
        \ell_1)_1^{11})_1.
    \end{align*}
    The two worlds will then abduce different required pre-conditions in further
    steps and finally result in the following (simplified) contracts for the function
    \texttt{nested}:
    \begin{align*}
	( y\neq\Null\formsep  & \
        x.\texttt{data}\mapsto \ell_1*y.\texttt{data}\mapsto\ell_2\statesep \\
        (&y\neq\Null\land
        \mathit{return}_\mathit{nested}=\ell_2\formsep x.\texttt{data}\mapsto
        \ell_1*y.\texttt{data}\mapsto\ell_2) \\ \lor             \ (  &
        y\neq\Null\land
        \mathit{return}_\mathit{nested}=\ell_1\formsep x.\texttt{data}\mapsto
        \ell_1*y.\texttt{data}\mapsto\ell_2)),\\
        (        y=\Null\formsep     & \ x.\texttt{data}\mapsto
        \ell_1*z.\texttt{data}\mapsto\ell_3\statesep\\ (&y=\Null\land
        \mathit{return}_\mathit{nested}=\ell_3\formsep x.\texttt{data}\mapsto
        \ell_1*z.\texttt{data}\mapsto\ell_3) \\ \lor              \ ( &
        y=\Null\land
        \mathit{return}_\mathit{nested}=\ell_1\formsep x.\texttt{data}\mapsto
        \ell_1*z.\texttt{data}\mapsto\ell_3)).
    \end{align*}
\end{example}

\subsection{Comparison with Disjunctive Domains}

It may be tempting to consider shared abduction with worlds to be just a
disjunctive closure of conjunctive formulas used commonly in various abstract interpretation approaches.
However, when using a disjunctive closure, the symbolic
execution is typically performed independently for
each disjunct, perhaps followed by attempts to join some of the disjuncts or to
prune them away using entailment checks---as done in
\citep{predator,biabd,Broom}.
In contrast, our analysis with worlds differs in that $(1)$
the worlds are, in fact, not purely disjunctive due to a single precondition
shared by all current post-conditions in a world and due to working with sets of
worlds, $(2)$ state splits either result in two new post-conditions or two new
worlds, and $(3)$ the symbolic execution from a single disjunct
can influence all other disjuncts in the same world via shared abduction.
\section{Biabductive Loop Acceleration with Shape Extrapolation}\label{sec:extrapol}

\begin{wrapfigure}[7]{r}{39mm}
    \vspace*{-1mm}
\begin{lstlisting}[language=C,label=lst:shapeexpsimp,
                   caption={Deallocating a list}]
void free_list(node *x) {
  while (x != NULL) {
    node *aux = x;
    x = x->next;
    free(aux);
  }
}\end{lstlisting}
\end{wrapfigure}

We first introduce the central steps of our technique for a simplified setting.
In this setting, loops only have loop conditions of the form $x\neq\Null$ where
$x$ is a function parameter.
Furthermore, we assume that loops do not contain branching.
We will show how to lift these restrictions in
\Cref{subsec:shape-extrapol-more}.

We will explain the steps of our biabductive loop acceleration with the help of
the example in \Cref{lst:shapeexpsimp}, which falls into the fragment of
programs allowed in the simplified setting.
The example shows a simple loop that frees a given list node by node.
As such, the expected contract would be $\left(\ls(x,\Null)\statesep
\textsf{emp}\right)$.

\subsection{Basic Biabductive Loop Acceleration}
Whereas \Cref{alg:extrapol} describes biabductive loop acceleration
on a high level, the following paragraphs describe the main steps of 
the procedure in more detail.

\begin{algorithm}[H]
    \algsetup{indent=2em}
    \caption{Biabductive loop acceleration}
    \label{alg:extrapol}
    \begin{algorithmic}
        \REQUIRE {A function $f$ consisting of a loop $l$ with body
            $\mathit{body}_l$ and exit condition $e_l$}
        \ENSURE {A sound contract $c$ for $f$ or \textsc{Failure}}
        \STATE {$s_0.\mathit{pre}\gets\textsf{true},\quad s_0.\mathit{curr}\gets\bigwedge \{x=X\mid x \in \mathit{PVar}\}$}
        \STATE {Compute $s_1 \gets \ana_{B,\abs}(\mathit{body}_l,s_0)$}
        \STATE $(\effect_\mathit{pre}*\noshapeeffect_\mathit{pre}\statesep \effect_\mathit{curr}*\noshapeeffect_\mathit{curr}) \gets \textsc{Partition}(s_1)$
        \STATE
        \STATE {$\oldsh,\newsh\gets$
            \textsc{ShapeExtrapolation}($\effect_{\mathit{pre}},\effect_{\mathit{curr}}$)}
        \STATE {Construct $s_\mathit{inv}$ from $\oldsh$, $\newsh$, $\noshapeeffect_\mathit{pre}$, and $\noshapeeffect_\mathit{curr}$}
        \STATE {$s_2\gets \ana_{B,\abs}(body_l, s_\mathit{inv})$}
        \STATE {Check that $s_2.\mathit{curr}\vdash s_\mathit{inv}.\mathit{curr}$}
        \STATE
        \STATE {Construct $s_\mathit{final}$ from $\oldsh$, $\newsh$, $\noshapeeffect_\mathit{pre}$, and $\noshapeeffect_\mathit{curr}$}
        \RETURN $c\gets s_\mathit{final}$
    \end{algorithmic}
\end{algorithm}

\paragraph{Partitioning}

Our algorithm first analyzes a single loop iteration starting from the initial
analysis state $s_0$.
If this analysis run ends in a state $s_1$, the algorithm then continues by
determining which parts of $s_1$ describe the shape of the traversed data
structure, i.e., the traversed singly-linked lists in our simplified setting.
To this end, the algorithm partitions the candidate pre-condition as well as the
current post-condition of the state $s_1$ into subformulas
$\effect_{\mathit{pre}/\mathit{curr}}$ and
$\noshapeeffect_{\mathit{pre}/\mathit{curr}}$ such that the $\effect$ formulas contain the
\emph{transformed}\footnote{ Here ``transformed'' means the changed value of the
loop variable $x$ and the shape that is described in between the old and the new
value of $x$.} parts of the state that should be related to the shape, whereas
the \emph{remaining} parts of the state are collected in the subformulas
$\noshapeeffect$, which comprise both completely unchanged predicates as well
as changed memory locations that are not part of the shape.
This separation is done for both the pre- and current post-condition of the
state $s_1$ to capture changes to the shape of the data structure.
Some more technical details of the partitioning, which are not needed now, will
be presented in \Cref{subsec:shape-extrapolation-implementation}.

\begin{example}
    For \Cref{lst:shapeexpsimp}, the analysis finds the state
    $s_1\mathdef\ \left(X.\texttt{next}\mapsto\ell_1\statesep x=\ell_1\right)$ after
    one loop iteration.
    There, the partition of $s_1$ is trivially
    $\effect_\mathit{pre}\mathdef s_1.\mathit{pre}$ and
    $\effect_\mathit{curr}\mathdef s_1.\mathit{curr}$ as this simple loop does
    not affect anything except the traversed list.
    On the other hand, the inner loop of \Cref{lst:mot2} does not change the
    shape of the traversed list but accesses and changes further parts of
    the program state.
    As a result, the partitions are $\effect_\mathit{pre}=\effect_\mathit{curr}
        \mathdef I.\texttt{next}\mapsto\ell_1*I.\texttt{elem}\mapsto\ell_2$ and $i=I$ or $i=\ell_1$, respectively, for the
    predicates that relate to the list and $\noshapeeffect_\mathit{pre}=
        o.\texttt{wgt}\mapsto \ell_3*\mathit{sum}\mapsto\ell_4$ and
    $\noshapeeffect_\mathit{curr}=o.\texttt{wgt}\mapsto
        \ell_3*\mathit{sum}\mapsto(\ell_4+\ell_3\cdot\ell_2)$ for the ones
    relating to the context.
\end{example}

\paragraph{Invariant Construction}

The main step of our procedure is the construction of the candidate loop
invariant $s_\mathit{inv}$.
For that, we first need to find an abstract description of the shape of the
involved data structures.
The abstraction must satisfy specific properties described below that are needed
to ensure the soundness of the approach.
We call this step \emph{shape extrapolation} and provide a minimum viable heuristic
implementing it in \Cref{subsec:shape-extrapolation-implementation}.
However, we stress that this algorithm can be changed as long as the properties
in \Cref{enum:extrapol} hold.

\begin{wrapfigure}[4]{r}{70mm}
    \vspace*{-1mm}
    \begin{enumerate}
        \item $\effect_\mathit{pre} \vdash \oldsh$ and $\effect_\mathit{curr} \vdash \newsh$,
        \item $\newsh \land X = x \vdash \mathsf{emp}$,
        \item $\oldsh[a/X,b/x]*\oldsh[b/X,c/x]\vdash \oldsh[a/X,c/x]$
    \end{enumerate}
    \caption{Extrapolation properties}\label[figure]{enum:extrapol}
\end{wrapfigure}

In general, we need shape extrapolation to abstract the two subformulas
$\effect_\mathit{pre}$ and $\effect_\mathit{curr}$ to list-segment shapes
$\oldsh(X,x)$ and $\newsh(X,x)$ where the parameter $X$ denotes the first node
of the list segment and $x$ denotes the current position in the segment; we omit
the parameters $X$ and $x$ and simply write $\oldsh$ and $\newsh$ when there is
no danger of confusion.
We require $\oldsh$ and $\newsh$ to satisfy the three properties given in
\Cref{enum:extrapol}.
These conditions are generalized in \Cref{enum:extrapol2}, and the proof
of soundness (see \Cref{app:sub_proof_shape_extrapol}) shows 
that they allow to establish a loop invariant.

Intuitively, Point $(1)$ of \Cref{enum:extrapol} simply ensures that $\oldsh$ and
$\newsh$ are actual abstractions of $\effect_\mathit{pre}$ and $\effect_\mathit{curr}$,
respectively.
In addition, Property $(2)$ ensures that $\newsh$ only describes the so-far traversed
and transformed part of the list.
Thereby, if $X=x$, i.e., at the start of the loop, the so-far
traversed and transformed part of the list must be empty.
Finally, $(3)$ states that consecutive list segments may
always be composed \mbox{into one list segment.}

\begin{example} In the example in \Cref{lst:shapeexpsimp}, the shape
    $\effect_\mathit{pre}$ obtained after one loop iteration
        is extrapolated (see \Cref{alg:guess}) to the formula
        $\ls(X,\ell_1)$.
        Since $x=\ell_1$, the shape $\oldsh$
        becomes $\ls(X,x)$ after normalization.
        On the other hand, $\effect_\mathit{curr}=\textsf{emp}$ does not contain any
        spatial predicates, and so the extrapolation produces $\textsf{emp}$ as
        $\newsh$, since  the transformation of the list consists in
        deleting it -- if the list was just traversed, we would obtain
        $\ls(X,x)$.
        It is easy to verify that all properties of
        \Cref{enum:extrapol} are satisfied.
\end{example}

\paragraph{Loop Invariant Checking}

In contrast to previous analyses, we explicitly
construct a candidate loop invariant from the abstract shapes $\oldsh$ and
$\newsh$ and subsequently check whether it is sound.
The candidate loop invariant $s_\mathit{inv}$ is meant to describe an
intermediate state of the loop:
\[
    s_\mathit{inv}\mathdef\ \left(\noshapeeffect_\mathit{pre}*\oldsh*\oldsh[x/X,\Null/x]\statesep
    \noshapeeffect_\mathit{curr}*\newsh*\oldsh[x/X,\Null/x]\right).
\]
The pre-condition of this state contains two (sub-)shapes $\oldsh$ and
$\oldsh[x/X,\Null/x]$ where the first describes the already traversed list segment
starting in $X$ and ending in the current value of $x$, whereas the latter
denotes the not yet traversed part of the list starting at $x$ and ending in
$\Null$.
Due to the extrapolation Property $(3)$, the two sub-shapes
together form the full extrapolated shape $\ls(X,\Null)$.
In contrast, the post-condition also takes into account the effects of the loop
on the already traversed list segment and, thus, contains $\newsh$ instead of $\oldsh$.

To prove that $s_\mathit{inv}$ is a loop invariant, the analysis also needs to
check whether the post-condition's memory footprint is sufficient for another
loop iteration and whether it also holds after this iteration.
This is proven by analyzing another loop iteration starting from $s_\mathit{inv}$
in which the abduction of new pre-condition predicates is disallowed, thus
forcing the analysis to fail if the shapes describe an insufficient memory
footprint.

Suppose the invariant checking step successfully finishes the symbolic execution of the loop body in some state $s_2$.
In that case, this implies that the loop body can be safely executed from the state $s_\mathit{inv}$.
Next, we check whether $s_2.\mathit{curr}\vdash s_\mathit{inv}.\mathit{curr}$,
i.e. whether $s_\mathit{inv}$ is actually a loop invariant.
If the check succeeds, the shapes are sound for all loop iterations, and the
loop acceleration procedure can continue with the final step.
We specifically note that our approach requires the analysis of just a single
loop iteration plus another invariant check iteration to filter out unsound
extrapolation results in most cases, whereas previous work could not
do so without a full second analysis phase.

\begin{example}
    It trivially holds that the following state is invariant for the loop in
    \Cref{lst:shapeexpsimp}, i.e., it is sound before and after each loop
    iteration.
    $$s_\mathit{inv}\mathdef \left(\ls(X,\ell_1) * \ls(\ell_1,\Null)\statesep x=\ell_1 \formsep \ls(\ell_1,\Null)\right)$$
    As the state depicts the program at an arbitrary point of the iteration,
    it contains both the already traversed shape $\ls(X,\ell_1)$ in the
    pre-condition (which has been freed in the post-condition) and the unchanged,
    still-to-traverse shape $\ls(\ell_1,\Null)$.
\end{example}

\paragraph{Finalizing}

Lastly, the loop analysis is finalized by constructing the final state reached
after finishing the loop from the shapes $\oldsh$ and $\newsh$ as
\[
    s_\mathit{final}\mathdef
\left(\noshapeeffect_\mathit{pre}*\oldsh[\Null/x]\statesep
x=\Null*\noshapeeffect_\mathit{curr}*\newsh*\oldsh[x/X,\Null/x]\right)
\]
This state is simply obtained from $s_\mathit{inv}$ by adding the negated loop
condition $x=\Null$ and using extrapolation Property $(3)$ to
simplify the pre-condition.
If the extrapolated shape additionally satisfies the property
$\oldsh(x/X,\Null/x) \land x = \Null \vdash \mathsf{emp}$ (we call this 
Property (3.5)), which is natural for
traversing linked lists until the end, the final state can be simplified even
further.

\begin{example} For the loop in \Cref{lst:shapeexpsimp}, the freed list is 
    represented by the shape $\newsh$ being empty, making the final state 
    $s_\mathit{final}\mathdef  \left(\ls(X,\Null)\statesep x=\Null*\ls(x,\Null)\right)$.
    Since the list segment to $\Null$ satisfies Property $(3.5)$, we obtain the expected
    final state.
    In addition, this state is also the contract of the function
    \texttt{free\_list}, and thus the analysis reaches its end for this function.
\end{example}

\subsection{Lifting Restrictions on Biabductive Loop Acceleration}
\label{subsec:shape-extrapol-more}

We now explore how the restrictions introduced above can be lifted
to make biabductive loop acceleration more applicable
in practice.
We write $\overline{x}$ for an ordered list of elements $x_i$ with
$0\leq i\leq n$ for some $n$.
We denote by $f(\overline{x})$ the ordered list $\overline{y}$ where
$y_i=f(x_i)$ for $0\leq i\leq n$.

\begin{wrapfigure}[6]{l}{40mm}
    \vspace{-1mm}
        \begin{lstlisting}[language=C,label=lst:shapeexpsimp2,caption={Cyclic/to-null lists}]
void either_list(node *x) {
  node *head = x;
  while (x != NULL
    && x->next != head)
  {...} 
}\end{lstlisting}
\end{wrapfigure}

\paragraph{Extension: General Loop Conditions}

The first restriction that we lift concerns the loop condition.
We assume that the loop condition $e_l$ is of the form $e_l=\bigwedge_i e_i$
with single atomic conjuncts $e_i$ of arbitrary form.
An example of such a loop can be seen in \Cref{lst:shapeexpsimp2}, which
handles both cyclic and null-terminated lists equally.
Handling such a more general loop condition requires further adjustments to the
loop acceleration procedure.
To be able to express multiple exit conditions that relate to multiple different
variables, the algorithm needs to be able to describe the traversed shape
relative to these variables.
Hence, the shapes $\oldsh$ and $\newsh$ are now parameterized over all program
variables changed throughout the loop---namely, all variables $x$ for which
$s_1.\mathit{curr}.\Pi \nvdash x= X$.
We call the set of these variables $\changed$ and re-define the $\oldsh$
and $\newsh$ shapes as $\oldsh(\overline{X},\overline{x})$ and
$\newsh(\overline{X},\overline{x})$, respectively, where $\overline{x}$ is the
ordered list of the variables from $\changed$ that occur in $\oldsh$ and
$\newsh$, and $\overline{X}$ is the ordered list of the corresponding anchor
variables.
Below, we will use $\oldsh(\overline{a},\overline{b})$ to denote the predicate
$\oldsh(\overline{X},\overline{x})[\overline{a}/\overline{X},\overline{b}/\overline{x}]$,
i.e., the predicate obtained from $\oldsh(\overline{X},\overline{x})$ by
simultaneously substituting the variables $\overline{X}$ with $\overline{a}$,
and $\overline{x}$ with $\overline{b}$ (we will use the same notation for
$\newsh$).
Note that for lists, this is equal to setting the two parameters of the list
segment predicate to $a$ and $b$, respectively.

\begin{wrapfigure}[6]{r}{55mm}
    \begin{enumerate}
        \item $\effect_\mathit{pre} \vdash \oldsh(\overline{X},\overline{x}) \land
            \effect_\mathit{curr} \vdash \newsh(\overline{X},\overline{x})$,
        \item $\newsh(\overline{X},\overline{x}) \land \bigwedge_{x\in\changed}X = x \vdash \mathsf{emp}$,
        \item $\oldsh(\overline{a},\overline{b})*\oldsh(\overline{b},\overline{c})\vdash
                  \oldsh(\overline{a},\overline{c})$.
    \end{enumerate}
    \caption{General extrapolation properties}\label[figure]{enum:extrapol2}
\end{wrapfigure}

With this notation, we re-define the properties of extrapolation to consider the
new parameters in \Cref{enum:extrapol2}.
We further define a mapping $\exit{x}$ of variables $x\in\changed$ to
the values they can have at a loop exit.
As these values can be challenging to determine from the loop condition alone, we
restrict the map to hold only logically constant values, i.e., $\Null$ or other
program variables outside of $\changed$ (as their values stay constant
throughout the loop), and define the other entries to map to fresh logical variables instead.

Furthermore, the post-condition of the final state has to encode that any of the
loop conditions can be unsatisfied for the program to leave the loop.
This is done by taking the disjunction of the previous final state
post-condition combined with one dissatisfied loop condition (note that the
disjunction represents the world's current post-conditions):
\[
    s_\mathit{final}.\mathit{curr} \mathdef \bigvee_i \left(\neg
    e_i*\noshapeeffect_\mathit{curr}*\newsh(\overline{X},\overline{x})*\oldsh(\overline{x},
        \exit{\overline{x}})\right).
\]

\paragraph{Extension: Branching Loop Body}

Branching in loop bodies can be handled by collecting all states $s_1$ after the
first loop iteration analysis, extrapolating their shapes, and combining them if
possible into one compound shape via a join operation akin to the ones described
in \citep{predator} or \citep{RysavyLukas2024Jofb}.
Because such an operation is mostly orthogonal to the central ideas of shape
extrapolation, we refer to the literature for more details.

\paragraph{Extension: Overlapping Shape Changes}

The extension to allow for more general loop conditions can lead to problems
with shape extrapolation if the involved shapes overlap, i.e., if the new and
old memory locations to which program variables point to are the same.
This can, e.g., happen if a list is reversed (see \Cref{app:examples:listrev}).
Such cases can be detected if the new value of a program variable in 
$\changed$ is the anchor of another variable.
In the example of list reversal, the program variable tracking the reversed list
will be set to the initial value of the original list as that list's first node
becomes the last node in the reversed one.
Such an overlap would cause problems in the implementation of
\textsc{ShapeExtrapolation} presented as \Cref{alg:guess}.
To circumvent this problem, the analysis symbolically executes further additional
loop iterations to find a program state in which there is no overlap anymore, and only then
performs the extrapolation.

\paragraph{Extension: Further Loop Effects}

As depicted in \Cref{lst:mot2}, loops can not only traverse data structures but
also change the program state in arbitrary ways.
In such cases, the candidate invariant $s_\mathit{inv}$ might not be an actual
invariant, i.e. $s_2.\mathit{curr}\nvdash s_\mathit{inv}.\mathit{curr}$.
To handle such cases, we apply a join in the corresponding pure value domain of the analysis.
In the simplest case, this step exchanges the values of variables and memory
locations that are the cause of $s_2\nvdash s_\mathit{inv}$ with the value $\top$.
In the example \Cref{lst:mot2}, the value stored at \texttt{sum} after the first
iteration is $\ell_\mathit{sum}+\ell_1 \cdot\ell_w$, resulting in the points-to
predicate $\mathit{sum}\mapsto \ell_\mathit{sum}+\ell_1 \cdot\ell_w$ being a
part of $s_\mathit{inv}$.
After the second iteration, the predicate changes to $\mathit{sum}\mapsto
    \ell_\mathit{sum}+\ell_1\cdot\ell_w+\ell_2\cdot\ell_w$, which does not
entail its counterpart in $s_\mathit{inv}$.
However, by joining the two values of the memory location to
$\mathit{sum}\mapsto\top$, the entailment is ensured.
The same problem actually occurs if the initial value of a variable does not entail
its representation in the invariant, e.g., because it is set to a constant in
the loop (see \Cref{app:examples:inv}).
In this case, we also need to abstract the variable's values in
$s_\mathit{inv}.\mathit{curr}$ to $\top$, thus guaranteeing that $s_\mathit{inv}$
also holds before the first iteration.

\begin{theorem}[Soundness of Shape Extrapolation]\label{thm:loops}
    Let $\ana_{B,\abs}$  compute only sound contracts for loop-free functions.
    If \Cref{alg:extrapol} uses $\ana_{B,\abs}$, then \Cref{alg:extrapol} with all
    extensions described in this section applied to a loop
    $l$ either fails or returns a contract $(P,Q)$ such that $\{P\}\ l\ \{Q\}$.
\end{theorem}
\begin{proof}
    See \Cref{app:sub_proof_shape_extrapol}.
\end{proof}

\subsection{Shape Extrapolation}
\label{subsec:shape-extrapolation-implementation}

We now propose a concrete shape extrapolation procedure based on the principles presented above.
This procedure is supposed to be the easiest possible heuristic that suffices to
find reasonable loop invaraiants.
To this end, it follows the original idea of obtaining inductive shapes through
abstraction, but in a ``smarter'' way.

\paragraph{Partition}

The initial partitioning is one of the most crucial steps for our shape
extrapolation procedure.
The $\effect$ formulas are built by collecting all predicates that describe the
shape traversed, i.e., the shape between the anchors and the new values of the
variables in $\changed$.
These shapes contain all transitively reachable predicates, where reachability
is defined spatially.
Thereby, a predicate is reachable if there exists a sequence of points-to and
list predicates that pairwise overlap in their source/drain variables, modulo
variable equalities.
For example, if $x\in\changed$, then $\ell_1=\ell_4\formsep
X\mapsto\ell_1*ls (\ell_1,\ell_2)*\ell_4.\mathit{data}\mapsto\ell_3$ contains
only predicates reachable from the anchor $X$.

\paragraph{Transformation Map}

In addition to the partitioning of variables, our concrete shape extrapolation
algorithm also needs to know the new value of the variables in $\changed$.
We encapsulate this information in the \textit{transformation map}
$\transfm$ which maps $x\in \changed$ to $\ell_x\in\mathit{LVar}$
such that $s_1.\mathit{curr}.\Pi_P\vdash x=\ell_x$.
Recall that, due to the normal form of \emph{SL}, every
program variable only occurs in a single equality such as $x=\ell_1$, and so
$\transfm$ can be computed by simply comparing their values before and
after the loop.

\paragraph{Shape Extrapolation}

\Cref{alg:guess} gives a detailed description of our concrete shape
extrapolation procedure.
It computes $\oldsh$ (and $\newsh$) by first extrapolating the corresponding
$\effect_i$ into two copies $\effect^1_i$ and $\effect^2_i$.
These two copies are supposed to represent the shape accessed by two consecutive
loop iterations via an intermediate, fresh auxiliary location.
The resulting formulas are then combined via separating conjunctions and
abstracted by the abstraction function $\abs$ to form the abstract shapes
$\expast$, which in turn get parameterized by renaming schemas to the
\mbox{final $\oldsh$ and $\newsh$.}
The use of two copies is a heuristic that has proven to be reliable in making
the abstraction find better abstract shapes.
Note that, in \Cref{alg:guess}, we omit additional renamings of logical
variables for clarity, as these only help to guide the abstraction but do not
affect the resulting shapes any further.

\begin{algorithm}[t]
    \algsetup{indent=2em}
    \caption{ShapeExtrapolation}
    \label{alg:guess}
    \begin{algorithmic}
        \REQUIRE {$\effect_\mathit{pre}$, $\effect_\mathit{curr}$}
        \ENSURE {Parametric extrapolated shapes $\oldsh$ and $\newsh$}
        \STATE {Compute the maps $\changed$ and $\transfm$ as described
            in the text}
        \STATE
        \COMMENT{$i\in \{\mathit{pre},\mathit{curr}\}$, $x'$ fresh}
        \STATE {$\effect^1_i\gets \effect_i[x'/\transfm(x)\mid x\in \changed]$}
        \STATE {$\effect^2_i\gets \effect_i[x'/X\mid x\in\changed]$}
        \STATE {$\expast_i \gets \abs(\effect^1_i*\effect^2_i)$}
        \STATE
        \STATE {$\oldsh(\overline{x}_1,\overline{x}_2) \gets \expast_\mathit{pre}\left[\overline{x}_1/\overline{X},\overline{x}_2/\transfm(\overline{x})\right]$}
        \STATE {$\newsh(\overline{x}_1,\overline{x}_2) \gets \expast_\mathit{curr}\left[\overline{x}_1/\overline{X},\overline{x}_2/\transfm(\overline{x})\right]$}
        \RETURN $\oldsh,\newsh$
    \end{algorithmic}
\end{algorithm}

\begin{example}
    In the example from \Cref{lst:mot2}, the inner loop can be extrapolated as
    follows:
    The procedure takes the effect $\effect_\mathit{pre}\mathdef I.\texttt{elem}
        \mapsto\ell_1*I.\texttt{next}\mapsto \ell_2$ from the transformation and
    introduces the two auxiliary formulas $\effect_\mathit{pre}^1\mathdef I.
        \texttt{elem}\mapsto\ell_1*I.\texttt{next}\mapsto i'$ and $\effect_
        \mathit{pre}^2\mathdef i'.\texttt{elem}\mapsto\ell_1*i'.\texttt{next}\mapsto
        \ell_2$ where $i'$ is the auxiliary location representing the
    intermediate value of $i$.
    From these formulas, the abstraction then finds the abstract shape
    $\expast_\mathit{pre} \mathdef \ls(I,i)$.
    This abstracted shape is then the basis for the extrapolated shape $\oldsh$.
    Similarly, $\newsh$ is computed to be $\ls(I,i)$, too.
\end{example}

\subsection{Limitations}

We note that shape extrapolation is a heuristic, which is sound (see
\Cref{thm:loops}) but inherently incomplete.
The extrapolation step can fail if the partitioned information does not
suffice to find a reasonable shape, e.g., if the abstraction function cannot
find a fitting inductive shape predicate.
However, since our shape extrapolation procedure imitates the loop acceleration
procedure of the original analysis, it is guaranteed to be applicable for at 
least the same programs but in a fundamentally sound (and oftentimes faster) way.
Furthermore, we note that shape extrapolation is currently limited to list-like
data structures that are traversed linearly.
List manipulation is, however, by far the most frequent data structure pattern
in low-level code, and so we have focused our efforts on this kind of data
structures, in accordance with prior work~\citep{biabd,Broom}.
Nonetheless, we believe that shape extrapolation, which is based on the
intuition of locality, can be extended toward tree-like data structures in
future work.


\section{Implementation and Experimental Evaluation}\label{sec:impl}

\subsection{Prototype Implementation}

We have implemented our techniques as a proof-of-concept in the prototype analyzer 
Broom \citep{Broom} written in OCaml and call the resulting tool \newname.
It is available as an artifact on Zenodo \citep{brush_artifact}.
The original Broom implements a biabduction-based shape analysis with a focus on low-level
primitives and byte-precise memory management and is sound
for functions without branching \citep[see Theorem 3]{Broom}.
However, since Broom is also still a prototype that focuses more on exact
handling of complex memory manipulation than on scalability, neither Broom nor
\newname\:are able to handle large-scale code bases yet.
Our new techniques, especially shape extrapolation, improve
scalability, but \newname\:still shares most of its code with Broom and is thus
not as mature as industrial-strength tools such as Infer.

\subsection{Implementation Limitations}
As \newname\:is largely based on the source code of Broom and does not differ much
from it apart from our new techniques, they share mostly the same limitations.
On the one hand, neither tool supports recursive functions.
Similarly, they can handle neither stack allocations nor switch-case statements.
On the other hand, the logic both tools are based on contains only inductive
predicates for linked lists with parameters describing the shape of single nodes.
Therefore, the tools can, in general, not analyze programs containing other
inductive data structures.

Furthermore, we remark that the running times of Broom and \newname\:are much
higher than for comparable tools, which in part is due to the need for precise
pointer arithmetic.
This precision is achieved, among other things, by calling an SMT solver, which is more costly
for simpler cases than using native solvers, such as in Infer.

\subsection{Case Study}

We have conducted experiments with two research questions in mind: $(1)$ whether
\newname\:can handle new use cases that existing tools cannot handle; $(2)$
whether \newname\:is also at least as efficient as Broom or would even
improve scalability.

\paragraph{Qualitative Experiments}

To answer research question $(1)$, we ran all four analyzers on selected
examples that are either presented in \citep[Table~1]{Broom}, are a part of the
test suite for Broom, or are hand-crafted test cases for shared abduction and
shape extrapolation.
The results of our case study are depicted in \Cref{tab:handle}.
All test files are included in the accompanying material.
We primarily investigated whether the analyzers found the expected bugs and sound
contracts or whether they report other spurious errors.

We note that the biabduction-based shape analysis that Infer was based on is deprecated nowadays, and Infer's focus has shifted from over-approximation to
under-approx-imation (see \citep{isl,ISLX}).
Due to this, we not only compare \newname\:with the release
v1.1.0\footnote{Available at
\url{https://github.com/facebook/infer/releases/tag/v1.1.0}.} which was also
used for comparison in \citep{Broom}, but also with its predecessor tool
Abductor.\footnote{Available at
\url{http://www0.cs.ucl.ac.uk/staff/p.ohearn/abductor.html}.}
We excluded the second-order biabduction tool S2 from our experiments since it
is quite limited and cannot be applied to most of our benchmark
programs.\footnote{Of the 73 programs in \Cref{tab:handle}, the tool reported
internal errors for 56 cases while causing segmentation faults for nine further
cases. If we only compare the 52 programs without loops, it fails for 43
instances and causes segmentation faults in 4 further cases. The internal errors
range from unsupported language features such as pointer arithmetic (3 programs)
to linker errors with unknown symbols (13 cases/6 cases without loops) and
unsupported type casts (35/30 cases). All of these cases are correctly handled
and accepted by standard C compilers as utilized as frontends by Broom and
\newname.}

\begin{table}[t]
  \caption{Examples handled correctly (\checkmark) and incorrectly ($\times$)
    by the analyzers}
    \Description{A table of examples handled by different tools (Broom, Infer, \newname)
    and which tool handles which example in a single analysis run. Our implementation
    handles all of the cases that fall in the supported fragment (modulo bugs and
    missing features).}
    \centering
  \begin{tabular}{|l|c|c|c|c||c|}
    \hline
    Class of inputs                 & $\#$ of test cases & Broom                   & Infer                   & Abductor                & \newname                \\
    \hline
    \citep[table 1]{Broom}          & 10                 & $10 \checkmark/0\times$ & $0\checkmark/10\times$  & $0\checkmark/10\times$  & $10\checkmark/0\times$  \\
    \hline
    tests from \texttt{broom/tests} & 47                 & $47 \checkmark/0\times$ & $34\checkmark/14\times$ & $12\checkmark/35\times$ & $47 \checkmark/0\times$ \\
    \hline
    \texttt{*\_branches.c}          & 2                  & $0\checkmark/2\times$   & $0\checkmark/2\times$   & $0\checkmark/2\times$   & $2\checkmark/0\times$   \\
    \texttt{nested\_*.c}            & 3                  & $1\checkmark/2\times$   & $3\checkmark/0\times$   & $3\checkmark/0\times$   & $3\checkmark/0\times$   \\
    \texttt{motivation*.c}          & 3                  & $0\checkmark/3\times$   & $0\checkmark/3\times$   & $0\checkmark/3\times$   & $3\checkmark/0\times$   \\
    \texttt{sll*.c}                 & 3                  & $3\checkmark/0\times$   & $0\checkmark/3\times$   & $0\checkmark/3\times$   & $3\checkmark/0\times$   \\
    other                           & 5                  & $3\checkmark/2\times$   & $0\checkmark/5\times$   & $3\checkmark/2\times$   & $5\checkmark/0\times$   \\
    \hline\hline
    overall                         & 73                 & $64\checkmark/9\times$  & $47\checkmark/26\times$ & $18\checkmark/55\times$ & $73\checkmark/0\times$  \\
    \hline
  \end{tabular}
  \label{tab:handle}
\end{table}

The ten tests from \citep[Table 1]{Broom} are program fragments of 30--200 LOC.
Each test consists of a set of library functions, including creation of a
linked-list, insertion and deletion of an element from the list.
Moreover, 8 of the tests also contain a top-level test harness performing
a concrete manipulation of the particular list.
There are three types of lists: (i) circular doubly-linked lists,
(ii) linux-lists taken from the Linux kernel, and (iii) intrusive lists.\footnote{Described by \citep{intrusive} and implemented in \url{https://github.com/robbiev/coh-linkedlist}.}
The 43 tests from \texttt{broom/tests} are regression tests of Broom.
Each one is usually up to 10 LOC and tests the analysis of a particular
kind of statement.

The newly added hand-crafted use cases are small-scale programs (10--70
LOC), which are, however, rather challenging for the existing analyzers.
In particular, the \texttt{*\_branches} test cases contain multiple cases of
non-determinable branching, which can lead to the unsoundness described in
\Cref{subsec:sharemot}.
The \texttt{nested\_*} programs contain multiple examples of nested loops and
nested lists.
The \texttt{sll\_*} test cases contain whole programs that create, iterate
and destroy singly-linked lists.
The \texttt{motivation} programs are as described in \Cref{sec:mot}.
Lastly, the other test cases contain programs with more complex list allocation,
deallocation, and transformation.

We specifically note that we have not used common benchmark sets such as the SV-COMP
memory-safety benchmark, as these consist primarily of closed programs and focus
on data structures that cannot be described by the logic of Broom and \newname.
Thus, these benchmark sets lie outside the scope of this work, and we have instead
used test cases that allow us to evaluate our research questions explicitly.

In the table, we use \checkmark to denote that
at least one (sound) contract was computed for each function within the particular example and
that the expected errors were reported without false positives.
On the other hand, we use $\times$ to denote that either no contract could be computed (for a
function that would have a sound contract) or the respective tool reported a false positive.

\emph{We conclude that shared abduction and shape extrapolation enable \newname\:to
  work for strictly more programs than Broom or any of the other tools.
  This is especially important, as these small but challenging test cases are
  mostly based on realistic iteration patterns that can be found in code bases
  such as the Linux kernel.}

\paragraph{Quantitative Experiments}

\begin{figure}[ht]
  \centering
  \begin{tikzpicture}
    \begin{loglogaxis}[xlabel=Runtime of Brush in sec,ylabel=Runtime of Broom in sec,enlargelimits=false,legend pos=south east]
        \addplot+[
            only marks,
            color=ACMRed,
            mark=star,
            mark size=2.9pt]
        table [x=new, y=old, col sep=comma] {result-loop.csv};
        \addplot+[
            only marks,
            color=ACMRed,
            mark=star,
            mark size=2.9pt]
        table [x=new, y=old, col sep=comma] {result-branching.csv};
        \addplot+[
            only marks,
            color=ACMRed,
            mark=star,
            mark size=2.9pt]
        table [x=new, y=old, col sep=comma] {result-other.csv};
        
        \addplot+[draw=black,mark=] coordinates
            {(0.02,0.02) (1,1) (200,200)};
        \addplot+[draw=lightgray,mark=,solid] coordinates {(0.02,0.03) (1,1.5) (200,300)};
        \addplot+[draw=lightgray,mark=,solid] coordinates {(0.02,0.01) (1,0.5) (200,100)};
        \legend{Test Case,,,Equal time, $\pm 50\%$}
    \end{loglogaxis}
\end{tikzpicture}
  \caption{Runtime of Broom and \newname\:in sec. for tests from \Cref{tab:handle} that both tools handled}
  \Description{A scatter plot of the runtime of Broom and \newname\:for all cases in
    \Cref{tab:handle} that both tools can handle. Most data points are above the
    diagonal, denoting equal runtime, showing that Broom takes longer in these cases.}
  \label{fig:scatter}
\end{figure}

In another series of experiments, we also evaluated the runtime of
\newname\:versus Broom\footnote{We used the commit
  \broomv\:found at
  \url{https://pajda.fit.vutbr.cz/rogalew/broom} with small additional bug fixes.} on test cases that both
tools can handle.
Some test cases were split to survey the runtime of single, interesting
functions without their calling context.
All experiments were run ten times on an Intel Core i7-1260P CPU with 32GiB RAM,
and we took the mean over all runs.
A plot of the results can be found in \Cref{fig:scatter}, the raw data is 
displayed in \Cref{app:data}.
The overall means for programs with loops are 1.85s for \newname\:to 38.2s for Broom,
whereas the means for loop-free branching programs are 22.3s for \newname\:and
24.8s for Broom.
The overall means are 7.24s for \newname\:to 14.0s for Broom.
These numbers show that \newname\:provides a significant speedup over
Broom.
For the examples with branching but without loops, we at least
find a tendency towards faster runtime for \newname.

We directly relate the time improvements of \newname\:with the number loop
iterations analyzed.
As seen in \Cref{tab:iter}, \newname\:only requires a fraction of loop iterations due
to shape extrapolation for all examples from \Cref{tab:handle}, that both Broom
and \newname\:can handle and which contain loops.
The mean over all examples is 3.55 iterations for \newname\:and 11 for Broom.
We note that in most cases, \newname\:requires exactly two iterations per loop,
which corresponds to the initial analysis and the invariant checking iteration
as described in \Cref{sec:extrapol}.
Only the case \texttt{reversal} (see \Cref{app:examples:listrev}) requires 3
iterations, as it reverses a list such that the new and old shapes partially
overlap.
As described in the paragraph about \emph{Extension: Overlapping Shape Changes},
the analysis requires an additional iteration to reach a state in which this
overlap has been removed.

\begin{table}[t]
  \caption{Loop iterations analyzed by the analyzers (\newname/Broom) until a
    fixed point was reached, as well as how many loops were
    present in the programs (in brackets)}
    \Description{A table of examples and how many loop iterations were
    analyzed by Broom and Brush, respectively, as well as how many loops were
    present in the programs (in brackets).}
    \centering
  \begin{tabular}{|c|c|c|c|}
    \hline
    \textttt{copy\_alloc}    & \textttt{dll-as-sll-traverse}  & \textttt{reversal} & \textttt{nested\_lists2} \\
    \hline
    2/6 (1)                 & 2/6 (1)                       & 3/8 (1)                                & 6/14 (3)                \\
    \hline
    \hline
    \textttt{sll-fst-shared} & \textttt{sll-fst-shared-alloc} & \textttt{sll-fst-shared-iter}           & \textttt{sll}            \\
    \hline
    4/12 (2)                & 2/8 (1)                       & 2/6 (1)                                & 6/20 (3)                \\
    \hline
    \hline
    \textttt{sll-alloc}      & \textttt{sll-shared-sll-after} & \textttt{sll-shared-sll-after-alloc}    & \textbf{overall}        \\
    \hline
    2/6 (1)                 & 6/21 (3)                      & 4/14 (2)                               & 39/121 (19)             \\
    \hline
  \end{tabular}
  \label{tab:iter}
\end{table}

\emph{We conclude that shared abduction improves both the precision and the 
performance of the analysis compared to Broom, while biabductive loop 
acceleration considerably improves the performance of the analysis for the benchmarks.}



\section{Related Work}\label{sec:relatedwork}

\paragraph{Biabduction-based Shape Analysis}

Our work builds on biabduction-based shape analysis introduced by \citet{biabd,biabd_conference} and later 
implemented in Infer \citep{Infer}.
Our new techniques avoid the unsoundness issues of the first phase of
the analysis and constitute significant theoretical and
practical advancements as demonstrated in \Cref{sec:worlds,sec:extrapol,sec:impl}.

The approach of \citet{biabd,biabd_conference} was extended in the Broom
analyzer \citep{Broom} by ways of handling low-level primitives and
byte-precise memory handling.
Since these extensions are orthogonal to the problems of unsoundness, our
techniques are equally effective for improving Broom as depicted in
\Cref{sec:impl}.

Another related work is the \emph{second-order biabduction} by \citet{secondorder}.
Their approach does not consider a fixed class of inductive predicates,
but discovers them as part of the analysis by instantiating
second-order variables with a technique called \textit{shape inference}.
The analysis first collects the unknown predicates with corresponding relational
assumptions and synthesizes fitting shapes in a second step.
Their computation method ensures that these shapes make the resulting contracts
sound.
Albeit this makes their technique similar to shape extrapolation,
the approach of \citet{secondorder} can handle more complex shapes of
dynamic data structures.
On the other hand, it requires solving the complex predicate inference problem
for which their tool uses a relatively simple algorithm.
In our experience, this algorithm can easily fail even for programs with simple
data structures if they are not compatible with the shape inference procedure.
Furthermore, their implementation is rather limited, as described in \Cref{sec:impl}.

Very recently, \citet{quiver} combined biabductive reasoning with auto-active,
foundational program verification.
Their tool Quiver takes C programs and specification sketches as annotations as
input, translates them into a representation in the Caesium C semantics
\citep{refinedc}, and finally infers and proves full function specifications in
the proof assistant Coq.
The central reasoning mechanism of Quiver is called \emph{abductive deductive
verification}, which is closely related to biabduction.
As Quiver does not only work with predicates for memory safety but also with a
refinement type system for C, it has a broader focus than our work.
However, Quiver requires the user to provide specification sketches, refinement types, and loop invariants, while our work focuses on fully automated, biabductive shape analysis.
\paragraph{Other Shape Analyses.}
There are many different shape analysis methods not based on
biabduction in literature.
Of these, the Predator analyzer \citep{predator} based on \emph{symbolic memory
graphs} is quite
successful with regard to the Competition on Software Verification (SV-COMP),
see \citep{SVCOMP20,SVCOMP24}.
Their approach handles abstraction, entailment, and state pruning as special
cases of a general graph joining procedure.
They focus on closed programs, and their approach uses function summaries that are
computed in a top-down fashion, following the call tree (whereas biabduction-based
shape analysis works from the bottom up).
This top-down fashion requires a re-analysis of functions for different
contexts but can ignore irrelevant code paths.
As Predator implements a classic forward analysis that does not
compute contracts with pre-conditions, Predator circumvents the problem of
unsoundness.
On the other hand, Predator only works on closed programs and can thus not be
used for modular and incremental analysis.

Another approach to shape analysis has been recently introduced by \citet{transformer,relational_shape_abstract_domain}.
They utilize \emph{transformers} to describe the effects of functions and
compute these with regard to the calling context in a top-down fashion, as in
Predator.
The transformer-based analysis is built around a transformer abstract domain and
is based on abstract interpretation.
We note that the use of transformers has partially motivated the inner workings
of our shape analysis procedure.
The authors noted that biabduction might be applicable for transformers
as well, and we strengthen this point by showing how our shape
extrapolation procedure combines both ideas to some degree.

We also like to mention the work by \citet{enea:sas13}, where the authors consider
overlaid data structures.
Their technique is based on a fragment of separation logic
that differentiates per object and per field separation.
Note that the separation logic fragment we base our work on uses per field
separation, while per-object separation is only enforced implicitly.

\paragraph{Other Related Analyses.}
Recent work about \emph{Incorrectness (Separation) Logic} (ISL)
\citep{isl,cisl,ISLX} has introduced a different approach to program analysis.
This line of work does not focus on verifying the absence of memory bugs in an 
over-approximating way but instead tries to find bugs in an under-approximative way.
Although the bug-finding ability of incorrectness logics makes them very useful
in practice, over-approximating analyses are still relevant for certification
and low-level systems software.

Lastly, the emergence of incorrectness logic has also motivated the development
of combined logic systems that inherit the benefits of both over- and
under-approximating logics.
Recent work in this direction includes \textit{Exact Separation Logic} \citep{exactsl} as well as \textit{Outcome (Separation) Logic} \citep{outcome,outcomesl}.
In particular, \emph{tri-abduction}, introduced in \citep{outcomesl}, offers an alternative solution to the branching problem we address with shared abduction.
By solving the abduction problem for the pre-conditions required by both branches, the tri-abduction operation can potentially compute better contracts than the (greedy) techniques proposed in this paper, which will first solve the abduction problem for one branch and then for the other.
This increased precision, however, comes with the burden of implementing a new operator, whereas we can simply lift existing implementations of biabduction operators to shared abduction.
We believe that the increased precision of the triabduction operation is rarely needed, and thus the more lightweight solution of shared abduction is preferable in practice.
For our experiments, the precision of shared abduction was sufficient.
At the same time, the approach by \citet{outcomesl} has not yet been 
implemented into a tool that could be used for experimental comparison.

\section{Conclusion and Future Work}

This work introduces the two novel techniques of shared abduction and
biabductive loop acceleration with shape extrapolation.
We provide soundness proofs for both techniques and implement them in our prototype analyzer \newname, which is based on the state-of-the-art analyzer Broom.
We experimentally demonstrate that these techniques enable our biabduction-based shape analysis to find sound contracts in a single analysis phase.
In particular, we show that shared abduction and shape extrapolation enable \newname\ to analyze strictly more
programs than Broom or any of the Infer versions, and to considerably improve the performance compared to Broom.

While our work is limited to non-recursive programs, we believe that shape
extrapolation can easily be extended to recursive programs (e.g., tree traversals)
and that this direction constitutes an exciting avenue for future work.
We also hope to incorporate techniques that track the content of data structures, i.e., the data values stored in the data structure.
Specifically, we would like to enrich the logic and biabduction procedure to track data values, e.g., for tracking the value of sum in the example of
\Cref{lst:mot2}.
For this, we plan to take inspiration from prior work that combines shape domains and data domains, such as the product domain studied in \citep{10.1007/978-3-031-44245-2_15}.
Another interesting direction for future research is whether the idea underlying shape extrapolation has application in the synthesis of heap-manipulating programs, e.g., as studied in \citep{journals/pacmpl/PolikarpovaS19}.


\bibliographystyle{ACM-Reference-Format}
\bibliography{onephase.bib}


\begin{thebibliography}{32}


\ifx \showCODEN    \undefined \def \showCODEN     #1{\unskip}     \fi
\ifx \showDOI      \undefined \def \showDOI       #1{#1}\fi
\ifx \showISBNx    \undefined \def \showISBNx     #1{\unskip}     \fi
\ifx \showISBNxiii \undefined \def \showISBNxiii  #1{\unskip}     \fi
\ifx \showISSN     \undefined \def \showISSN      #1{\unskip}     \fi
\ifx \showLCCN     \undefined \def \showLCCN      #1{\unskip}     \fi
\ifx \shownote     \undefined \def \shownote      #1{#1}          \fi
\ifx \showarticletitle \undefined \def \showarticletitle #1{#1}   \fi
\ifx \showURL      \undefined \def \showURL       {\relax}        \fi
\providecommand\bibfield[2]{#2}
\providecommand\bibinfo[2]{#2}
\providecommand\natexlab[1]{#1}
\providecommand\showeprint[2][]{arXiv:#2}

\bibitem[Berdine et~al\mbox{.}(2005)]%
        {decidable}
\bibfield{author}{\bibinfo{person}{Josh Berdine}, \bibinfo{person}{Cristiano Calcagno}, {and} \bibinfo{person}{Peter~W. O'Hearn}.} \bibinfo{year}{2005}\natexlab{}.
\newblock \showarticletitle{A Decidable Fragment of Separation Logic}. In \bibinfo{booktitle}{\emph{FSTTCS}} \emph{(\bibinfo{series}{LNCS})}, \bibfield{editor}{\bibinfo{person}{Kamal Lodaya} {and} \bibinfo{person}{Meena Mahajan}} (Eds.). \bibinfo{publisher}{Springer}, \bibinfo{address}{Berlin, Heidelberg}, \bibinfo{pages}{97--109}.
\newblock
\showISBNx{978-3-540-30538-5}
\urldef\tempurl%
\url{https://doi.org/10.1007/978-3-540-30538-5_9}
\showDOI{\tempurl}


\bibitem[Beyer(2020)]%
        {SVCOMP20}
\bibfield{author}{\bibinfo{person}{Dirk Beyer}.} \bibinfo{year}{2020}\natexlab{}.
\newblock \showarticletitle{Advances in Automatic Software Verification: {SV-COMP 2020}}. In \bibinfo{booktitle}{\emph{TACAS}} \emph{(\bibinfo{series}{LNCS})}, \bibfield{editor}{\bibinfo{person}{Armin Biere} {and} \bibinfo{person}{David Parker}} (Eds.). \bibinfo{publisher}{Springer}, \bibinfo{address}{Cham}, \bibinfo{pages}{347--367}.
\newblock
\urldef\tempurl%
\url{https://doi.org/10.1007/978-3-030-45237-7_21}
\showDOI{\tempurl}


\bibitem[Beyer(2024)]%
        {SVCOMP24}
\bibfield{author}{\bibinfo{person}{Dirk Beyer}.} \bibinfo{year}{2024}\natexlab{}.
\newblock \showarticletitle{State of the Art in Software Verification and Witness Validation: {SV-COMP 2024}}. In \bibinfo{booktitle}{\emph{TACAS}} \emph{(\bibinfo{series}{LNCS})}, \bibfield{editor}{\bibinfo{person}{Bernd Finkbeiner} {and} \bibinfo{person}{Laura Kov{\'a}cs}} (Eds.). \bibinfo{publisher}{Springer}, \bibinfo{address}{Cham}, \bibinfo{pages}{299--329}.
\newblock
\urldef\tempurl%
\url{https://doi.org/10.1007/978-3-031-57256-2_15}
\showDOI{\tempurl}


\bibitem[Calcagno and Distefano(2011)]%
        {Infer}
\bibfield{author}{\bibinfo{person}{Cristiano Calcagno} {and} \bibinfo{person}{Dino Distefano}.} \bibinfo{year}{2011}\natexlab{}.
\newblock \showarticletitle{Infer: An Automatic Program Verifier for Memory Safety of {C} Programs}. In \bibinfo{booktitle}{\emph{NASA Formal Methods}} \emph{(\bibinfo{series}{LNCS})}, \bibfield{editor}{\bibinfo{person}{Mihaela Bobaru}, \bibinfo{person}{Klaus Havelund}, \bibinfo{person}{Gerard~J. Holzmann}, {and} \bibinfo{person}{Rajeev Joshi}} (Eds.). \bibinfo{publisher}{Springer}, \bibinfo{address}{Berlin, Heidelberg}, \bibinfo{pages}{459--465}.
\newblock
\urldef\tempurl%
\url{https://doi.org/10.1007/978-3-642-20398-5_33}
\showDOI{\tempurl}


\bibitem[Calcagno et~al\mbox{.}(2009)]%
        {biabd_conference}
\bibfield{author}{\bibinfo{person}{Cristiano Calcagno}, \bibinfo{person}{Dino Distefano}, \bibinfo{person}{Peter O'Hearn}, {and} \bibinfo{person}{Hongseok Yang}.} \bibinfo{year}{2009}\natexlab{}.
\newblock \showarticletitle{Compositional Shape Analysis by Means of Bi-Abduction}.
\newblock \bibinfo{journal}{\emph{SIGPLAN Not.}} \bibinfo{volume}{44}, \bibinfo{number}{1} (\bibinfo{year}{2009}), \bibinfo{pages}{289--300}.
\newblock
\urldef\tempurl%
\url{https://doi.org/10.1145/1594834.1480917}
\showDOI{\tempurl}


\bibitem[Calcagno et~al\mbox{.}(2011)]%
        {biabd}
\bibfield{author}{\bibinfo{person}{Cristiano Calcagno}, \bibinfo{person}{Dino Distefano}, \bibinfo{person}{Peter~W. O'Hearn}, {and} \bibinfo{person}{Hongseok Yang}.} \bibinfo{year}{2011}\natexlab{}.
\newblock \showarticletitle{Compositional Shape Analysis by Means of Bi-Abduction}.
\newblock \bibinfo{journal}{\emph{J. ACM}} \bibinfo{volume}{58}, \bibinfo{number}{6}, Article \bibinfo{articleno}{26} (\bibinfo{year}{2011}), \bibinfo{numpages}{66}~pages.
\newblock
\urldef\tempurl%
\url{https://doi.org/10.1145/2049697.2049700}
\showDOI{\tempurl}


\bibitem[Calcagno et~al\mbox{.}(2007)]%
        {abstrsl}
\bibfield{author}{\bibinfo{person}{Cristiano Calcagno}, \bibinfo{person}{Peter~W. O'Hearn}, {and} \bibinfo{person}{Hongseok Yang}.} \bibinfo{year}{2007}\natexlab{}.
\newblock \showarticletitle{Local Action and Abstract Separation Logic}. In \bibinfo{booktitle}{\emph{LICS}}. \bibinfo{publisher}{IEEE Computer Society}, \bibinfo{address}{USA}, \bibinfo{pages}{366–--378}.
\newblock
\urldef\tempurl%
\url{https://doi.org/10.1109/LICS.2007.30}
\showDOI{\tempurl}


\bibitem[Distefano et~al\mbox{.}(2006)]%
        {abstraction_origin}
\bibfield{author}{\bibinfo{person}{Dino Distefano}, \bibinfo{person}{Peter~W. O'Hearn}, {and} \bibinfo{person}{Hongseok Yang}.} \bibinfo{year}{2006}\natexlab{}.
\newblock \showarticletitle{A Local Shape Analysis Based on Separation Logic}. In \bibinfo{booktitle}{\emph{TACAS}} \emph{(\bibinfo{series}{LNCS})}, \bibfield{editor}{\bibinfo{person}{Holger Hermanns} {and} \bibinfo{person}{Jens Palsberg}} (Eds.). \bibinfo{publisher}{Springer}, \bibinfo{address}{Berlin, Heidelberg}, \bibinfo{pages}{287--302}.
\newblock
\urldef\tempurl%
\url{https://doi.org/10.1007/11691372_19}
\showDOI{\tempurl}


\bibitem[Dr{\u{a}}goi et~al\mbox{.}(2013)]%
        {enea:sas13}
\bibfield{author}{\bibinfo{person}{Cezara Dr{\u{a}}goi}, \bibinfo{person}{Constantin Enea}, {and} \bibinfo{person}{Mihaela Sighireanu}.} \bibinfo{year}{2013}\natexlab{}.
\newblock \showarticletitle{Local Shape Analysis for Overlaid Data Structures}. In \bibinfo{booktitle}{\emph{SAS}} \emph{(\bibinfo{series}{LNCS})}. \bibinfo{publisher}{Springer}, \bibinfo{address}{Berlin, Heidelberg}, \bibinfo{pages}{150--171}.
\newblock
\urldef\tempurl%
\url{https://doi.org/10.1007/978-3-642-38856-9_10}
\showDOI{\tempurl}


\bibitem[Dudka et~al\mbox{.}(2013)]%
        {predator}
\bibfield{author}{\bibinfo{person}{Kamil Dudka}, \bibinfo{person}{Petr Peringer}, {and} \bibinfo{person}{Tom{\'a}{\v{s}} Vojnar}.} \bibinfo{year}{2013}\natexlab{}.
\newblock \showarticletitle{Byte-Precise Verification of Low-Level List Manipulation}. In \bibinfo{booktitle}{\emph{SAS}} \emph{(\bibinfo{series}{LNCS})}, \bibfield{editor}{\bibinfo{person}{Francesco Logozzo} {and} \bibinfo{person}{Manuel F{\"a}hndrich}} (Eds.). \bibinfo{publisher}{Springer}, \bibinfo{address}{Berlin, Heidelberg}, \bibinfo{pages}{215--237}.
\newblock
\showISBNx{978-3-642-38856-9}
\urldef\tempurl%
\url{https://doi.org/10.1007/978-3-642-38856-9_13}
\showDOI{\tempurl}


\bibitem[Giet et~al\mbox{.}(2023)]%
        {10.1007/978-3-031-44245-2_15}
\bibfield{author}{\bibinfo{person}{Josselin Giet}, \bibinfo{person}{F{\'e}lix Ridoux}, {and} \bibinfo{person}{Xavier Rival}.} \bibinfo{year}{2023}\natexlab{}.
\newblock \showarticletitle{A Product of Shape and Sequence Abstractions}. In \bibinfo{booktitle}{\emph{SAS}} \emph{(\bibinfo{series}{LNCS})}, \bibfield{editor}{\bibinfo{person}{Manuel~V. Hermenegildo} {and} \bibinfo{person}{Jos{\'e}~F. Morales}} (Eds.). \bibinfo{publisher}{Springer}, \bibinfo{address}{Cham}, \bibinfo{pages}{310--342}.
\newblock
\showISBNx{978-3-031-44245-2}
\urldef\tempurl%
\url{https://doi.org/10.1007/978-3-031-44245-2_15}
\showDOI{\tempurl}


\bibitem[Gulavani et~al\mbox{.}(2009)]%
        {bottomup}
\bibfield{author}{\bibinfo{person}{Bhargav~S. Gulavani}, \bibinfo{person}{Supratik Chakraborty}, \bibinfo{person}{Ganesan Ramalingam}, {and} \bibinfo{person}{Aditya~V. Nori}.} \bibinfo{year}{2009}\natexlab{}.
\newblock \showarticletitle{Bottom-Up Shape Analysis}. In \bibinfo{booktitle}{\emph{SAS}} \emph{(\bibinfo{series}{LNCS})}, \bibfield{editor}{\bibinfo{person}{Jens Palsberg} {and} \bibinfo{person}{Zhendong Su}} (Eds.). \bibinfo{publisher}{Springer}, \bibinfo{address}{Berlin, Heidelberg}, \bibinfo{pages}{188--204}.
\newblock
\urldef\tempurl%
\url{https://doi.org/10.1007/978-3-642-03237-0_14}
\showDOI{\tempurl}


\bibitem[Hol{\'\i}k et~al\mbox{.}(2022)]%
        {Broom}
\bibfield{author}{\bibinfo{person}{Luk\'{a}\v{s} Hol{\'\i}k}, \bibinfo{person}{Petr Peringer}, \bibinfo{person}{Adam Rogalewicz}, \bibinfo{person}{Veronika \v{S}okov\'{a}}, \bibinfo{person}{Tom\'{a}\v{s} Vojnar}, {and} \bibinfo{person}{Florian Zuleger}.} \bibinfo{year}{2022}\natexlab{}.
\newblock \showarticletitle{{Low-Level Bi-Abduction}}. In \bibinfo{booktitle}{\emph{ECOOP}} \emph{(\bibinfo{series}{LIPIcs}, Vol.~\bibinfo{volume}{222})}, \bibfield{editor}{\bibinfo{person}{Karim Ali} {and} \bibinfo{person}{Jan Vitek}} (Eds.). \bibinfo{publisher}{Schloss Dagstuhl -- Leibniz-Zentrum f{\"u}r Informatik}, \bibinfo{address}{Dagstuhl}, \bibinfo{pages}{19:1--19:30}.
\newblock
\urldef\tempurl%
\url{https://doi.org/10.4230/LIPIcs.ECOOP.2022.19}
\showDOI{\tempurl}


\bibitem[Illous et~al\mbox{.}(2017)]%
        {relational_shape_abstract_domain}
\bibfield{author}{\bibinfo{person}{Hugo Illous}, \bibinfo{person}{Matthieu Lemerre}, {and} \bibinfo{person}{Xavier Rival}.} \bibinfo{year}{2017}\natexlab{}.
\newblock \showarticletitle{A Relational Shape Abstract Domain}. In \bibinfo{booktitle}{\emph{NASA Formal Methods}} \emph{(\bibinfo{series}{LNCS})}, \bibfield{editor}{\bibinfo{person}{Clark Barrett}, \bibinfo{person}{Misty Davies}, {and} \bibinfo{person}{Temesghen Kahsai}} (Eds.). \bibinfo{publisher}{Springer}, \bibinfo{address}{Cham}, \bibinfo{pages}{212--229}.
\newblock
\urldef\tempurl%
\url{https://doi.org/10.1007/978-3-319-57288-8_15}
\showDOI{\tempurl}


\bibitem[Illous et~al\mbox{.}(2020)]%
        {transformer}
\bibfield{author}{\bibinfo{person}{Hugo Illous}, \bibinfo{person}{Matthieu Lemerre}, {and} \bibinfo{person}{Xavier Rival}.} \bibinfo{year}{2020}\natexlab{}.
\newblock \showarticletitle{Interprocedural Shape Analysis Using Separation Logic-Based Transformer Summaries}. In \bibinfo{booktitle}{\emph{SAS}} \emph{(\bibinfo{series}{LNCS})}, \bibfield{editor}{\bibinfo{person}{David Pichardie} {and} \bibinfo{person}{Mihaela Sighireanu}} (Eds.). \bibinfo{publisher}{Springer}, \bibinfo{address}{Cham}, \bibinfo{pages}{248--273}.
\newblock
\urldef\tempurl%
\url{https://doi.org/10.1007/978-3-030-65474-0_12}
\showDOI{\tempurl}


\bibitem[Kaindlstorfer(2023)]%
        {david}
\bibfield{author}{\bibinfo{person}{David Kaindlstorfer}.} \bibinfo{year}{2023}\natexlab{}.
\newblock \emph{\bibinfo{title}{Enhancing Abstraction and Symbolic Execution for Shape Analysis of C-Programs operating on Linked Lists}}.
\newblock Diploma Thesis. \bibinfo{school}{TU Wien}.
\newblock
\urldef\tempurl%
\url{https://doi.org/10.34726/hss.2023.109623}
\showDOI{\tempurl}


\bibitem[Le et~al\mbox{.}(2014)]%
        {secondorder}
\bibfield{author}{\bibinfo{person}{Quang~Loc Le}, \bibinfo{person}{Cristian Gherghina}, \bibinfo{person}{Shengchao Qin}, {and} \bibinfo{person}{Wei-Ngan Chin}.} \bibinfo{year}{2014}\natexlab{}.
\newblock \showarticletitle{Shape Analysis via Second-Order Bi-Abduction}. In \bibinfo{booktitle}{\emph{CAV}} \emph{(\bibinfo{series}{LNCS})}, \bibfield{editor}{\bibinfo{person}{Armin Biere} {and} \bibinfo{person}{Roderick Bloem}} (Eds.). \bibinfo{publisher}{Springer}, \bibinfo{address}{Cham}, \bibinfo{pages}{52--68}.
\newblock
\showISBNx{978-3-319-08867-9}
\urldef\tempurl%
\url{https://doi.org/10.1007/978-3-319-08867-9_4}
\showDOI{\tempurl}


\bibitem[Le et~al\mbox{.}(2022)]%
        {ISLX}
\bibfield{author}{\bibinfo{person}{Quang~Loc Le}, \bibinfo{person}{Azalea Raad}, \bibinfo{person}{Jules Villard}, \bibinfo{person}{Josh Berdine}, \bibinfo{person}{Derek Dreyer}, {and} \bibinfo{person}{Peter~W. O'Hearn}.} \bibinfo{year}{2022}\natexlab{}.
\newblock \showarticletitle{Finding Real Bugs in Big Programs with Incorrectness Logic}.
\newblock \bibinfo{journal}{\emph{Proc. ACM Program. Lang.}} \bibinfo{volume}{6}, \bibinfo{number}{OOPSLA1}, Article \bibinfo{articleno}{81} (\bibinfo{year}{2022}), \bibinfo{numpages}{27}~pages.
\newblock
\urldef\tempurl%
\url{https://doi.org/10.1145/3527325}
\showDOI{\tempurl}


\bibitem[Magill et~al\mbox{.}(2015)]%
        {otherabstraction}
\bibfield{author}{\bibinfo{person}{Stephen Magill}, \bibinfo{person}{Aleksandar Nanevski}, \bibinfo{person}{Edmund~M. Clarke}, {and} \bibinfo{person}{Peter Lee}.} \bibinfo{year}{2015}\natexlab{}.
\newblock \bibinfo{title}{Inferring Invariants in Separation Logic for Imperative List-processing Programs}.  (\bibinfo{year}{2015}).
\newblock
\newblock
\shownote{Draft}.


\bibitem[Maksimovi\'{c} et~al\mbox{.}(2023)]%
        {exactsl}
\bibfield{author}{\bibinfo{person}{Petar Maksimovi\'{c}}, \bibinfo{person}{Caroline Cronj\"{a}ger}, \bibinfo{person}{Andreas L\"{o}\"{o}w}, \bibinfo{person}{Julian Sutherland}, {and} \bibinfo{person}{Philippa Gardner}.} \bibinfo{year}{2023}\natexlab{}.
\newblock \showarticletitle{Exact Separation Logic: Towards Bridging the Gap Between Verification and Bug-Finding}. In \bibinfo{booktitle}{\emph{ECOOP}} \emph{(\bibinfo{series}{LIPIcs}, Vol.~\bibinfo{volume}{263})}, \bibfield{editor}{\bibinfo{person}{Karim Ali} {and} \bibinfo{person}{Guido Salvaneschi}} (Eds.). \bibinfo{publisher}{Schloss Dagstuhl -- Leibniz-Zentrum f{\"u}r Informatik}, \bibinfo{address}{Dagstuhl}, \bibinfo{pages}{19:1--19:27}.
\newblock
\urldef\tempurl%
\url{https://doi.org/10.4230/LIPIcs.ECOOP.2023.19}
\showDOI{\tempurl}


\bibitem[Polikarpova and Sergey(2019)]%
        {journals/pacmpl/PolikarpovaS19}
\bibfield{author}{\bibinfo{person}{Nadia Polikarpova} {and} \bibinfo{person}{Ilya Sergey}.} \bibinfo{year}{2019}\natexlab{}.
\newblock \showarticletitle{Structuring the synthesis of heap-manipulating programs}.
\newblock \bibinfo{journal}{\emph{Proc. ACM Program. Lang.}} \bibinfo{volume}{3}, \bibinfo{number}{POPL}, Article \bibinfo{articleno}{72} (\bibinfo{year}{2019}), \bibinfo{numpages}{30}~pages.
\newblock
\urldef\tempurl%
\url{https://doi.org/10.1145/3290385}
\showDOI{\tempurl}


\bibitem[Qin et~al\mbox{.}(2017)]%
        {10.1016/j.scico.2017.05.007}
\bibfield{author}{\bibinfo{person}{Shengchao Qin}, \bibinfo{person}{Guanhua He}, \bibinfo{person}{Wei-Ngan Chin}, \bibinfo{person}{Florin Craciun}, \bibinfo{person}{Mengda He}, {and} \bibinfo{person}{Zhong Ming}.} \bibinfo{year}{2017}\natexlab{}.
\newblock \showarticletitle{Automated specification inference in a combined domain via user-defined predicates}.
\newblock \bibinfo{journal}{\emph{Sci. Comput. Program.}} \bibinfo{volume}{148}, \bibinfo{number}{C} (\bibinfo{year}{2017}), \bibinfo{pages}{189--212}.
\newblock
\showISSN{0167-6423}
\urldef\tempurl%
\url{https://doi.org/10.1016/j.scico.2017.05.007}
\showDOI{\tempurl}


\bibitem[Raad et~al\mbox{.}(2020)]%
        {isl}
\bibfield{author}{\bibinfo{person}{Azalea Raad}, \bibinfo{person}{Josh Berdine}, \bibinfo{person}{Hoang-Hai Dang}, \bibinfo{person}{Derek Dreyer}, \bibinfo{person}{Peter O'Hearn}, {and} \bibinfo{person}{Jules Villard}.} \bibinfo{year}{2020}\natexlab{}.
\newblock \showarticletitle{Local Reasoning About the Presence of Bugs: Incorrectness Separation Logic}. In \bibinfo{booktitle}{\emph{CAV}} \emph{(\bibinfo{series}{LNCS})}. \bibinfo{publisher}{Springer}, \bibinfo{address}{Berlin, Heidelberg}, \bibinfo{pages}{225--252}.
\newblock
\urldef\tempurl%
\url{https://doi.org/10.1007/978-3-030-53291-8_14}
\showDOI{\tempurl}


\bibitem[Raad et~al\mbox{.}(2022)]%
        {cisl}
\bibfield{author}{\bibinfo{person}{Azalea Raad}, \bibinfo{person}{Josh Berdine}, \bibinfo{person}{Derek Dreyer}, {and} \bibinfo{person}{Peter~W. O'Hearn}.} \bibinfo{year}{2022}\natexlab{}.
\newblock \showarticletitle{Concurrent Incorrectness Separation Logic}.
\newblock \bibinfo{journal}{\emph{Proc. ACM Program. Lang.}} \bibinfo{volume}{6}, \bibinfo{number}{POPL}, Article \bibinfo{articleno}{34} (\bibinfo{year}{2022}), \bibinfo{numpages}{29}~pages.
\newblock
\urldef\tempurl%
\url{https://doi.org/10.1145/3498695}
\showDOI{\tempurl}


\bibitem[Reynolds(2002)]%
        {reynolds}
\bibfield{author}{\bibinfo{person}{John~C. Reynolds}.} \bibinfo{year}{2002}\natexlab{}.
\newblock \showarticletitle{Separation Logic: A Logic for Shared Mutable Data Structures}. In \bibinfo{booktitle}{\emph{LICS}}. \bibinfo{publisher}{IEEE Computer Society}, \bibinfo{address}{USA}, \bibinfo{pages}{55--74}.
\newblock
\showISBNx{0769514839}
\urldef\tempurl%
\url{https://doi.org/10.1109/LICS.2002.1029817}
\showDOI{\tempurl}


\bibitem[Rysavy(2024)]%
        {RysavyLukas2024Jofb}
\bibfield{author}{\bibinfo{person}{Lukas Rysavy}.} \bibinfo{year}{2024}\natexlab{}.
\newblock \emph{\bibinfo{title}{Join operators for bi-abductive analysis of low-level code}}.
\newblock Diploma Thesis. \bibinfo{school}{TU Wien}.
\newblock
\urldef\tempurl%
\url{https://doi.org/10.34726/hss.2024.119373}
\showDOI{\tempurl}


\bibitem[Sammler et~al\mbox{.}(2021)]%
        {refinedc}
\bibfield{author}{\bibinfo{person}{Michael Sammler}, \bibinfo{person}{Rodolphe Lepigre}, \bibinfo{person}{Robbert Krebbers}, \bibinfo{person}{Kayvan Memarian}, \bibinfo{person}{Derek Dreyer}, {and} \bibinfo{person}{Deepak Garg}.} \bibinfo{year}{2021}\natexlab{}.
\newblock \showarticletitle{{RefinedC}: automating the foundational verification of C code with refined ownership types}. In \bibinfo{booktitle}{\emph{PLDI}}. \bibinfo{publisher}{Association for Computing Machinery}, \bibinfo{address}{New York}, \bibinfo{pages}{158--174}.
\newblock
\urldef\tempurl%
\url{https://doi.org/10.1145/3453483.3454036}
\showDOI{\tempurl}


\bibitem[Sextl et~al\mbox{.}(2025)]%
        {brush_artifact}
\bibfield{author}{\bibinfo{person}{Florian Sextl}, \bibinfo{person}{Adam Rogalewicz}, \bibinfo{person}{Tomas Vojnar}, {and} \bibinfo{person}{Florian Zuleger}.} \bibinfo{year}{2025}\natexlab{}.
\newblock \bibinfo{booktitle}{\emph{Artifact for "Compositional Shape Analysis with Shared Abduction and Biabductive Loop Acceleration"}}.
\newblock
\urldef\tempurl%
\url{https://doi.org/10.5281/zenodo.14623977}
\showDOI{\tempurl}


\bibitem[Spies et~al\mbox{.}(2024)]%
        {quiver}
\bibfield{author}{\bibinfo{person}{Simon Spies}, \bibinfo{person}{Lennard G\"{a}her}, \bibinfo{person}{Michael Sammler}, {and} \bibinfo{person}{Derek Dreyer}.} \bibinfo{year}{2024}\natexlab{}.
\newblock \showarticletitle{Quiver: Guided Abductive Inference of Separation Logic Specifications in Coq}.
\newblock \bibinfo{journal}{\emph{Proc. ACM Program. Lang.}} \bibinfo{volume}{8}, \bibinfo{number}{PLDI}, Article \bibinfo{articleno}{183} (\bibinfo{year}{2024}), \bibinfo{numpages}{25}~pages.
\newblock
\urldef\tempurl%
\url{https://doi.org/10.1145/3656413}
\showDOI{\tempurl}


\bibitem[Wyatt(2012)]%
        {intrusive}
\bibfield{author}{\bibinfo{person}{Patrick Wyatt}.} \bibinfo{year}{2012}\natexlab{}.
\newblock \bibinfo{booktitle}{\emph{Avoiding game crashes related to linked lists}}.
\newblock
\urldef\tempurl%
\url{http://www.codeofhonor.com/blog/avoiding-game-crashes-related-to-linked-lists}
\showURL{%
\tempurl}


\bibitem[Zilberstein et~al\mbox{.}(2023)]%
        {outcome}
\bibfield{author}{\bibinfo{person}{Noam Zilberstein}, \bibinfo{person}{Derek Dreyer}, {and} \bibinfo{person}{Alexandra Silva}.} \bibinfo{year}{2023}\natexlab{}.
\newblock \showarticletitle{Outcome Logic: A Unifying Foundation for Correctness and Incorrectness Reasoning}.
\newblock \bibinfo{journal}{\emph{Proc. ACM Program. Lang.}} \bibinfo{volume}{7}, \bibinfo{number}{OOPSLA1}, Article \bibinfo{articleno}{93} (\bibinfo{year}{2023}), \bibinfo{numpages}{29}~pages.
\newblock
\urldef\tempurl%
\url{https://doi.org/10.1145/3586045}
\showDOI{\tempurl}


\bibitem[Zilberstein et~al\mbox{.}(2024)]%
        {outcomesl}
\bibfield{author}{\bibinfo{person}{Noam Zilberstein}, \bibinfo{person}{Angelina Saliling}, {and} \bibinfo{person}{Alexandra Silva}.} \bibinfo{year}{2024}\natexlab{}.
\newblock \showarticletitle{Outcome Separation Logic: Local Reasoning for Correctness and Incorrectness with Computational Effects}.
\newblock \bibinfo{journal}{\emph{Proc. ACM Program. Lang.}} \bibinfo{volume}{8}, \bibinfo{number}{OOPSLA1}, Article \bibinfo{articleno}{104} (\bibinfo{year}{2024}), \bibinfo{numpages}{29}~pages.
\newblock
\urldef\tempurl%
\url{https://doi.org/10.1145/3649821}
\showDOI{\tempurl}


\end{thebibliography}
\newpage
\appendix
\section{Concrete Semantics}
\paragraph{More Notation.}
Let $f$ and $g$ be partial functions.
Then $\mathit{dom}(f)$ is the domain of $f$ and
we write $f(x)=\bot$ if $x\notin \mathit{dom}(f)$.
Furthermore, we denote sets with the usual notation $\{x_0,\dots,x_n\}$ and use $\emptyset$
for the empty set.
We use the function-update syntax $f[a\hookrightarrow b]$ to denote the
(partial) function $f'$ that is defined as $f'(x)\mathdef b$ if $x=a$, and
$f'(x)\mathdef f(x)$ otherwise.
We write $f\uplus g$ for the disjoint union of two partial functions $f$ and $g$ if $\mathit{dom}(f)\cap\mathit{dom}(g)=\emptyset$, such that $\left(f\uplus g\right)(x) \mathdef f(x)\text{ if } x\in \mathit{dom}(f), g(x) \text{ if } x\in \mathit{dom}(g), \bot$ else.

The semantics of our separation logic and programming language are defined in the following.
We note that $\llbracket x\rrbracket S$ is only defined if $x\in \mathit{dom}(S)$.
However, we can assume that this condition is always satisfied as this condition can be ensured with syntactic type safety.
The semantics are shown in \Cref{fig:seplogsemantics,fig:langsem}.
We note that allocation is non-deterministic and can either succeed and return a fresh memory location or fail, not change the heap, and return $\Null$ instead.
Moreover, we also note that for pure formulas, the stack suffices to define the semantics.

\begin{figure}[H]
    \begin{gather*}
        \llbracket \Null\rrbracket S \mathdef \Null\quad \llbracket k\rrbracket S \mathdef k \quad \llbracket x\rrbracket S \mathdef S(x), \text{ if } x\in \mathit{dom}(S) \\ 
        \llbracket \textsf{unop } e\rrbracket S \mathdef \textsf{unop } \llbracket e\rrbracket S \quad \llbracket e_1 \textsf{ binop } e_2\rrbracket S \mathdef \llbracket e_1 \rrbracket S \textsf{ binop } \llbracket e_2\rrbracket S \\ \llbracket \top \rrbracket S \mathdef \text{ some value } v\in\mathit{Val}
    \end{gather*}
    \begin{align*}
        (S,H)\vDash x.f \mapsto \varepsilon        & \ \Leftrightarrow \mathit{dom}(H)=\{(\llbracket x\rrbracket S,f)\} \land H(\llbracket x\rrbracket S,f) = \llbracket\varepsilon\rrbracket S  \\
        (S,H)\vDash \Sigma_1 *\Sigma_2             & \ \Leftrightarrow \exists H_1,H_2.\ (S,H_1)\vDash \Sigma_1 \land (S,H_2)\vDash \Sigma_2                                                     \\&\land \mathit{dom}(H_1)\cap\mathit{dom}(H_2)=\emptyset \land H=H_1\uplus H_2\\
        (S,H)\vDash \ls(x,\varepsilon)              & \ \Leftrightarrow \left(\mathit{dom}(H)=\emptyset\land (S,H)\vDash x=\varepsilon\right)                                                     \\
                                                   & \ \lor x\neq \epsilon \land \exists l.\ (S[y\hookrightarrow l],H)\vDash x.\texttt{next}\mapsto y*\ls(y,\varepsilon), y\notin \mathit{dom}(S) \\
        (S,H)\vDash \textsf{emp}                   & \ \Leftrightarrow \mathit{dom}(H)=\emptyset                                                                                                 \\
        S\vDash \Pi_1\land \Pi_2                   & \ \Leftrightarrow S\vDash \Pi_1 \land S\vDash \Pi_2                                                                                         \\
        S\vDash \textsf{true}                      & \                                                                                                                                           \\
        S\vDash \varepsilon_1 \oplus \varepsilon_2 & \ \Leftrightarrow \llbracket \varepsilon_1\rrbracket S \oplus \llbracket \varepsilon_2\rrbracket S                                          \\
        (S,H)\vDash \Pi\formsep \Sigma             & \ \Leftrightarrow S\vDash \Pi\land (S,H)\vDash \Sigma                                                                                       \\
        (S,H)\vDash \varphi\lor \Delta             & \ \Leftrightarrow (S,H)\vDash \varphi\lor (S,H)\vDash \Delta
    \end{align*}
    \caption{Semantics of the separation logic}
    \Description{The semantics of our separation logic. Everything is defined as expected.}
    \label{fig:seplogsemantics}
\end{figure}
\newpage

\begin{figure}[H]
    \begin{gather*}
        \llbracket k\rrbracket S \mathdef k \quad
        \llbracket x\rrbracket S \mathdef S(x), \text{ if } x\in \mathit{dom}(S) \quad
        \llbracket ?\rrbracket S \mathdef\top \\
        \llbracket \textsf{unop } e\rrbracket S \mathdef \textsf{unop } \llbracket e\rrbracket S \quad
        \llbracket e_1 \textsf{ binop } e_2\rrbracket S  \mathdef \llbracket e_1 \rrbracket S \textsf{ binop } \llbracket e_2\rrbracket S
    \end{gather*}
    \begin{align*}
        \left(\mathit{err}, st\right)\rightsquigarrow                        & \ \mathit{err}, \text{ for any statement } st                                                                                                    \\
        \left((S,H), x\!=\!e\right)\rightsquigarrow                          & \ (S[x\hookrightarrow \llbracket e\rrbracket S],H)                                                                                               \\
        \left((S,H), x_1\!=\!*x_2.f\right)\rightsquigarrow                   & \ (S[x_1 \hookrightarrow H(\llbracket x_2\rrbracket S,f)],H), \text{ if } (\llbracket x_2\rrbracket S,f) \in \mathit{dom}(H)                     \\
        \left((S,H), x_1\!=\!*x_2.f\right)\rightsquigarrow                   & \ \mathit{err}, \text{ if } (\llbracket x_2\rrbracket S,f) \notin  \mathit{dom}(H)                                                               \\
        \left((S,H),\! *x_1.f\!=\!x_2\right)\rightsquigarrow                 & \ (S,H[(\llbracket x_1\rrbracket S,f)\hookrightarrow \llbracket x_2\rrbracket S]),\text{ if } (\llbracket x_2\rrbracket S,f) \in \mathit{dom}(H) \\
        \left((S,H),\! *x_1.f\!=\!x_2\right)\rightsquigarrow                 & \ \mathit{err},\text{ if } (\llbracket x_2\rrbracket S,f) \notin \mathit{dom}(H)                                                                 \\
        \left((S,H), \textsc{return}\ x\right)\rightsquigarrow               & \ (S[\texttt{return}_f\hookrightarrow \llbracket x\rrbracket S],H)                                                                               \\
        \left((S,H), x=\textsc{Alloc}(f_1,\dots,f_n) \right)\rightsquigarrow & \ (S[x\hookrightarrow l], H'),                                                                                                                   \\
                                                                             & \text{ where either } H' = H \text{ and } l=\Null                                                                                                \\
                                                                             & \text{ or } H'=H\uplus[(l,f_i)\hookrightarrow v_i \mid l\in\mathit{Loc}-\{\Null\},                                                               \\
                                                                             & (l,f_i)\notin \mathit{dom}(H) \text{ and } v_i\in\mathit{Val} \text{ arbitrary for } 1\leq i\leq n]                                              \\
        \left((S,H), \textsc{Free}(x) \right)\rightsquigarrow                & \ (S, H), \text{ if } \llbracket x\rrbracket S = \Null                                                                                           \\
        \left((S,H), \textsc{Free}(x) \right)\rightsquigarrow                & \ (S, H[(\llbracket x\rrbracket S,f)\hookrightarrow \bot]), \text{ for all }f \text{ s.t. } (\llbracket x\rrbracket S,f)\in \mathit{dom}(H)      \\
        \left((S,H), \textsc{Free}(x) \right)\rightsquigarrow                & \ \mathit{err}, \text{ if } \forall f.\ (\llbracket x\rrbracket S,f)\notin \mathit{dom}(H)                                                       \\
        \left((S,H), \textsc{assume}(x_1\oplus x_2)\right)\rightsquigarrow   & \ (S,H), \text{ if } \llbracket x_1\rrbracket S \oplus \llbracket x_2\rrbracket S \text{ holds}                                                  \\
        \left((S,H), \textsc{assert}(x_1\oplus x_2)\right)\rightsquigarrow   & \ (S,H), \text{ if } \llbracket x_1\rrbracket S \oplus \llbracket x_2\rrbracket S \text{ holds}                                                  \\
        \left((S,H), \textsc{assert}(x_1\oplus x_2)\right)\rightsquigarrow   & \ \mathit{err}, \text{ if } \llbracket x_1\rrbracket S \oplus \llbracket x_2\rrbracket S \text{ does not hold}                                   \\
        \left((S,H),x=f(x_1,\dots,x_n)\right)\rightsquigarrow                & \ \mathit{err},                                                                                                                                  \\&\text{ if } \mathit{err}\in f_{(S_f,H)}\text { where } S_f=[a_i \hookrightarrow \llbracket x_i\rrbracket S\mid a_i \text{ argument of } f]\\
        \left((S,H),x=f(x_1,\dots,x_n)\right)\rightsquigarrow                & \ (S[x\hookrightarrow \llbracket\texttt{return}_f\rrbracket S'_f],H'),                                                                           \\&\text{if } (S'_f,H') \in f_{(S_f,H)}\land \mathit{err}\notin f_{(S_f,H)},\\
                                                                             & \ \text{where }S_f=[a_i \hookrightarrow \llbracket x_i\rrbracket S\mid a_i \text{ argument of } f],                                              \\
        f_{(S,H)} \mathdef                                                   & \ \{C\mid \exists t_f.\ \left((S,H),t_f\right)\rightsquigarrow^* C \text{ with } t_f= [\mathit{entry},\dots,\mathit{exit}] \}
    \end{align*}
    \caption{The semantics of the programming language}
    \Description{The semantics of our programming language. We define the semantics of each basic statement, traces and full functions.}
    \label{fig:langsem}
\end{figure}
\newpage

\section{Atomic Contracts}\label{app:ctxt}
Based on the introduced separation logic and programming language with their respective semantics, we define the basic contracts of atomic statements.
We note that the contracts are all trivially sound and suffice for a soundness-preserving frame rule.
For simplicity of presentation, we assume that compound expressions are built stepwise with intermediate variables as in SSA form, such that $y=x*10-2$ would become $c_{10}=10;x'=x*c_{10};c_2=2; y=x'-c_2$.
Moreover, we state the contracts for the statements as if these were represented by corresponding functions with only the single statement as body.
This approach follows the presentation in \citep{Broom}.
\begin{align*}
    \{y=Y\}\ &x=y\ \{x=Y\land y=Y\}\\ \{\mathit{emp}\}           & \ x=k\ \{x=k\}\\ \{\mathit{emp}\}\ &x=\ ?\ \{x=\ell_1\}, \text{ where } \ell_1\text{ is fresh} \\
    \{y=Y\}                                                        & \ x=\textsc{unop}\ y\ \{y=Y\land x=\textsc{unop}\ y\}                                            \\ \{y=Y\land z=Z\}&\ x=y\ \textsc{binop}\ z\ \{y=Y\land z=Z\land x=y\ \textsc{binop}\ z\}\\
    \{y=Y\formsep Y.f\mapsto Z\}                                   & \ x=*y.f\ \{y=Y\land x=Z\formsep Y.f\mapsto Z\}                                                  \\
    \{x=X\land y=Y\formsep X.f\mapsto Z\}                          & \ *\!x.f=y\ \{x=X\land y=Y\formsep X.f\mapsto Y\}                                                \\
    \{x=X\}                                                        & \ \textsc{return}\ x\ \{x=X\land \mathit{return}_f=X\}, \text{in function } f                    \\
    \{x=X\land y=Y\}                                               & \ \textsc{assume}(x\oplus y)\ \{x=X\land y=Y \land X\oplus Y\}                                   \\
    \{x=X\land y=Y\land X\oplus Y\}                                & \ \textsc{assert}(x\oplus y)\ \{x=X\land y=Y \land X\oplus Y\}                                   \\
    \{x=X\formsep X.f_1\mapsto X_1* \dotso *X.f_n\mapsto X_{n}\}\  & \textsc{Free}(x)\ \{x=X\}                                                                        \\
    \{x=\Null\}\                                                   & \textsc{Free}(x)\ \{x=\Null\}
\end{align*}
\begin{align*}
    \{\mathit{emp}\} & \ x=\textsc{Alloc}(f_1,\dots,f_n)\ \{(x=\ell_x\formsep \ell_x.f_1\mapsto \ell_1*\dotso*\ell_x.f_n\mapsto \ell_n) \lor (x=\Null)\}, \\ &\text{ where } \ell_1,\dots\ell_n,\ell_x \text{ are fresh}\\
\end{align*}

\section{Branching in Biabductive Shape Analysis}\label{app:assume}
Following the seminal work by \citet{biabd,biabd_conference}, we introduce
two modes of handling \textsc{Assume} statements to enhance precision of the
analysis and make it more path-sensitive.
This path-sensitivity has proven to be useful as it allows to precisely locate
bugs in the program that make verification impossible.

First, if the condition $\cond$ can be expressed as part of the
pre-condition $P$, i.e. if the variables in $\cond$ can be
reached from the anchor variables through propositions in $P$ and
$\cond$, we can apply the \textit{assume-as-assert} mode.
As an example, if the function has arguments $a$ and $b$ and the condition is
$b==\Null$ under the current state $b=B$, \textit{assume-as-assert} trivially
applies, as the condition can be evaluated just from knowing the anchor value $B$.

This mode handles \textsc{Assume}(\textit{cond}) as if it was an \textsc{Assert}
and thus not only adds the assumption $\cond$ into the post-condition
$Q$ but also adds it to the pre-condition $P$, i.e., $Q$ is updated to
$\cond * Q$ and $P$ to $\cond * P$ (modulo renaming of variables).

If the condition cannot be expressed as part of the pre-condition,
the analysis falls back to the \textit{assume-as-assume} mode.
As an example, if the condition is $i>0$ where $i$ is obtained from user input,
the condition cannot be evaluated just from knowing the function arguments, i.e.,
it is not expressible in terms of the anchor variables.

As we assume all vertices in CFGs to have at most two successors, branching
points such as for an \texttt{if-then-else} are vertices with two outgoing
edges, each annotated with an \textsc{assume} statement, such that the
conditions are negations of each other.
In the case of \textit{assume-as-assert}, the current \absstate{} would
have both contradicting conditions added to its candidate pre-condition.
To circumvent this problem, we require the analysis to do \textit{state
    splitting}, i.e., instead of introducing inconsistencies at
\textit{assume-as-assert} branching points, it instead continues with one
copy of its current \absstate\ for each branch.
Moreover, to not lose precision unnecessarily, we assume the shape analysis
to also split its \absstate\ at \textit{assume-as-assume} such that the two
resulting states share the same pre-condition.

\section{Proofs}\label{app:proofs}
\paragraph{Further Notation.}
We define $[v_0,st_1,\dots,st_n,\mathit{entry}_g,g^m], m\in \mathbb{N}$ to be a family of traces that share a partial trace $[v_0,st_1,\dots,st_n,\mathit{entry}_g]$ and then take an arbitrary path through $m$ repetitions of the (partial) CFG $g$ such that $\mathit{entry}_g = \mathit{exit}_g$ for all but the last repetition.
Based on this, we define $[v_0,st_1,\dots,st_n,\mathit{entry}_g,g^*]$ to be the family of all traces $[v_0,st_1,\dots,st_n,\mathit{entry}_g,g^m]$ for $m\geq 0$.

\begin{lemma}\label{thm:conseq}
    The rule of consequence for Hoare triples:
    \[
        \frac{P\vdash P'\quad \{P'\}\ t\ \{Q'\}\quad Q'\vdash Q}{\{P\}\ t\ \{Q\}}
    \]
\end{lemma}
\begin{proof}
    The rule follows directly from the definitions of entailment (see \Cref{def:seplog}) and Hoare semantics (see \Cref{def:soundst}).
\end{proof}

\subsection{Shared Abduction}\label{subsec:appproofshared}

For worlds where the exact pre- and post-conditions are not relevant, we write $(P,Q)\soundw t$ instead of $\{P\}\ t \{Q\}$.

\begin{lemma}\label{thm:soundlearn}
    Let $\ana_{B,\abs}$  be a biabduction-based shape analysis based on worlds that computes only sound contracts for functions without branching.
    Then, if $\forall t\in T. \ w\soundw t$ holds and the world $w$ is transformed to $w'$ by a shared abduction step for current post-condition $Q_i^{l_i}\in w.\mathit{curr}$ along the edge $(l_i,st_i,l_i+1)$ in the CFG, it follows that $\forall t\in T'.\ w'\soundw t$ where $T'\mathdef T\cup T_\mathit{new}$ and $T_\mathit{new}\mathdef \{[t,l_i,st_i,l_i+1]\mid [t,l_i]\in T\land (w.\mathit{pre}, Q_i)\soundw [t,l_i]\}$.
\end{lemma}

\begin{proof}
    We take a fixed but arbitrary world $w$ with $n$ current post-conditions and a family of traces $T$ such that $\forall t\in T.\ w\soundw t\ (1)$.
    We also fix a trace $t_i$ ending in $l_i$ and assume that the analysis does an analysis step for $Q_i^{l_i}\in w.\mathit{curr}$ along the edge $(l_i,st_i,l_i+1)$.
    There, $B$ finds the antiframe $M$ and frame $F$ as the solution to the biabduction query for a contract $(L,R)$ of $st_i$ and updates the world to $w'$ with the new current post-condition $(F*R)_{n+1}^{l_i+1}$.
    If $M\not\subseteq w.\mathit{pre}*M\downarrow_\mathit{AnchVar}$, the analysis fails as it has found a requirement about local variables that is impossible to fulfill.
    Thus, we can assume $M\subset w.\mathit{pre}*M\downarrow_\mathit{AnchVar}$.
    Further, if the constraints in $M$ contradict with $w.\mathit{pre}$, the analysis also fails.
    As this case is trivial, we assume that $w.\mathit{pre}*M$ is satisfiable.
    We now do a case analysis on the traces in $T'$ to show $\forall t\in T'.\ w'\soundw t$.

    In the first case, we take an arbitrary but fixed trace $[t,l_i,st_i,l_i+1] \in T_\mathit{new}$ such that $[t,l_i]\in T\land (w.\mathit{pre}, Q_i)\soundw [t,l_i]$.
    From the definition of the biabductive symbolic execution, it follows that also $(w.\mathit{pre}*M,Q_i*M)\soundw [t,l_i]$.
    By the frame rule of separation logic, it further follows that $(w.\mathit{pre}*M\statesep F*R)\soundw [t,l_i,st_i,l_i+1]$.
    Because $(w.\mathit{pre}*M\statesep F*R)$ is the equivalent \absstate\:to the newly added current post-condition of $w'$, $w'\soundw [t,l_i,st_i,l_i+1]$ holds by construction.

    In the second case, we take an arbitrary but fixed trace $t\in T$ for a post-condition $Q_j$ such that $j\neq i\land (w.\mathit{pre}\statesep Q_j)\soundw t$.
    By computing $w'$, the $Q_j$ gets updated to $Q_j*M$.
    Because $M$ does not contain any variables $x\in \mathit{PVar}$, by the soundness of the frame rule, it holds that $(w.\mathit{pre}*M\statesep Q_j*M)\soundw t$.
    Thus, $w'\soundw t_j$ is guaranteed to still hold after the shared abduction step.
\end{proof}

Proof of \Cref{thm:worlds}:
\begin{proof}
    By assumption, the analysis with $\ana_{B,\abs}$  already produces only sound contracts for branching-free functions.
    This property remains even with worlds, as worlds do not differ from \absstate s in the case of branching-free functions.
    Furthermore, for branching programs, the computed worlds will have correct pre- and post-conditions if no shared abduction happens as the world is then again the same as a set of \absstate s with the same pre-condition.
    As a result, the analysis can only produce an unsound contract if the shared abduction step breaks soundness.
    However, by \Cref{thm:soundlearn}, shared abduction preserves soundness for traces $t\in T$ unrelated to the most recent symbolic execution step.
    Furthermore, shared abduction also guarantees that the symbolic execution step also preserves soundness for any equivalent extended trace $t\in T_\mathit{new}$.
    As a result, any contract computed by $\ana_{B,\abs}'$ is by construction sound.
\end{proof}

\subsection{Shape Extrapolation}\label{app:sub_proof_shape_extrapol}

Proof of \Cref{thm:loops}:
\begin{proof}
    We take an arbitrary but fixed loop $l$, its corresponding loop function $f_l$, and apply \Cref{alg:extrapol} to it.
    If the procedure does not fail, we call the resulting contract $c$ with pre-condition $P$ and post-condition $Q$.
    Because the procedure did not fail, it has computed the two extrapolated shapes $\oldsh$ and $\newsh$.
    Based on these both the invariant \absstate\:$s_\mathit{inv}$ and the final \absstate, that has been used to derive $c$, can be constructed.

    We show that $c$ is sound for $f_l$ by arguing that it can be verified via Hoare style.
    This suffices as Hoare style semantics coincide with the soundness of contracts as defined in \Cref{def:soundst}.
    For a Hoare style verification of $f_l$ the following three steps suffice:
    \begin{enumerate}
        \item Show that the pre-condition $P$ (which is equivalent to $s_\mathit{final}.\mathit{pre}$) entails the invariant $s_\mathit{inv}.\mathit{curr}$
        \item Show that $s_\mathit{inv}.\mathit{curr}$ is actually an invariant for $l$
        \item Show that under the assumption $\bigvee_i\neg e_i$, i.e. if the program leaves the loop, the invariant $s_\mathit{inv}.\mathit{curr}$ entails the final \absstate\:$s_\mathit{final}.\mathit{curr}$ which is equivalent to $Q$
    \end{enumerate}

    \paragraph{Proof of $(1)$.}
    First, we show $P\vdash s_\mathit{inv}.\mathit{curr}$.
    We assume a fixed, but arbitrary configuration \textit{conf} such that $\mathit{conf}\vDash P\ (1)$.
    In $P$ all program variables implicitly have their initial values, i.e. $x=X$ for $x\in \mathit{PVar}$.
    Thus, from $(1)$ it follows that $\mathit{conf}.S(x)=\mathit{conf}.S(X)$.
    Due to condition $(3)$ of shape extrapolation, this further means that the full shape $p(\overline{X},\textsc{Exit}(\overline{x}))$ that occurs in $P$ is equivalent to $q(\overline{X},\overline{x})*p(\overline{x},\textsc{Exit}(\overline{x}))$ as the shape represented by $q$ is empty in \textit{conf}.

    In general, it is not guaranteed that $\noshapeeffect_\mathit{pre}\vdash \noshapeeffect_\mathit{curr}$.
    However, since $\noshapeeffect_\mathit{curr}$ can only differ from $\noshapeeffect_\mathit{pre}$ in changed memory locations that are not part of the shape or changed pure variables, we can reduce the entailment to an entailment of these changed parts.
    Due to $s_\mathit{inv}$ being a loop invariant, we further know that the values of the changed memory location in $\noshapeeffect_\mathit{curr}$ has been abstracted.
    Similarly, as described in \Cref{subsec:shape-extrapol-more}, the changed
    pure variables have also been abstracted.
    As  a result, the entailment has to hold.
    As $\noshapeeffect_\mathit{pre}$ is satisfied by \textit{conf}, $\noshapeeffect_\mathit{curr}$, which is part of $s_\mathit{inv}.\mathit{curr}$, is also satisfied by \textit{conf}.
    In conclusion, \textit{conf} satisfies all parts of $s_\mathit{inv}.\mathit{curr}$, i.e.,  $\mathit{conf}\vDash s_\mathit{inv}.\mathit{curr}$ holds.
    Thus, $P\vdash s_\mathit{inv}.\mathit{curr}$.

    \paragraph{Proof of $(2)$.}
    The second step is guaranteed to hold by the second iteration and the invariant
    check.

    \paragraph{Proof of $(3)$.}
    Lastly, we show $s_\mathit{inv}.\mathit{curr}*\bigvee_i\neg e_i \vdash s_\mathit{final}.\mathit{curr}$.
    This trivially holds due to De Morgan's laws and the definition of $s_\mathit{final}.\mathit{curr}$.
\end{proof}

\begin{theorem}\label{thm:extrapolstep}
    Let $\ana_{B,\abs}$  be a biabduction-based shape analysis.
    If $\ana_{B,\abs}$  computes sound contracts for all loop-free functions and \Cref{alg:guess} constructs extrapolated shapes $p$ and $q$ for $s_\mathit{inv}$ and the second iteration of $\mathit{body}_l$ via analysis with $\ana_{B,\abs}$  successfully results in a \absstate\:$s_2$, then $\{s_\mathit{inv}'.\mathit{curr}\}\ \mathit{body}_l\ \{s_\mathit{inv}'.\mathit{curr}\}$ holds where $s_\mathit{inv}'$ is the state after applying value abstraction.
\end{theorem}
\begin{proof}
    We assume, that both the shape extrapolation and second iteration have succeeded.
    Then, we know that $\{s_\mathit{inv}.\mathit{curr}\}\ \mathit{body}_l\ \{s_2.\mathit{curr}\}\ (1)$ by our assumption about $\ana_{B,\abs}$ .
    Due to \Cref{thm:conseq} and the definition of value abstraction guaranteeing $s_2.\mathit{curr}\vdash s_\mathit{ind}.\mathit{curr}$, we further get $\{s_\mathit{inv}.\mathit{curr}\}\ \mathit{body}_l\ \{s_\mathit{ind}.\mathit{curr}\}$.

    We take a fixed, but arbitrary configuration $\mathit{conf}$ and trace $t\in \mathit{body}_l$ through the loop body $\mathit{body}_l$, such that $\mathit{conf}\vDash s_\mathit{inv}.\mathit{curr}$.
    Then there must have been a configuration $\mathit{conf}'$ with $\left(\mathit{conf},t\right)\rightsquigarrow^*\mathit{conf}'$ and $\mathit{conf'}\vDash s_\mathit{ind}.\mathit{curr}\ (2)$ due to $(1)$.

    Due to our assumptions of $s_1$ being the only post-state after the first loop iteration, $t$ has to be the only trace through the loop body.
    This also means, that the exact values of variables and memory locations do not matter for the trace.
    To be more precise, any configuration $\mathit{conf}''$ such that $\mathit{conf}''\vDash s_\mathit{ind}.\mathit{curr}$ would also satisfy $s_\mathit{inv}.\mathit{curr}$ up to the abstracted values.
    Because these values do not influence the trace and because any changes to them throughout the loop body still satisfy the abstracted version in $s_\mathit{ind}.\mathit{curr}$, it also follows that $\{s_\mathit{ind}.\mathit{curr}\}\ \mathit{body}_l\ \{s_\mathit{ind}.\mathit{curr}\}$ needs to hold.
\end{proof}

Due to \Cref{thm:loops} and the assumption about $\ana_{B,\abs}$, the following corollary holds trivially:
\begin{corollary}[Sound Analysis for all Functions]\label{thm:funs}
    Let $\ana_{B,\abs}$  be a biabduction-based shape analysis.
    If $\ana_{B,\abs}$  computes sound contracts for all loop-free functions
    and \Cref{alg:extrapol} used this analysis procedure, then by extending $\ana_{B,\abs}$  with shape extrapolation via \Cref{alg:extrapol} for handling loops every found contract for any function is sound.
\end{corollary}

\section{Detailed Data}\label{app:data}

The plots in \Cref{fig:scatter2} show a break down of the 
results for the test cases in which our new techniques are especially applicable.
The plot labeled with ``Loops'' shows only the results for programs containing
loops for which shape extrapolation is applicable,\footnote{This specifically
excludes the \texttt{predator-test-0156-no-include} test case, which has a loop
condition that can currently not be handled by our heuristic.} whereas the plot labeled
with ``Branching'' is obtained from programs with non-loop branching for which
shared abstraction can be applied.

\begin{figure}[t]
    \centering
    \begin{tikzpicture}
      \begin{loglogaxis}[xlabel=Runtime of Brush in sec,ylabel=Runtime of Broom in sec,enlargelimits=false,legend pos=south east]
          \addplot+[
              only marks,
              color=ACMRed,
              mark=star,
              mark size=2.9pt]
          table [x=new, y=old, col sep=comma] {result-loop.csv};
          
          \addplot+[draw=black,mark=] coordinates {(0.3,0.3) (300,300)};
          \addplot+[draw=lightgray,mark=] coordinates {(0.3,0.6) (150,300)};
          \addplot+[draw=lightgray,mark=,dotted,thick] coordinates {(0.3,0.9) (100,300)};
          \legend{Loops,Equal time, $+100\%$, $+200\%$}
      \end{loglogaxis}
  \end{tikzpicture}
  \begin{tikzpicture}
    \begin{loglogaxis}[xlabel=Runtime of Brush in sec,ylabel=Runtime of Broom in sec,enlargelimits=false,legend pos=south east]
        \addplot+[
            only marks,
            color=ACMRed,
            mark=star,
            mark size=2.9pt]
        table [x=new, y=old, col sep=comma] {result-branching.csv};
        
        \addplot+[draw=black,mark=] coordinates {(0.1,0.1) (1,1) (150,150)};
        \addplot+[draw=lightgray,mark=] coordinates {(0.1,0.15) (1,1.5) (150,225)};
        \addplot+[draw=lightgray,mark=] coordinates {(0.1,0.05) (1,0.5) (150,75)};
        \legend{Branching,Equal time, $\pm 50\%$}
    \end{loglogaxis}
  \end{tikzpicture}
    \caption{Runtime of Broom and \newname\:in seconds, but only the test cases
      from above with loops or branching.}
    \Description{The scatter plot from above, but split into just the loop cases
      on the left and just the branching cases on the right. The loop cases show
      significant speed-up for Brush, whereas the branching cases are mostly just
      slightly faster or take a similar time.}
    \label{fig:scatter2}
  \end{figure}

Below is the table \Cref{tab:data} containing the full raw experimental data, 
i.e. the exact timings in seconds (rounded to two digits after the decimal point)
for all ten runs for every file.

\begin{landscape}
    \begin{longtable}[H]{|l|c|c|c|c|c|c|c|c|c|c|}
    \caption{Detailed experimental data}\label{tab:data}\\
    \hline
    Filename &  &  &  &  &  &  &  &  &  & \\
    \hline
    \endfirsthead
\textbf{easy-13-ok} &  &  &  &  &  &  &  &  &  & \\
Broom: & 0.26 & 0.26 & 0.26 & 0.26 & 0.26 & 0.26 & 0.26 & 0.26 & 0.26 & 0.27 \\
Brush: & 0.23 & 0.23 & 0.23 & 0.23 & 0.23 & 0.23 & 0.23 & 0.23 & 0.23 & 0.23 \\\hline
\textbf{easy-05c-err} &  &  &  &  &  &  &  &  &  & \\
Broom: & 0.20 & 0.20 & 0.20 & 0.20 & 0.20 & 0.20 & 0.20 & 0.20 & 0.20 & 0.20 \\
Brush: & 0.13 & 0.13 & 0.13 & 0.13 & 0.12 & 0.13 & 0.12 & 0.13 & 0.13 & 0.13 \\\hline
\textbf{easy-15-err} &  &  &  &  &  &  &  &  &  & \\
Broom: & 0.15 & 0.15 & 0.15 & 0.16 & 0.15 & 0.15 & 0.15 & 0.15 & 0.15 & 0.15 \\
Brush: & 0.12 & 0.12 & 0.13 & 0.13 & 0.12 & 0.12 & 0.12 & 0.12 & 0.12 & 0.12 \\\hline
\textbf{easy-05-err} &  &  &  &  &  &  &  &  &  & \\
Broom: & 0.08 & 0.08 & 0.08 & 0.08 & 0.08 & 0.08 & 0.08 & 0.08 & 0.08 & 0.08 \\
Brush: & 0.07 & 0.07 & 0.07 & 0.07 & 0.07 & 0.07 & 0.07 & 0.07 & 0.07 & 0.07 \\\hline
\textbf{easy-08-err} &  &  &  &  &  &  &  &  &  & \\
Broom: & 0.08 & 0.08 & 0.08 & 0.08 & 0.08 & 0.09 & 0.08 & 0.08 & 0.08 & 0.08 \\
Brush: & 0.08 & 0.08 & 0.08 & 0.08 & 0.08 & 0.08 & 0.08 & 0.08 & 0.08 & 0.08 \\\hline
\textbf{test\_intrusive\_single\_file} &  &  &  &  &  &  &  &  &  & \\
Broom: & 87.17 & 85.18 & 95.21 & 90.99 & 89.96 & 89.69 & 90.25 & 88.62 & 84.60 & 84.60 \\
Brush: & 89.39 & 90.75 & 90.56 & 96.03 & 98.47 & 86.97 & 90.60 & 87.70 & 88.64 & 88.15 \\\hline
\textbf{call-01-ok\_gcc} &  &  &  &  &  &  &  &  &  & \\
Broom: & 0.95 & 0.96 & 0.94 & 0.95 & 0.96 & 0.94 & 0.94 & 0.96 & 0.94 & 0.95 \\
Brush: & 0.57 & 0.54 & 0.56 & 0.54 & 0.54 & 0.53 & 0.54 & 0.53 & 0.53 & 0.54 \\\hline
\textbf{memcpy-01c-ok} &  &  &  &  &  &  &  &  &  & \\
Broom: & 8.07 & 7.72 & 8.24 & 8.40 & 8.13 & 8.01 & 8.03 & 7.61 & 7.42 & 7.51 \\
Brush: & 2.27 & 2.13 & 2.32 & 2.30 & 2.40 & 2.29 & 2.31 & 2.44 & 2.13 & 2.14 \\\hline
\textbf{easy-15-ok} &  &  &  &  &  &  &  &  &  & \\
Broom: & 1.87 & 1.85 & 1.94 & 1.81 & 1.84 & 1.84 & 1.85 & 1.87 & 1.85 & 1.87 \\
Brush: & 0.55 & 0.59 & 0.57 & 0.68 & 0.60 & 0.61 & 0.57 & 0.61 & 0.58 & 0.60 \\\hline
\textbf{easy-12-ok\_gcc} &  &  &  &  &  &  &  &  &  & \\
Broom: & 0.12 & 0.12 & 0.12 & 0.12 & 0.12 & 0.12 & 0.12 & 0.12 & 0.12 & 0.12 \\
Brush: & 0.10 & 0.09 & 0.09 & 0.09 & 0.09 & 0.09 & 0.09 & 0.09 & 0.10 & 0.09 \\\hline
\textbf{memcpy-04-ok} &  &  &  &  &  &  &  &  &  & \\
Broom: & 24.58 & 24.71 & 25.55 & 26.14 & 24.04 & 24.44 & 25.19 & 25.15 & 25.05 & 25.45 \\
Brush: & 18.51 & 19.08 & 19.30 & 17.92 & 19.12 & 18.85 & 19.31 & 18.20 & 19.25 & 19.29 \\\hline
\textbf{memcpy-06-ok} &  &  &  &  &  &  &  &  &  & \\
Broom: & 0.45 & 0.43 & 0.45 & 0.50 & 0.44 & 0.43 & 0.44 & 0.42 & 0.42 & 0.43 \\
Brush: & 0.28 & 0.27 & 0.29 & 0.30 & 0.31 & 0.30 & 0.29 & 0.28 & 0.28 & 0.32 \\\hline
\textbf{easy-11-err} &  &  &  &  &  &  &  &  &  & \\
Broom: & 0.10 & 0.10 & 0.10 & 0.10 & 0.10 & 0.10 & 0.10 & 0.10 & 0.10 & 0.10 \\
Brush: & 0.08 & 0.08 & 0.08 & 0.08 & 0.08 & 0.08 & 0.08 & 0.08 & 0.08 & 0.08 \\\hline
\textbf{space-err} &  &  &  &  &  &  &  &  &  & \\
Broom: & 0.32 & 0.32 & 0.32 & 0.32 & 0.32 & 0.32 & 0.32 & 0.32 & 0.32 & 0.32 \\
Brush: & 0.31 & 0.30 & 0.29 & 0.30 & 0.31 & 0.30 & 0.32 & 0.29 & 0.29 & 0.31 \\\hline
\textbf{easy-08-ok} &  &  &  &  &  &  &  &  &  & \\
Broom: & 0.34 & 0.34 & 0.34 & 0.34 & 0.35 & 0.34 & 0.35 & 0.34 & 0.35 & 0.34 \\
Brush: & 0.26 & 0.26 & 0.26 & 0.26 & 0.26 & 0.26 & 0.27 & 0.26 & 0.26 & 0.26 \\\hline
\textbf{sll-shared-sll-after-alloc} &  &  &  &  &  &  &  &  &  & \\
Broom: & 32.38 & 32.40 & 32.39 & 32.76 & 32.65 & 31.73 & 31.77 & 31.89 & 31.92 & 32.31 \\
Brush: & 1.81 & 1.46 & 1.56 & 1.49 & 1.61 & 1.55 & 1.55 & 1.53 & 1.56 & 1.54 \\\hline
\textbf{memcpy-05-err} &  &  &  &  &  &  &  &  &  & \\
Broom: & 0.73 & 0.74 & 0.73 & 0.82 & 0.71 & 0.73 & 0.71 & 0.72 & 0.73 & 0.71 \\
Brush: & 0.71 & 0.65 & 0.67 & 0.68 & 0.65 & 0.66 & 0.64 & 0.67 & 0.63 & 0.64 \\\hline
\textbf{global-mem-leaks-err} &  &  &  &  &  &  &  &  &  & \\
Broom: & 0.43 & 0.42 & 0.43 & 0.41 & 0.41 & 0.43 & 0.41 & 0.46 & 0.41 & 0.43 \\
Brush: & 0.35 & 0.35 & 0.35 & 0.35 & 0.34 & 0.35 & 0.34 & 0.34 & 0.36 & 0.37 \\\hline
\textbf{easy-07-err} &  &  &  &  &  &  &  &  &  & \\
Broom: & 0.15 & 0.15 & 0.15 & 0.15 & 0.15 & 0.15 & 0.15 & 0.15 & 0.15 & 0.15 \\
Brush: & 0.07 & 0.07 & 0.08 & 0.07 & 0.07 & 0.07 & 0.07 & 0.07 & 0.07 & 0.07 \\\hline
\textbf{easy-01-err} &  &  &  &  &  &  &  &  &  & \\
Broom: & 0.04 & 0.04 & 0.04 & 0.04 & 0.04 & 0.04 & 0.04 & 0.04 & 0.04 & 0.04 \\
Brush: & 0.03 & 0.03 & 0.03 & 0.03 & 0.03 & 0.03 & 0.03 & 0.03 & 0.03 & 0.03 \\\hline
\textbf{sll-shared-sll-after} &  &  &  &  &  &  &  &  &  & \\
Broom: & 46.07 & 44.40 & 46.27 & 52.70 & 45.52 & 45.38 & 45.20 & 44.97 & 45.87 & 44.89 \\
Brush: & 2.49 & 2.58 & 2.43 & 2.73 & 2.49 & 2.48 & 2.49 & 2.47 & 2.51 & 2.47 \\\hline
\textbf{easy-03-ok} &  &  &  &  &  &  &  &  &  & \\
Broom: & 0.08 & 0.08 & 0.08 & 0.08 & 0.08 & 0.08 & 0.08 & 0.08 & 0.08 & 0.08 \\
Brush: & 0.07 & 0.07 & 0.07 & 0.07 & 0.07 & 0.07 & 0.07 & 0.07 & 0.07 & 0.07 \\\hline
\textbf{sll-fst-shared} &  &  &  &  &  &  &  &  &  & \\
Broom: & 44.62 & 44.28 & 44.44 & 49.99 & 48.28 & 44.52 & 44.72 & 48.97 & 44.47 & 44.49 \\
Brush: & 3.96 & 3.93 & 3.95 & 4.19 & 4.00 & 3.96 & 4.04 & 4.00 & 4.05 & 3.96 \\\hline
\textbf{circ\_dll\_simple-err} &  &  &  &  &  &  &  &  &  & \\
Broom: & 3.02 & 3.05 & 3.06 & 3.03 & 3.07 & 3.04 & 3.06 & 3.05 & 3.06 & 3.07 \\
Brush: & 2.28 & 2.26 & 2.26 & 2.27 & 2.27 & 2.28 & 2.29 & 2.28 & 2.26 & 2.28 \\\hline
\textbf{easy-04-err} &  &  &  &  &  &  &  &  &  & \\
Broom: & 0.08 & 0.08 & 0.08 & 0.08 & 0.08 & 0.08 & 0.08 & 0.08 & 0.08 & 0.08 \\
Brush: & 0.07 & 0.07 & 0.07 & 0.07 & 0.07 & 0.07 & 0.07 & 0.07 & 0.07 & 0.07 \\\hline
\textbf{easy-09-err} &  &  &  &  &  &  &  &  &  & \\
Broom: & 0.10 & 0.10 & 0.09 & 0.10 & 0.10 & 0.10 & 0.10 & 0.10 & 0.10 & 0.10 \\
Brush: & 0.08 & 0.08 & 0.08 & 0.08 & 0.08 & 0.08 & 0.08 & 0.08 & 0.08 & 0.08 \\\hline
\textbf{memcpy-07-ok} &  &  &  &  &  &  &  &  &  & \\
Broom: & 0.64 & 0.65 & 0.65 & 0.72 & 0.69 & 0.68 & 0.67 & 0.63 & 0.63 & 0.63 \\
Brush: & 0.56 & 0.56 & 0.62 & 0.56 & 0.59 & 0.59 & 0.67 & 0.56 & 0.54 & 0.58 \\\hline
\textbf{no-field-ok} &  &  &  &  &  &  &  &  &  & \\
Broom: & 1.45 & 1.45 & 1.42 & 1.44 & 1.45 & 1.45 & 1.44 & 1.44 & 1.44 & 1.44 \\
Brush: & 1.22 & 1.25 & 1.24 & 1.24 & 1.25 & 1.23 & 1.24 & 1.24 & 1.23 & 1.23 \\\hline
\textbf{easy-08b-ok} &  &  &  &  &  &  &  &  &  & \\
Broom: & 0.34 & 0.35 & 0.35 & 0.34 & 0.34 & 0.34 & 0.34 & 0.34 & 0.34 & 0.35 \\
Brush: & 0.26 & 0.26 & 0.26 & 0.26 & 0.27 & 0.26 & 0.26 & 0.26 & 0.26 & 0.26 \\\hline
\textbf{sll-fst-shared-alloc} &  &  &  &  &  &  &  &  &  & \\
Broom: & 36.70 & 36.56 & 36.60 & 36.63 & 36.65 & 36.67 & 36.71 & 36.64 & 36.63 & 36.64 \\
Brush: & 1.08 & 1.06 & 1.07 & 1.07 & 1.07 & 1.09 & 1.08 & 1.08 & 1.08 & 1.07 \\\hline
\textbf{sll-alloc} &  &  &  &  &  &  &  &  &  & \\
Broom: & 2.75 & 2.60 & 2.67 & 2.68 & 2.73 & 2.72 & 2.67 & 2.69 & 2.61 & 2.93 \\
Brush: & 0.61 & 0.57 & 0.57 & 0.63 & 0.59 & 0.58 & 0.60 & 0.64 & 0.57 & 0.61 \\\hline
\textbf{sll-fst-shared-iter} &  &  &  &  &  &  &  &  &  & \\
Broom: & 18.76 & 18.57 & 18.72 & 18.95 & 18.87 & 18.81 & 18.89 & 18.76 & 18.60 & 18.81 \\
Brush: & 2.62 & 2.55 & 2.59 & 2.62 & 2.52 & 2.53 & 2.58 & 2.58 & 2.56 & 2.52 \\\hline
\textbf{copy\_alloc} &  &  &  &  &  &  &  &  &  & \\
Broom: & 5.32 & 5.22 & 5.22 & 5.39 & 5.40 & 5.13 & 5.22 & 5.15 & 5.20 & 5.17 \\
Brush: & 0.83 & 0.86 & 0.88 & 0.88 & 0.91 & 0.86 & 0.86 & 0.85 & 0.86 & 0.86 \\\hline
\textbf{sll} &  &  &  &  &  &  &  &  &  & \\
Broom: & 35.25 & 34.92 & 35.40 & 35.01 & 34.65 & 34.96 & 34.59 & 34.65 & 34.59 & 34.82 \\
Brush: & 1.93 & 2.12 & 1.95 & 2.11 & 1.96 & 1.92 & 1.97 & 1.98 & 1.95 & 1.93 \\\hline
\textbf{easy-10-err} &  &  &  &  &  &  &  &  &  & \\
Broom: & 0.24 & 0.24 & 0.24 & 0.24 & 0.24 & 0.24 & 0.24 & 0.24 & 0.24 & 0.24 \\
Brush: & 0.19 & 0.19 & 0.19 & 0.19 & 0.19 & 0.19 & 0.19 & 0.19 & 0.19 & 0.19 \\\hline
\textbf{reversal} &  &  &  &  &  &  &  &  &  & \\
Broom: & 11.80 & 11.70 & 11.66 & 11.85 & 11.78 & 11.79 & 11.69 & 11.64 & 11.75 & 11.67 \\
Brush: & 2.03 & 2.01 & 2.01 & 2.02 & 2.01 & 2.06 & 2.06 & 2.09 & 2.08 & 2.04 \\\hline
\textbf{test-junk-ok} &  &  &  &  &  &  &  &  &  & \\
Broom: & 1.99 & 1.95 & 1.99 & 1.98 & 1.97 & 1.98 & 1.98 & 1.98 & 2.00 & 1.98 \\
Brush: & 1.01 & 1.01 & 1.02 & 1.00 & 1.02 & 1.04 & 1.02 & 1.03 & 1.01 & 1.02 \\\hline
\textbf{easy-02-err} &  &  &  &  &  &  &  &  &  & \\
Broom: & 0.10 & 0.10 & 0.10 & 0.10 & 0.10 & 0.10 & 0.10 & 0.10 & 0.10 & 0.10 \\
Brush: & 0.03 & 0.03 & 0.03 & 0.03 & 0.03 & 0.03 & 0.03 & 0.03 & 0.03 & 0.03 \\\hline
\textbf{easy-01-ok} &  &  &  &  &  &  &  &  &  & \\
Broom: & 0.09 & 0.09 & 0.09 & 0.09 & 0.09 & 0.10 & 0.09 & 0.09 & 0.09 & 0.09 \\
Brush: & 0.07 & 0.07 & 0.07 & 0.07 & 0.07 & 0.07 & 0.07 & 0.07 & 0.07 & 0.07 \\\hline
\textbf{memcpy-07-err} &  &  &  &  &  &  &  &  &  & \\
Broom: & 1.05 & 0.99 & 0.99 & 0.98 & 0.99 & 0.99 & 0.99 & 0.99 & 0.99 & 1.16 \\
Brush: & 1.44 & 1.47 & 1.43 & 1.41 & 1.41 & 1.44 & 1.43 & 1.39 & 1.42 & 1.45 \\\hline
\textbf{easy-05b-err} &  &  &  &  &  &  &  &  &  & \\
Broom: & 0.54 & 0.54 & 0.54 & 0.54 & 0.54 & 0.54 & 0.54 & 0.54 & 0.54 & 0.54 \\
Brush: & 0.22 & 0.22 & 0.22 & 0.22 & 0.22 & 0.22 & 0.22 & 0.23 & 0.22 & 0.22 \\\hline
\textbf{global-rerun-ok} &  &  &  &  &  &  &  &  &  & \\
Broom: & 0.40 & 0.41 & 0.40 & 0.40 & 0.40 & 0.40 & 0.40 & 0.40 & 0.40 & 0.40 \\
Brush: & 0.15 & 0.15 & 0.15 & 0.15 & 0.15 & 0.15 & 0.15 & 0.16 & 0.15 & 0.15 \\\hline
\textbf{intrusive-list} &  &  &  &  &  &  &  &  &  & \\
Broom: & 60.98 & 60.32 & 60.12 & 60.79 & 60.33 & 62.06 & 60.38 & 60.43 & 60.23 & 64.90 \\
Brush: & 64.72 & 67.74 & 68.57 & 67.98 & 67.58 & 68.21 & 68.13 & 69.50 & 68.35 & 68.28 \\\hline
\textbf{return-struct-ok} &  &  &  &  &  &  &  &  &  & \\
Broom: & 0.85 & 0.85 & 0.85 & 0.83 & 0.83 & 0.83 & 0.86 & 0.86 & 0.87 & 0.86 \\
Brush: & 0.67 & 0.67 & 0.67 & 0.68 & 0.67 & 0.68 & 0.68 & 0.72 & 0.69 & 0.68 \\\hline
\textbf{easy-06-err} &  &  &  &  &  &  &  &  &  & \\
Broom: & 0.11 & 0.11 & 0.11 & 0.11 & 0.10 & 0.10 & 0.11 & 0.10 & 0.11 & 0.11 \\
Brush: & 0.08 & 0.08 & 0.08 & 0.08 & 0.08 & 0.08 & 0.08 & 0.08 & 0.08 & 0.08 \\\hline
\textbf{linux-list-t2} &  &  &  &  &  &  &  &  &  & \\
Broom: & 27.73 & 29.64 & 27.50 & 27.99 & 27.50 & 27.51 & 27.65 & 27.44 & 27.63 & 27.52 \\
Brush: & 15.24 & 14.35 & 14.30 & 14.32 & 14.52 & 14.42 & 14.38 & 14.42 & 14.37 & 14.39 \\\hline
\textbf{linux-list} &  &  &  &  &  &  &  &  &  & \\
Broom: & 92.64 & 92.44 & 92.58 & 92.84 & 91.85 & 92.53 & 92.48 & 97.33 & 93.66 & 102.85 \\
Brush: & 96.46 & 94.13 & 94.83 & 94.58 & 93.96 & 94.83 & 93.37 & 94.96 & 96.17 & 106.62 \\\hline
\textbf{easy-01b-ok} &  &  &  &  &  &  &  &  &  & \\
Broom: & 0.37 & 0.37 & 0.37 & 0.37 & 0.37 & 0.37 & 0.36 & 0.37 & 0.37 & 0.37 \\
Brush: & 0.15 & 0.16 & 0.15 & 0.15 & 0.15 & 0.16 & 0.16 & 0.15 & 0.15 & 0.15 \\\hline
\textbf{memcpy-03-ok} &  &  &  &  &  &  &  &  &  & \\
Broom: & 10.42 & 11.28 & 10.37 & 9.76 & 9.59 & 9.68 & 9.61 & 9.60 & 9.61 & 9.62 \\
Brush: & 7.70 & 7.62 & 7.93 & 7.10 & 6.87 & 7.06 & 7.00 & 6.96 & 6.92 & 6.83 \\\hline
\textbf{memcpy-01-ok} &  &  &  &  &  &  &  &  &  & \\
Broom: & 2.98 & 2.88 & 2.99 & 2.86 & 2.99 & 2.89 & 2.95 & 2.90 & 2.90 & 2.90 \\
Brush: & 1.95 & 2.02 & 2.02 & 1.95 & 1.95 & 2.01 & 2.00 & 2.01 & 1.95 & 1.92 \\\hline
\textbf{memcpy-01b-ok} &  &  &  &  &  &  &  &  &  & \\
Broom: & 1.53 & 1.61 & 1.61 & 1.78 & 1.71 & 1.72 & 1.74 & 1.66 & 1.67 & 1.65 \\
Brush: & 1.02 & 1.03 & 1.03 & 1.24 & 1.12 & 1.08 & 1.13 & 1.12 & 1.10 & 1.03 \\\hline
\textbf{easy-10b-err} &  &  &  &  &  &  &  &  &  & \\
Broom: & 0.23 & 0.23 & 0.23 & 0.22 & 0.22 & 0.23 & 0.22 & 0.23 & 0.23 & 0.22 \\
Brush: & 0.18 & 0.17 & 0.18 & 0.18 & 0.17 & 0.17 & 0.18 & 0.18 & 0.18 & 0.18 \\\hline
\textbf{circ\_dll\_simple} &  &  &  &  &  &  &  &  &  & \\
Broom: & 2.93 & 2.94 & 2.92 & 2.94 & 2.90 & 2.94 & 2.88 & 2.92 & 2.94 & 2.90 \\
Brush: & 2.09 & 2.08 & 2.11 & 2.11 & 2.09 & 2.07 & 2.10 & 2.08 & 2.07 & 2.02 \\\hline
\textbf{easy-16-err\_gcc} &  &  &  &  &  &  &  &  &  & \\
Broom: & 0.27 & 0.27 & 0.28 & 0.27 & 0.27 & 0.26 & 0.26 & 0.27 & 0.26 & 0.26 \\
Brush: & 0.23 & 0.22 & 0.22 & 0.22 & 0.22 & 0.22 & 0.22 & 0.22 & 0.22 & 0.22 \\\hline
\textbf{easy-16b-err\_gcc} &  &  &  &  &  &  &  &  &  & \\
Broom: & 1.06 & 1.06 & 1.03 & 1.06 & 1.06 & 1.03 & 1.03 & 1.05 & 1.06 & 1.09 \\
Brush: & 0.92 & 0.93 & 0.91 & 0.89 & 0.89 & 0.90 & 0.91 & 0.89 & 0.91 & 0.92 \\\hline
\textbf{memcpy-02-err} &  &  &  &  &  &  &  &  &  & \\
Broom: & 0.93 & 0.93 & 0.94 & 0.93 & 0.91 & 0.92 & 0.94 & 0.94 & 0.94 & 0.93 \\
Brush: & 0.80 & 0.85 & 0.83 & 0.83 & 0.85 & 0.80 & 0.81 & 0.84 & 0.83 & 0.84 \\\hline
\textbf{global\_var\_move} &  &  &  &  &  &  &  &  &  & \\
Broom: & 0.37 & 0.38 & 0.38 & 0.38 & 0.38 & 0.38 & 0.38 & 0.38 & 0.39 & 0.38 \\
Brush: & 0.29 & 0.29 & 0.30 & 0.30 & 0.30 & 0.29 & 0.29 & 0.29 & 0.29 & 0.29 \\\hline
\textbf{easy-04b-err} &  &  &  &  &  &  &  &  &  & \\
Broom: & 0.29 & 0.29 & 0.29 & 0.29 & 0.29 & 0.29 & 0.29 & 0.29 & 0.29 & 0.29 \\
Brush: & 0.22 & 0.22 & 0.22 & 0.22 & 0.22 & 0.22 & 0.22 & 0.22 & 0.22 & 0.22 \\\hline
\textbf{easy-01b-err} &  &  &  &  &  &  &  &  &  & \\
Broom: & 0.16 & 0.16 & 0.16 & 0.16 & 0.16 & 0.16 & 0.17 & 0.17 & 0.16 & 0.16 \\
Brush: & 0.14 & 0.14 & 0.14 & 0.14 & 0.14 & 0.14 & 0.14 & 0.14 & 0.14 & 0.14 \\\hline
\textbf{memcpy-03b-ok} &  &  &  &  &  &  &  &  &  & \\
Broom: & 18.72 & 19.09 & 19.84 & 18.99 & 18.75 & 18.78 & 19.83 & 19.88 & 20.56 & 20.82 \\
Brush: & 7.54 & 7.59 & 7.67 & 7.68 & 7.69 & 7.62 & 7.49 & 7.75 & 7.75 & 8.05 \\\hline
\textbf{predator-test-0156-no-include} &  &  &  &  &  &  &  &  &  & \\
Broom: & 58.11 & 59.16 & 57.98 & 58.07 & 57.89 & 58.56 & 59.03 & 59.46 & 58.94 & 59.06 \\
Brush: & 78.98 & 96.02 & 61.78 & 76.48 & 74.36 & 61.85 & 64.84 & 71.67 & 65.46 & 73.05 \\\hline
\textbf{dll-as-sll-traverse} &  &  &  &  &  &  &  &  &  & \\
Broom: & 2.03 & 2.11 & 2.10 & 2.05 & 2.15 & 2.15 & 2.09 & 2.14 & 2.25 & 2.11 \\
Brush: & 0.39 & 0.41 & 0.38 & 0.38 & 0.41 & 0.39 & 0.39 & 0.40 & 0.42 & 0.39 \\\hline
\textbf{memcpy-04b-ok} &  &  &  &  &  &  &  &  &  & \\
Broom: & 47.19 & 48.48 & 47.38 & 46.77 & 46.60 & 50.31 & 51.59 & 52.33 & 48.82 & 50.86 \\
Brush: & 19.03 & 20.19 & 20.34 & 19.32 & 19.14 & 19.60 & 20.79 & 20.76 & 19.92 & 19.37 \\\hline
\textbf{nested\_lists2} &  &  &  &  &  &  &  &  &  & \\
Broom: & 168.53 & 188.72 & 189.03 & 189.91 & 188.95 & 188.94 & 187.73 & 187.89 & 163.84 & 187.04 \\
Brush: & 2.70 & 2.73 & 2.69 & 2.78 & 2.76 & 2.74 & 2.75 & 2.74 & 2.71 & 2.70 \\\hline
\textbf{circ\_dll\_embeded\_int} &  &  &  &  &  &  &  &  &  & \\
Broom: & 4.87 & 4.89 & 4.92 & 4.82 & 4.87 & 4.83 & 4.81 & 4.80 & 4.88 & 4.83 \\
Brush: & 5.12 & 5.13 & 5.12 & 5.15 & 5.09 & 5.06 & 5.09 & 5.08 & 5.06 & 5.04 \\\hline
\textbf{easy-14-ok} &  &  &  &  &  &  &  &  &  & \\
Broom: & 0.48 & 0.48 & 0.48 & 0.48 & 0.48 & 0.48 & 0.49 & 0.49 & 0.48 & 0.49 \\
Brush: & 0.39 & 0.39 & 0.39 & 0.39 & 0.38 & 0.39 & 0.39 & 0.39 & 0.39 & 0.39 \\\hline
\textbf{intrusive-list-minimal-example} &  &  &  &  &  &  &  &  &  & \\
Broom: & 38.19 & 41.89 & 38.25 & 38.51 & 38.35 & 38.16 & 38.75 & 38.73 & 38.15 & 40.46 \\
Brush: & 36.05 & 36.23 & 35.96 & 35.88 & 36.02 & 36.62 & 36.30 & 36.41 & 36.08 & 36.11 \\\hline
\textbf{linux-list-t2-err} &  &  &  &  &  &  &  &  &  & \\
Broom: & 19.90 & 18.49 & 18.38 & 18.54 & 20.07 & 18.91 & 18.79 & 18.81 & 18.74 & 18.81 \\
Brush: & 12.63 & 12.91 & 12.57 & 12.55 & 12.87 & 12.77 & 12.78 & 12.75 & 12.82 & 12.79 \\\hline
\end{longtable}

\end{landscape}

\section{Further Examples}\label{app:examples}

\subsection{About unsound Loop Acceleration}\label{subsec:forwardcomp}

First, we want to briefly explore how our biabductive loop acceleration approach
improves on existing loop acceleration approaches in the case of unsound 
abstract results.
We recall that in \citep{biabd_conference,biabd} loops were accelerated
by lifting a direct application of abstraction as found in standard forward 
analyses to the biabductive setting.

\begin{lstlisting}[language=C,label=lst:unsound,caption={Two-step list traversal}]
  void two_steps(node *x) {
    while (x != NULL) {
      x = x->next;
      x = x->next;
    } 
  }
\end{lstlisting}

As described for \Cref{lst:mot2}, this loop acceleration technique
consists of four steps:
$(1)$ The loop body is analyzed for the first time.
$(2)$ When the analysis reaches the loop header again, the analysis transforms
the analysis state $(P,Q)$ via abstraction to $(\abs(P),\abs(Q))$, thus 
abstracting the already traversed part of the data structure in the pre- and
post-condition, respectively.
$(3)$ If the state has been computed before, the analysis has found a fixed
point and continues after the loop.
Otherwise, continue with Step $(1)$.

\begin{example}
    For \Cref{lst:unsound}, Step $(1)$ leads to abducing the partial formula
    $X.\texttt{next}\mapsto\ell_1*\ell_1.\texttt{next}\mapsto\ell_2$.
    The state-of-the-art abstraction procedures of Broom or Abductor
    combine the two points-to predicates into one list segment
    $\ls(X,\ell_2)$ in Step $(2)$.
    This list segment specifically represents a list of arbitrary length.
    As the analysis has found a new state, Step $(3)$ will lead to a second
    iteration with Step $(1)$ resulting in $\ls(X,\ell_2)*\ell_2.\texttt{next}
    \mapsto\ell_3*\ell_3.\texttt{next} \mapsto\ell_4$, which is abstracted into
    $\ls(X,\ell_4)$ in another Step $(2)$.
    The resulting list segment is equivalent to the old $\ls(X,\ell_2)$
    up to renaming the logical variables.
    Thus, Step $(3)$ finishes the analysis of the loop with the fixed point
    state $\left(x=X\formsep \ls(X,\Null)\statesep x=\Null\formsep \ls(X,\Null)\right)$.
    The list in the pre-condition has been over-approximated too much and
    unsafely includes lists of odd length and will thus be filtered out in the 
    second phase.
    
    Similarly, biabductive loop acceleration with shape extrapolation as described 
    above will not be able to find a sound contract, since our heuristic is
    based on the same abstraction.
    However, it does improve the status quo by failing early, i.e., by 
    constructing a candidate invariant with the list segment of arbitrary length
    which is then immediately found to be unsound.
    Moreover, any extension to the underlying logic or abstraction procedure that
    allows to find a sound over-approximation would still be faster in our new 
    framework.
\end{example}

\subsection{List Reversal}\label{app:examples:listrev}
\begin{center}
    \begin{lstlisting}[language=C,label=lst:app:example2,caption={An in-place list reversal algorithm},basicstyle=\ttfamily\small]
    node *reverse_list(node *x) {
        node *r = NULL;

        while (x != NULL) {
            node *next_r = x;
            x = x->next;
            next_r->next = r;
            r=next_r;
        }

        return r;
    }
\end{lstlisting}
\end{center}
This example shows a simple in-place list reversal for singly-linked lists as it is frequently studied in the literature (e.g., \citep{otherabstraction}).
The state after one iteration of the loop is akin to the following:

\[
    s_1\mathdef\ \left(X.\mathit{next}\mapsto\ell_1\statesep x=\ell_1\land r=X\formsep  X.\mathit{next}\mapsto R\right)
\]

From this, it is apparent that the set of anchors $\{X,R\}$ is not distinct from the image of \textit{TransfM}, due to $\mathit{TransfM}(r)=L$.
Thus, as described in \Cref{subsec:shape-extrapol-more}, the analysis needs to do another loop iteration before we can extrapolate the shape.
After the second iteration, the state would be the following:
\[
    s_1'\mathdef\ \left(X.\mathit{next}\mapsto\ell_1*\ell_1.\mathit{next}\mapsto\ell_2\statesep x=\ell_2\land r=\ell_1\formsep X.\mathit{next}\mapsto R*\ell_1.\mathit{next}\mapsto X\right)
\]

Now, $\mathit{TransfM}(r)=\ell_1$ and $\mathit{TransfM}(x)=\ell_2$, thus having no overlap with the anchors anymore.
As a result, the state can be partitioned such that $\mathit{Changed}\mathdef\{x,r\}$.
As a result, the extrapolated shapes are $p(X,R,x,r)\mathdef \ls(X,x)$ and $q(X,R,x,r)\mathdef \ls(r,R)$ and the following candidate invariant state is computed:
\[
    s_2\mathdef\left(\ls(X,\ell_2)*\ls(\ell_2,\Null)\statesep x=\ell_2\land r=\ell_1\formsep \ls(\ell_1,R)*\ls(\ell_2,\Null)\right)
\]

With the additional assumption of $\ell_2\neq\Null$, the second iteration succeeds in the following state:
\begin{align*}
    s_3.\mathit{pre}\mathdef\   & \ls(X,\ell_2)*\ell_2.\mathit{next}\mapsto\ell_3*\ls(\ell_3,\Null)                                \\
    s_3.\mathit{curr}\mathdef\  & x=\ell_3\land r=\ell_2\formsep \ls(\ell_1,R)*\ell_2.\mathit{next}\mapsto\ell_1*\ls(\ell_3,\Null)
\end{align*}

Because $s_3\vdash s_2$, we find that $s_2$ is an actual loop invariant and can finally compute the loop contract:
\[
    s_\mathit{final}\mathdef\left(\ls(X,\Null)\statesep x=\Null\land r=\ell_4\formsep \ls(\ell_4,R)\right)
\]

\subsection{De-/Allocation}
In the following, we explore how our shape extrapolation algorithm handles de-allocation and allocation at the same time.

\begin{center}
    \begin{lstlisting}[language=C,label=lst:app:example3,caption={The function moves and reverses a list by allocating new nodes and de-allocating old ones.},basicstyle=\ttfamily\small]
    node *move_to_rev(node *old) {
        node new = NULL;

        while (old != NULL) {
            node *to_free = old;
            node *new_node = alloc(next,data);

            new_node->data = old->data;
            new_node->next = new;

            new = new_node;

            old = old->next;
            free(to_free);
        }

        return new;
    }\end{lstlisting}
\end{center}

This example shows an artificial combined use-case of de-allocation and allocation.
A given input list is traversed, a copy is allocated for each node and appended to the output list and finally the input node is freed.
In addition, the order of the list is reversed.
The state after the first iteration is the following:
\begin{align*}
    s_1.\mathit{pre}\mathdef  & \ \mathit{OLD}.\mathit{next}\mapsto\ell_1*\mathit{OLD}.\mathit{data}\mapsto\ell_2                                                \\
    s_1.\mathit{curr}\mathdef & \ \mathit{old}=\ell_1\land \mathit{new}=\ell_4\formsep \ell_4.\mathit{next}\mapsto\mathit{NEW}*\ell_4.\mathit{data}\mapsto\ell_2
\end{align*}

From this state, the following shape is extrapolated as $p(\mathit{OLD}, \mathit{NEW}, \mathit{old}, \mathit{new}) \mathdef \ls(\mathit{OLD}, \mathit{old})$ and $q(\mathit{OLD}, \mathit{NEW}, \mathit{old}, \mathit{new}) \mathdef \ls(\mathit{new},\Null)$ with $\noshapeeffect_\mathit{pre}=\noshapeeffect_\mathit{curr}$ being empty.
Together, they form the invariant state:
\[
    s_\mathit{inv}\mathdef  \left(\ls(\mathit{OLD},\ell_1)*\ls(\ell_1,\Null)\statesep \mathit{old}=\ell_1\land \mathit{new}=\ell_4\formsep \ls(\ell_4,\mathit{NEW})*\ls(\ell_1,\Null)\right)
\]

Ultimately, the final state is:
\[
    s_\mathit{final}\mathdef \left(\ls(\mathit{OLD},\Null)\statesep \mathit{old}=\Null\land \mathit{new}=\ell_5\formsep \ls(\ell_5,\mathit{NEW})\right)
\]

\subsection{Invariant Value Abstraction}\label{app:examples:inv}
As described in \Cref{subsec:shape-extrapol-more}, the constructed candidate
invariant might describe shapes that guarantee memory safety, but might still
not be a loop invariant due to value mismatches.
The following example shows a problem where this problem occurs:

\begin{center}
    \begin{lstlisting}[language=C,label=lst:app:example4,caption={A function
        that requires non-trivial value invariants},basicstyle=\ttfamily\small]
    bool traverse_with_flag(node *list, bool flag){
        while(list) {
            list=list->next;
            flag=true;
        }

        return flag;
    }\end{lstlisting}
\end{center}

There, the value of \texttt{flag} before the loop (which is unknown) does
not entail the one in the invariant, since that value is obtained from the state after one iteration (i.e., \texttt{true}).
As we do not want such value incompatibilities to lead to disjunctive post-conditions of the contract, we instead simply abstract the value of \texttt{flag} in $s_\mathit{inv}$ to $\top$ (which, in this case, is equivalent to doing a join of the two values).
\end{document}